\documentclass[12pt]{article}
\usepackage{amsmath}
\usepackage{amssymb}
\usepackage[a4paper]{geometry}
\usepackage{amsthm}
\usepackage{subcaption}
\usepackage{cite}
\usepackage{graphicx}
\usepackage{tikz}
\usepackage{tikz-cd}
\usepackage{mathrsfs}
\usepackage{bm}
\usepackage{slashed}
\usepackage[hidelinks]{hyperref}
\usepackage{listings}
\usepackage{orcidlink}

\hypersetup{
    colorlinks = true,
    linkcolor = {black},
    citecolor = {blue},
    anchorcolor = {green},
    citecolor = {red},
    filecolor = {cyan},
    menucolor = {red},
    runcolor = {cyan},
    urlcolor = {magenta}
}
\newcommand{\twidth}{6in}
\textwidth=\twidth
\newcommand{\EE}{{\mathcal{E}}}

\newcommand{\SU}{{\mathrm{SU}}}
\newcommand{\SO}{{\mathrm{SO}}}
\newcommand{\U}{{\mathrm{U}}}
\newcommand{\ii}{{\mathrm{i}}}
\newcommand{\e}{{\mathrm{e}}}

\newcommand{\R}{{\mathbb{R}}}

\newcommand{\Z}{{\mathbb{Z}}}

\newcommand{\C}{{\mathbb{C}}}
\renewcommand{\H}{{\mathbb{H}}}

\newcommand{\bu}{{\bm{1}}}

\newcommand{\beq}{\begin{equation}}
\newcommand{\eeq}{\end{equation}}
\newcommand{\bea}{\begin{eqnarray}}
\newcommand{\eea}{\end{eqnarray}}
\newcommand{\bal}{\begin{align}}
\newcommand{\eal}{\end{align}}
\newcommand{\bml}{\begin{multline}}
\newcommand{\eml}{\end{multline}}

\newcommand{\wh}{\widehat}
\newcommand{\lto}{\longrightarrow}

\def \d{\mathrm{d}}
\newcommand{\ol}{\overline}

\newcommand{\tr}{{\operatorname{tr}}}

\newcommand{\sn}{{\rm sn}\, }
\newcommand{\cn}{{\rm cn}\, }
\newcommand{\dn}{{\rm dn}\, }
\newcommand{\scn}{{\rm sc}\, }
\newcommand{\nc}{{\rm nc}\, }

\newcommand{\dc}{{\rm dc}\, }
\newcommand{\nd}{{\rm nd}\, }
\newcommand{\id}{{\rm Id}}

\newcommand{\diag}{{\rm diag}}

\makeatletter
\newcommand\xleftrightarrow[2][]{%
  \ext@arrow 9999{\longleftrightarrowfill@}{#1}{#2}}
\newcommand\longleftrightarrowfill@{%
  \arrowfill@\leftarrow\relbar\rightarrow}
\makeatother

\newcommand{\bigslant}[2]{{\raisebox{.2em}{$#1$}\left/\raisebox{-.2em}{$#2$}\right.}}

\newtheorem{theorem}{Theorem}

\newtheorem{lemma}[theorem]{Lemma}

\begin{document}
\renewcommand*{\thefootnote}{\fnsymbol{footnote}}
\begin{titlepage}
\begin{center}
{\Large
A Nahm transform for rotating calorons \par}


\vspace{10mm}
{\Large Josh Cork\footnote{Email address: \texttt{josh.cork@leicester.ac.uk}}$^{\rm 1}$\,\orcidlink{0000-0002-9006-0108} and Derek Harland\footnote{Email address: \texttt{d.g.harland@leeds.ac.uk}}$^{\rm 2}$\,\orcidlink{0000-0002-4110-9673}}\\[10mm]

\noindent {\em ${}^{\rm 1}$ School of Computing and Mathematical Sciences\\
University of Leicester, University Road, Leicester, United Kingdom
}\\
\smallskip
\noindent {\em ${}^{\rm 2}$ School of Mathematics, University of Leeds\\ Woodhouse Lane, Leeds, United Kingdom }\\[10mm]
{\Large \today}
\vspace{15mm}
\begin{abstract}
Rotating calorons were introduced in the context of rotating quark-gluon plasmas.  They are anti-self-dual gauge fields on $\R^4$ that are invariant under a glide rotation.  We formulate a Nahm transform which identifies rotating calorons with solutions of a delayed-differential equation.  Using this transform, we prove existence of an eight-parameter family of charge 1 rotating calorons with nontrivial holonomy and rotational angle $\pi$, which we construct and visualise using a numerical implementation of the Nahm transform.
\end{abstract}
\end{center}
\end{titlepage}
\renewcommand*{\thefootnote}{\arabic{footnote}}
\setcounter{footnote}{0}
\hypersetup{
    linkcolor = {blue}
}
\numberwithin{equation}{section}
\section{Introduction}
The production of rotating quark-gluon plasmas in noncentral heavy ion collisions \cite{STARcollab2017global} has motivated several recent theoretical studies of QCD at non-zero angular momentum.  The influence of angular velocity on the QCD phase transition has been investigated using lattice QCD \cite{rotwisted1,lattice1,lattice2}, the Nambu--Jona--Lasinio model \cite{NJL}, holography \cite{holography}, and several other models and approaches \cite{hadronresonancegas,rotwisted2,rotwisted3,quarkmeson,yukawa,bagmodel,linearsigma1,linearsigma2,strongcouplingexpansion}.  Different theoretical models make different predictions for the effect of rotation on the QCD phase transition, and currently there is no clear consensus as to whether the transition temperature increases or decreases with increasing angular momentum.

Calorons \cite{calorons1} and their constituent monopoles (also known as dyons) \cite{calorons2,calorons3} play an important role in the QCD phase transition.  So the study of calorons for rotating QCD is well-motivated.  These calorons must satisfy the rotationally twisted (or ``rotwisted'') boundary condition \cite{rotwisted1,rotwisted2,rotwisted3}, and are known as rotating calorons \cite{chernodub2022instantons}.  A rotating caloron is an anti-self-dual gauge field on Euclidean $\mathbb{R}^4$ that is invariant under a combined translation of Euclidean time and rotation of space (a ``glide rotation'').

Explicit examples of rotating calorons with trivial asymptotic holonomy have been constructed in \cite{chernodub2022instantons}.  But for confinement, non-trivial asymptotic holonomy is essential, because this is what leads to constituent monopoles.  A rotating caloron with non-trivial asymptotic holonomy, analogous to the Kraan--van Baal--Lee--Lu caloron \cite{kraanVanBaal1998monopoleconsts,LeeLu1998}, has been constructed in \cite{rocalorons2,rocalorons3}, but only within an approximation that assumes well-separated constituent monopoles.  Other examples of rotating calorons with non-trivial asymptotic holonomy can be found in \cite{cork2018symmrot}, but the assumptions made in \cite{cork2018symmrot} are rather restrictive.  In particular, the charge 1 calorons in \cite{cork2018symmrot} are just the Kraan--van Baal--Lee--Lu calorons \cite{kraanVanBaal1998monopoleconsts,LeeLu1998}, which trivially satisfy the rotwisted boundary condition (being invariant under \emph{both} time translations and axial rotations).



This article addresses the problem of constructing rotating calorons with non-trivial asymptotic holonomy.  It does so within the framework of the Atiyah--Drinfeld--Hitchin--Manin--Nahm (ADHMN) transform \cite{ADHM1978construction,Nahm1983,Nahm1983trsfm,nahm1984MonoCal}, which for gauge group $\SU(2)$ is known to be a complete construction of all zero-temperature instantons and finite-temperature calorons \cite{ADHM1978construction,CharbonneauHurtubise2007nahm.tfm}.  The transform associates instantons with ADHM matrices solving an algebraic equation, and it associates calorons with matrix-valued solutions of an integrable system of ordinary differential equations, the Nahm equations.  The ADHMN approach is very powerful and can be applied in many other situations; see \cite{jardim2004,charbonneau2019spatially} for reviews and \cite{KronheimerNakajima1990yang,cherkis2011instantonsGrav,CherkisLarrain-HubachStern2024instantons-m-TN-II,charbonneauNagy2026construction} for some recent developments. 

Starting from the ADHM description of instantons, we develop in Section \ref{sec.Rotwisted-cal-NT} a Nahm transform for rotating calorons following an approach inspired by \cite{KraanVanBaal1998}.  Our transform has the interesting feature that it leads to delayed-differential equations, rather than differential equations.  We call these equations the delayed Nahm equations because they reduce to the Nahm equations when the rotation angle is zero.  Like the Nahm equations, the delayed Nahm equations are integrable.

In Section \ref{sec:rat-rot-cal} we show that, for rational rotation angles, our transform is related to the standard Nahm transform via a direct image construction.  This result has important consequences: when combined with work of Charbonneau--Hurtubise \cite{CharbonneauHurtubise2007nahm.tfm}, it means that our Nahm transform is complete for rational rotation angles.  In other words, the transform produces all rotating calorons of a given charge and rational rotation angle.

Our main result is a construction of rotating calorons with charge 1 and rotation angle $\pi$.  By solving the delayed Nahm equations explicitly and solving the boundary conditions implicitly, we prove in Section \ref{sec:solutions} existence of an eight-parameter family of rotating calorons.

We also construct the associated rotating calorons on a lattice using a numerical implementation of the Nahm transform.  In Section \ref{sec:pictures} we exhibit the first pictures of rotating calorons with non-trivial asymptotic holonomy.  It would be very interesting to compare these with lattice studies of rotating QCD.

We leave for future work the question of whether our transform is complete for non-rational rotation angles.  In fact, for non-rational rotation angles it seems to be very difficult to establish completeness using standard methods (we propose some alternative ways to construct rotating calorons that avoid this difficulty in the conclusion).  We note that the delayed Nahm equations are integrable, so add to the growing list of known integrable delay-differential equations \cite{delay1,delay2,delay3,delay4,delay5,delay6,delay7,delay8,delay9,delay10,delay11,delay12,delay13,delay14,delay15,delay16,delay17}.  We defer for the future a more thorough investigation of the delayed Nahm equations from this perspective (see Section \ref{sec:concl} for more comments on this).

\section{Rotating calorons and Nahm transform}\label{sec.Rotwisted-cal-NT}
A gauge field $A_\mu$ is called anti-self-dual if its field strength tensor $F_{\mu\nu}=\partial_\mu A_\nu-\partial_\nu A_\mu+[A_\mu,A_\nu]$ satisfies $\frac12\varepsilon_{\mu\nu\kappa\lambda}F_{\kappa\lambda}=-F_{\mu\nu}$.  
We are interested in gauge fields invariant under glide rotations, i.e.\ under the action of $\Z$ generated by
\begin{align}\label{twisted-gf}
(t,\vec{x})\mapsto(t+\tfrac{2\pi}{T},R\vec{x}),
\end{align}
where $T>0$ and $R$ is a rotation by some fixed angle $\frac{2\pi\theta}{T}$ around a unit vector $\vec{n}$.  This requirement is known as the \emph{rotwisted boundary condition} \cite{rotwisted1,rotwisted2,rotwisted3}, and we will refer to anti-self-dual gauge fields invariant under \eqref{twisted-gf} as \emph{rotating calorons}.  Without loss of generality we assume from now on that $\vec{n}=\vec{e}_3$.

The aim of this section is to adapt the ADHM construction \cite{ADHM1978construction} of instantons on $\R^4$ to obtain an analogous construction of rotating calorons with gauge group $\SU(2)$. To simplify things we shall here use the scale-invariance of the anti-self-duality equations to fix $T=2\pi$, and this can be rescaled back later.
\subsection{Nahm data for rotating calorons}
We start by reviewing the ADHM construction of instantons (see \cite{MantonSutcliffe2004} for details). The ADHM data for a charge $k$ instanton consists of a symmetric $k\times k$ matrix $M$ of quaternions and a row vector $L$ of quaternions. These matrices are required to satisfy the constraint that the matrix $\Delta(x)^\dagger\Delta(x)$ is real and invertible for all $x\in\H$, where
\begin{align}
    \Delta(x)=\begin{pmatrix}
        L\\M-x\id
    \end{pmatrix}.
\end{align}
It is sufficient to impose this at the origin $x=0$. To obtain an instanton we determine for each $x\in\H$ a $k+1$ column vector $V(x)$ solving
\begin{align}\label{kernel-ADHM}
    \Delta(x)^\dagger V(x)=0,\quad V(x)^\dagger V(x)=\bu,
\end{align}
and set the gauge field as that of the induced connection
\begin{align}
    A(x)=V(x)^\dagger\d V(x).\label{induced-connection}
\end{align}
Translations and rotations of $\R^4$ act on ADHM data as follows:
\begin{equation}
M\mapsto gMh^{-1}+y\id,\quad L\mapsto Lh^{-1}.
\end{equation}
Here $(g,h)\in \SU(2)\times \SU(2)\cong {\rm Spin}(4)$ are a pair of unit quaternions describing a rotation of $\R^4$, and $y$ is a quaternion describing a translation of $\R^4$.  Global gauge transformations of the instanton correspond to transformations
\begin{equation}
M\mapsto M,\quad L\mapsto pL
\end{equation}
of the ADHM data by a unit quaternion $p\in \SU(2)$.

In order to adapt this to obtain rotating calorons, we mimic the approach of Kraan--van Baal \cite{KraanVanBaal1998} in the construction of ordinary calorons.  The basic idea is to think of a rotating caloron as an infinite-charge instanton invariant under the action \eqref{twisted-gf}.  The corresponding ADHM data $\Delta(x)$ will be an infinite matrix which acts on an infinite vector $V$.  This is reinterpreted via the Fourier transform as a differential operator acting on a circle-valued function.  The ADHM constraints translate into delayed-differential equations for the coefficients of this operator, which in the case of non-rotating calorons are the Nahm equations.

To implement this construction we suppose that the infinite matrices $L,M$ can be decomposed into blocks $L_a,M_{ab}$ of size $1\times k$ and $k\times k$, indexed by integers $a,b$.  We represent the imaginary quaternions using the Pauli matrices $-\ii\sigma^a$.  The constraint
\begin{align}\label{twisted-caloron-ADHM}
\begin{aligned}
    L_{a+1}&=\exp(\ii\tfrac{\phi}{2}\sigma^3)L_a\exp(-\ii\tfrac{\theta}{2}\sigma^3),\\
    M_{a+1,b+1}&=\exp(\ii\tfrac{\theta}{2}\sigma^3)M_{a,b}\exp(-\ii\tfrac{\theta}{2}\sigma^3)+\delta_{a,b}\bu\,\id_k,
\end{aligned}
\end{align}
then ensures that the corresponding instanton is invariant under \eqref{twisted-gf}, up to a global gauge transformation $\exp(\ii\tfrac{\phi}{2}\sigma^3)$.  The parameter $\phi$ will ultimately be related to the asymptotic holonomy (or Polyakov loop) of the rotating caloron.  The solution $V(x)$ of \eqref{kernel-ADHM} will have a similar decomposition into $k\times 1$ blocks $V_a(x)$ satisfying
\begin{align}
    V_a(t+1,R\vec{x})=\exp(\ii\tfrac{\theta}{2}\sigma^3)V_{a-1}(t,\vec{x})\exp(-\ii\tfrac{\phi}{2}\sigma^3).
\end{align}
The gauge field \eqref{induced-connection} will therefore satisfy
\begin{align}
    A(t+1,R(\vec{e}_3,\theta)\vec{x})=\exp(\ii\tfrac{\phi}{2}\sigma^3)A(t,\vec{x})\exp(-\ii\tfrac{\phi}{2}\sigma^3).
\end{align}

Now we reinterpret these operators using the Fourier transform.  We interpret the $k\times 1$ blocks $V_a$ as Fourier coefficients of a function $v(s)=\sum_a V_ae^{\ii as}$.
The constraint \eqref{twisted-caloron-ADHM} is solved in the basis of Pauli matrices by
\begin{align}
    \begin{aligned}
        L_a&=\frac{1}{\sqrt{2\pi}}\begin{pmatrix}
            U_+{\rm e}^{\ii a(\tfrac{\phi}{2}-\tfrac{\theta}{2})}&W_+{\rm e}^{\ii a(\tfrac{\phi}{2}+\tfrac{\theta}{2})}\\
            U_-{\rm e}^{-\ii a(\tfrac{\phi}{2}+\tfrac{\theta}{2})}&W_-{\rm e}^{-\ii a(\tfrac{\phi}{2}-\tfrac{\theta}{2})}
        \end{pmatrix},\\
    M_{a,b}&=\begin{pmatrix}
        \ii A_{a-b}^\dagger+a\delta_{ab}\id_k&\ii B_{b-a}{\rm e}^{\ii b\theta}\\
        \ii B_{a-b}^\dagger{\rm e}^{-\ii a\theta}&-\ii A_{b-a}+a\delta_{ab}\id_k
    \end{pmatrix},
    \end{aligned}\label{twisted-ADHM-sol}
\end{align}
where $U_\pm,W_\pm$ are complex $k$ row vectors, and $A_a,B_a$ are $k\times k$ complex matrices. The factor of $\sqrt{2\pi}$ has been introduced for convenience. The matrices $A_a,B_a$ may be thought as components in a Fourier expansion of functions $\alpha(s)=\sum_aA_a{\rm e}^{\ii as}$ and $\beta(s)=\sum_aB_a{\rm e}^{\ii as}$ acting on $v(s)$ by multiplication.

The remaining parts of \eqref{twisted-ADHM-sol} can be identified with the operator $S_\pm$ which acts on $v(s)$ via $S_\pm v(s)=v(s\pm\theta)$, and the differential operator $\frac{\d}{\d s}$ which acts as $\frac{\d}{\d s}v(s)=v\,'(s)$. These act as linear operators in the Fourier basis as
\begin{align}
\begin{aligned}
    S_\pm v(s)&=\sum_a{\rm e}^{\pm\ii a\theta}V_a{\rm e}^{\ii as},\\
    \frac{\d}{\d s}v(s)&=\sum_a\ii aV_a{\rm e}^{\ii as}.
\end{aligned}
\end{align}
Therefore, from the solution \eqref{twisted-ADHM-sol}, acting on $V$ with $M$ and $M^\dagger$ as matrices in the Fourier basis is equivalent to acting on $v(s)$ with the operators
\begin{align}
    M=\ii\begin{pmatrix}
        \frac{\d}{\d s}+\alpha^\dagger&\beta S_+\\
         S_-\beta^\dagger&\frac{\d}{\d s}-\alpha
    \end{pmatrix},\quad M^\dagger=-\ii\begin{pmatrix}
        \frac{\d}{\d s}+\alpha&\beta S_+\\
        S_-\beta^\dagger&\frac{\d}{\d s}-\alpha^\dagger
    \end{pmatrix}.\label{M-ops}
\end{align}
Likewise
\begin{multline}
    L_a^\dagger L_b=\frac{1}{\sqrt{2\pi}}\begin{pmatrix}
        U_+^\dagger\\\e^{-\ii a\theta}W_+^\dagger
    \end{pmatrix}\begin{pmatrix}
        U_+&W_+\e^{\ii b\theta}
    \end{pmatrix}\e^{\tfrac{\ii}{2}(a-b)(\theta-\phi)}\\+\frac{1}{\sqrt{2\pi}}\begin{pmatrix}
        U_-^\dagger\\\e^{-\ii a\theta}W_-^\dagger
    \end{pmatrix}\begin{pmatrix}
        U_-&W_-\e^{\ii b\theta}
    \end{pmatrix}\e^{\tfrac{\ii}{2}(a-b)(\theta+\phi)},
\end{multline}
so the operator $L^\dagger L$ acts on $v$ as multiplication with
\begin{multline}\label{LL-ops}
    L^\dagger L=\begin{pmatrix}
        U_+^\dagger\\
        S_-W_+^\dagger
    \end{pmatrix}\delta\left(s+\tfrac{\theta}{2}-\tfrac{\phi}{2}\right)\begin{pmatrix}
        U_+&W_+S_+
    \end{pmatrix}\\
    +\begin{pmatrix}
        U_-^\dagger\\
        S_-W_-^\dagger
    \end{pmatrix}\delta\left(s+\tfrac{\theta}{2}+\tfrac{\phi}{2}\right)\begin{pmatrix}
        U_-&W_-S_+
    \end{pmatrix}.
\end{multline}
Recall that the operator $\Delta(0)^\dagger\Delta(0)=L^\dagger L+M^\dagger M$ must be real, i.e., proportional to the real quaternion $\id_2$. Using \eqref{M-ops} and \eqref{LL-ops} this holds if and only if
\begin{align}\label{delay-eqs}
    0&=\frac{\d(\alpha^\dagger+\alpha)}{\d s}(s)+[\alpha,\alpha^\dagger](s)+\beta(s)\beta^\dagger(s)-\beta^\dagger(s-\theta)\beta(s-\theta)\notag\\
    &\qquad\quad+U_+^\dagger U_+\delta(s+\tfrac{\theta}{2}-\tfrac{\phi}{2})+U_-^\dagger U_-\delta(s+\tfrac{\theta}{2}+\tfrac{\phi}{2})\notag\\
    &\qquad\quad-W_+^\dagger W_+\delta(s-\tfrac{\theta}{2}-\tfrac{\phi}{2})-W_-^\dagger W_-\delta(s-\tfrac{\theta}{2}+\tfrac{\phi}{2}),\\
    0&=\frac{\d\beta}{\d s}(s)+\alpha(s)\beta(s)-\beta(s)\alpha(s+\theta)\notag\\
    &\qquad\quad+U_+^\dagger W_+\delta(s+\tfrac{\theta}{2}-\tfrac{\phi}{2})+U_-^\dagger W_-\delta(s+\tfrac{\theta}{2}+\tfrac{\phi}{2}).\notag
\end{align}
When $\theta=0$ the equations \eqref{delay-eqs} are the Nahm equations on the circle $\bigslant{\R}{2\pi\Z}$ written in Donaldson's complex form \cite{donaldson1984}: writing $\alpha=T_0-\ii T_3$ and $\beta=-T_2-\ii T_1$ with $T_\mu$ anti-hermitian recovers the Nahm equation $\frac{\d T_i}{ds}+[T_0,T_i]+\frac12\varepsilon_{ijk}[T_j,T_k]=0$, decorated with $\delta$ function sources.

For $\theta\neq0$ the equations \eqref{delay-eqs} are delayed-differential equations which generalise the Nahm equations on the circle. Note that we can easily reinstate the arbitrary translation in $t$ of \eqref{twisted-gf} by allowing the Nahm data solving \eqref{delay-eqs} to be defined instead on the circle $\bigslant{\R}{T\Z}$.  This is important for physical applications where $T>0$ corresponds to the temperature.

When $\theta=0$ the $\delta$-function sources in \eqref{delay-eqs} are precisely those that correspond to $\SU(2)$ calorons with net magnetic charge 0.  So solutions of \eqref{delay-eqs} with $\theta\neq0$ should correspond to rotated calorons with net magnetic charge 0 (and we confirm this for rational $\theta/T$ in the next section).  The Nahm transform is also able to produce non-rotating calorons with non-zero net magnetic charge, and it would be interesting to try to extend this to rotating calorons with non-zero net magnetic charge.

The delayed equations \eqref{delay-eqs} are invariant under the action of unitary gauge transformations, which locally act as functions $g:\bigslant{\R}{T\Z}\lto\U(k)$ via
\begin{align}\label{gauge transformations}
    \begin{aligned}
        \alpha(s)&\mapsto g(s)^{-1}\alpha(s)g(s)+g(s)^{-1}\frac{\d g}{\d s}(s),\\
        \beta(s)&\mapsto g(s)^{-1}\beta(s)g(s+\theta),\\
        U_\pm&\mapsto U_\pm g(-\tfrac{\theta}{2}\pm\tfrac{\phi}{2}),\\
        W_\pm&\mapsto W_\pm g(\tfrac{\theta}{2}\pm\tfrac{\phi}{2}).
    \end{aligned}
\end{align}
We shall consider the moduli space of solutions of \eqref{delay-eqs} modulo this action as the space of Nahm data for rotating calorons.
\subsection{Nahm transform}

We have seen that the infinite ADHM data for a rotating caloron invariant under \eqref{twisted-gf} can be repackaged as a solution of the delayed Nahm equations \eqref{delay-eqs}.  We now explain how the rotating caloron can be recovered directly from this Nahm data.

Let $x=(t,x_1,x_2,x_3)\in\R^4$ and define $z=t-\ii x_3$, $w=-x_2-\ii x_1$.  Given a function $\psi:\R\lto\C^{2k}$ satisfying $\psi(s+T)=\psi(s)$ and complex numbers $\zeta_+,\zeta_-\in\C$, we define
\begin{multline}\label{Nahm operator}
    \Delta_x^\dagger(\psi,\zeta_-,\zeta_+)=-\frac{\d\psi}{\d s}+\begin{pmatrix}
    -\alpha+\ii {z}I_,&-\beta S_++\ii{w}I_k\\
    -S_-\beta^\dagger-\ii\ol{w}I_k&\alpha^\dagger+\ii\ol{z}I_k
\end{pmatrix}\psi\\+\begin{pmatrix}U_-^\dagger\delta(s+\tfrac{\theta}{2}+\tfrac{\phi}{2})\\
W_-^\dagger\delta(s-\tfrac{\theta}{2}+\tfrac{\phi}{2})  
\end{pmatrix}\zeta_-+\begin{pmatrix}U_+^\dagger\delta(s+\tfrac{\theta}{2}-\tfrac{\phi}{2})\\
W_+^\dagger\delta(s-\tfrac{\theta}{2}-\tfrac{\phi}{2})  
\end{pmatrix}\zeta_+.
\end{multline}
This is again a $T$-periodic function $:\R\lto\C^{2k}$.  Here $I_k$ is the identity matrix and $\delta$ is the $T$-periodic version of the Dirac delta function (meaning that $\delta(s)=0\iff s\neq0\mod T$, and $\int_{-T/2}^{T/2}\delta (s)\,\d s=1$).  We recall that $S_\pm\psi(s)=\psi(s\pm\theta)$.

Suppose that $\Delta_x^\dagger(\psi,\zeta_-,\zeta_+)=0$ has two linearly independent solutions $(\psi^1,\zeta_+^1,\zeta_-^1)$ and $(\psi^2,\zeta_+^2,\zeta_-^2)$.  These can be chosen to be orthonormal, in the sense that:
\begin{equation}\label{NT normalisation}
\ol{\zeta}_+^a\zeta_+^b+\ol{\zeta}_-^a\zeta^b_-+\int_0^T\psi^a(s)^\dagger\psi^b(s)\,\d s=\delta^{ab}.
\end{equation}
The rotating caloron can be recovered by setting
\begin{equation}\label{NT gauge field}
A_\mu^{ab}=\ol{\zeta}_+^a\frac{\partial\zeta_+^b}{\partial x^\mu}+\ol{\zeta}_-^a\frac{\partial\zeta^b_-}{\partial x^\mu}+\int_0^T\psi^a(s)^\dagger\frac{\partial\psi^b}{\partial x^\mu}(s)\,\d s=\delta^{ab}.
\end{equation}
In the special case $\theta\to0$ this is just the usual Nahm transform for $\SU(2)$ calorons \cite{CharbonneauHurtubise2007nahm.tfm}.  This formula for $A$ can be derived by rewriting \eqref{induced-connection} through the Fourier transform.  Rather than spelling out the details of that calculation, we will give an alternative justification, showing that the operator \eqref{Nahm operator} is gauge covariant and satisfies a reality condition.

We do so using the language of vector bundles.  The Nahm data act naturally on the trivial rank $k$ vector bundle $E$ over $S^1_T:=\bigslant{\R}{T\Z}$.  Gauge transformations correspond to automorphisms of this vector bundle.  From the explicit form \eqref{gauge transformations} of the gauge transformations, we see that
\begin{align}\label{Nahm-data-as-homomorphisms}
\begin{aligned}
    \alpha(s)&:E_s\lto E_s,&
    \beta(s)&:E_{s+\theta}\lto E_s,\\
    U_+&: E_{-\tfrac{\theta}{2}+\tfrac{\phi}{2}}\lto\C,&U_-&: E_{-\tfrac{\theta}{2}-\tfrac{\phi}{2}}\lto\C,\\
    W_+&:E_{\tfrac{\theta}{2}+\tfrac{\phi}{2}}\lto\C,&W_-&:E_{\tfrac{\theta}{2}-\tfrac{\phi}{2}}\lto\C.
   \end{aligned}
\end{align}
We now see that $\Delta_x^\dagger$ is an operator from $\mathcal{W}$ to $\mathcal{V}$, where
\begin{align*}
\mathcal{V} &= \Gamma(E\oplus E), \\
\mathcal{W} &= \mathcal{V}\oplus \C\oplus \C,
\end{align*}
The $L^2$ adjoint of this operator is 
\begin{equation}\Delta_x:\mathcal{V}\lto\mathcal{W},\quad \Delta_x=\begin{pmatrix} \slashed{D}_x & Q_+ & Q_- \end{pmatrix},
\end{equation}
where
\begin{align*}
    \slashed{D}_x\psi&=
        \frac{\d\psi}{\d s}+\begin{pmatrix}
    -\alpha^\dagger-\ii \ol{z}&-\beta S_++\ii{w}\\
    -S_-\beta^\dagger-\ii \ol{w}&\alpha-\ii{z}
\end{pmatrix}\psi,\\
Q_\pm(\psi)&= \begin{pmatrix}
    U_\pm\psi(-\tfrac{\theta}{2}\pm\tfrac{\phi}{2})&W_\pm \psi(\tfrac{\theta}{2}\pm\tfrac{\phi}{2})
\end{pmatrix}.
\end{align*}
The operator $\slashed{D}_x$ produces a well-defined section of $E\oplus E$ because each of its components produces a section.  For example, if $v(s)$ is a section of $E$ then $S_+v(s)=v(s+\theta)\in E_{s+\theta}$ and $\beta(s)S_+v(s)\in E_s$, so $\beta S_+v$ is a section of $E$.  Similarly, $Q_\pm$ are well-defined operators because (for example) $\psi(-\tfrac{\theta}{2}-\tfrac{\phi}{2})\in E_{-\tfrac{\theta}{2}-\tfrac{\phi}{2}}$ and $U_-:E_{-\tfrac{\theta}{2}-\tfrac{\phi}{2}}\lto \C$.  This means that the operator $\Delta_x$ is gauge covariant, and so is its adjoint $\Delta_x^\dagger$.

The standard framework of the Nahm transform \cite{MantonSutcliffe2004} shows that the gauge field defined in \eqref{NT gauge field} will be anti-self-dual if the operator $\Delta_x^\dagger\Delta_x$ is real for all $x\in\R^4$, in other words, it commutes with the $2\times 2$ Pauli matrices $\sigma^j$.  A straightforward but laborious calculation shows that this condition is equivalent to the delayed Nahm equations \eqref{delay-eqs}.

In order to guarantee every rotating caloron arises in this way we require an inverse Nahm transform that maps rotating calorons to Nahm data. Usually such a transform involves a family $D_s$ of Dirac operators acting on smooth spinor-valued functions over $\bigslant{\R^4}{\Z}$, where $\Z$ acts as in \eqref{twisted-gf}. In order for this to work it is important that these Dirac operators are Fredholm.  It seems difficult to prove that the operators are Fredholm if the angle of rotation in \eqref{twisted-gf} is an irrational multiple of $2\pi$.  But in the case of rational rotations the situation is much better.  We explain in the next section that the Nahm transform for rotating calorons with rational rotation angle is a special case of the Nahm transform for calorons, so should produce all rotating calorons.

\section{Rational rotations and calorons}\label{sec:rat-rot-cal}
\subsection{Nahm equations from delayed Nahm equations}
In this section and beyond we consider the case where  the rotation angle of \eqref{twisted-gf} is a rational multiple of $\pi$, i.e., where $\theta=\frac{m}{n}T$ with $1\leq m\leq n$ positive coprime integers. To avoid needing to consider several similar cases, we also assume $\tfrac{m-1}{m}\theta\leq\phi\leq\theta$.

Now, when the rotation angle is rational, invariance of the gauge field under \eqref{twisted-gf} implies also invariance under $(t,\vec{x})\mapsto(t+\tfrac{2\pi n}{T},\vec{x})$. Instantons with this symmetry are ordinary calorons with period $\tfrac{2\pi n}{T}$.  The topological charge of the $\frac{2\pi n}{T}$-periodic caloron is given by
\begin{equation}
    N=-\frac{1}{32\pi^2}\int_{\R^3}\int_0^{\frac{2\pi n}{T}}\tr (F_{\kappa\lambda}F_{\mu\nu})\epsilon^{\mu\nu\kappa\lambda}\,\d^4x.
\end{equation}
In the previous section we described ADHM data for a rotating caloron built out of $k\times k$ blocks.  The corresponding rotating caloron should be made up of $k$ instantons within the fundamental domain $0\leq t\leq \frac{2\pi}{T}$.  So we expect that
\begin{equation}
    k=-\frac{1}{32\pi^2}\int_{\R^3}\int_0^{\frac{2\pi}{T}}\tr (F_{\kappa\lambda}F_{\mu\nu})\epsilon^{\mu\nu\kappa\lambda}\,\d^4x.
\end{equation}
It then follows from glide rotation symmetry \eqref{twisted-gf} that $N=nk$.  So it can be expected that the delayed Nahm data for a charge $k$ rotating caloron with rational rotation angle is related to the Nahm data for a charge $nk$ caloron.

This expectation is correct, and the explicit map from delayed Nahm equations to Nahm equations is as follows.  Define
\begin{align}
    S:=\left(\begin{array}{c|c}
        0 & \id_{k(n-1)} \\\hline
        \id_{k} & 0
    \end{array}\right).
\end{align}
and let
\begin{align}\label{caloron-alpha-beta}
\begin{aligned}
    \wh{\alpha}(s)&=\diag\{\alpha(s+\tfrac{\theta}{m}),\alpha(s+\tfrac{2\theta}{m}),\dots,\alpha(s+\tfrac{n\theta}{m})\},\\
    \wh{\beta}(s)&=\diag\{\beta(s+\tfrac{\theta}{m}),\beta(s+\tfrac{2\theta}{m}),\dots,\beta(s+\tfrac{n\theta}{m})\}S^m.
\end{aligned}
\end{align}
Let
\begin{align}\label{IJ-UW}
    I_+=e_{n-m}\otimes U_-,\quad J_+=e_{n}\otimes W_-,\quad I_-=e_{n}\otimes U_+,\quad J_-=e_{m}\otimes W_+,
\end{align}
where $e_j\in\C^n$ denotes the canonical unit basis row vector with a $1$ in the $j$-th component and $0$ elsewhere.  Finally, let
\begin{align}\label{mu-phi-rel}
    \mu=\frac{\theta}{2}-\frac{\phi}{2}.
\end{align}
Our assumption that $\frac{m-1}{m}\theta\leq\phi\leq\theta$ means that $0\leq \mu \leq \frac{T}{2n}$.  
A direct calculation shows that the delayed Nahm equations \eqref{delay-eqs} are equivalent to
\begin{align}
        \frac{\d(\wh{\alpha}+\wh{\alpha}^\dagger)}{\d s}+[\wh{\alpha},\wh{\alpha}^\dagger]+[\wh{\beta},\wh{\beta}^\dagger]\qquad\qquad\qquad\qquad\qquad\qquad\qquad\qquad&\notag\\
        +\left(I_+^\dagger I_+-J_+^\dagger J_+\right)\delta(s-\mu)+\left(I_-^\dagger I_--J_-^\dagger J_-\right)\delta(s+\mu)&=0,\label{NE-cal1}\\
        \frac{\d \wh{\beta}}{\d s}+[\wh{\alpha},\wh{\beta}]+I_+^\dagger J_+\delta(s-\mu)+I_-^\dagger J_-\delta(s+\mu)&=0.\label{NE-cal2}
\end{align}
These are precisely the Nahm equations on the circle $S^1_{T/n}$ which describe $\SU(2)$ calorons invariant under $t\mapsto t+\frac{2\pi n}{T}$.  Calorons correspond to solutions for which $I_\pm$, $J_\pm$ are not all zero.  The asymptotic holonomy of the associated calorons satisfies
\begin{equation}
    \lim_{|\vec{x}|\to\infty}P\exp\left(-\int_0^{\frac{2\pi n}{T}}A_0(t,\vec{x})\,\d t\right)=\begin{pmatrix} e^{\frac{2\pi i \mu}{T}} & 0 \\ 0 &e^{-\frac{2\pi i \mu}{T}}\end{pmatrix}
\end{equation}
in a suitable gauge.

This relationship between the Nahm equations and the delayed Nahm equations has a natural interpretation in terms of vector bundles.  Recall that the delayed Nahm equations live on a rank $k$ vector bundle $E\to S^1_T$ (see \eqref{Nahm-data-as-homomorphisms}).  Let $V\to S^1_{{T}/{n}}$ be the direct image of this bundle under the projection $\pi:S^1_T=\bigslant{\R}{T\Z}\lto \bigslant{\R}{\frac{T}{n}\Z}=S^1_{{T}/{n}}$.  By definition, the fibres of $V$ are
\begin{align}\label{V-fibers}
V_s=\bigoplus_{j=1}^nK_{s,j},\quad K_{s,j}:=E_{s+\frac{j}{n}T}=E_{s+\frac{j}{m}\theta}.
\end{align}
Then from \eqref{Nahm-data-as-homomorphisms} we see that
\begin{equation}
\alpha(s+\tfrac{j}{m}\theta):K_{s,j}\lto K_{s,j}\text{ and }\beta(s+\tfrac{j}{m}\theta):K_{s,j+m}\lto K_{s,j}.
\end{equation}
Hence the matrices $\wh{\alpha}$, $\wh{\beta}$ in \eqref{caloron-alpha-beta} are well-defined maps from $V_s$ to $V_s$.  They are the direct images of $\alpha$ and $\beta$.  Similarly, from \eqref{Nahm-data-as-homomorphisms} we see that
\begin{align}
\begin{aligned}
    U_+&:K_{\mu,n}\lto\C,&U_-&: K_{-\mu,n-m}\lto\C,\\
    W_+&: K_{\mu,m}\lto\C,&W_-&:K_{-\mu,n}\lto\C.
\end{aligned}\label{UW-Ks}
\end{align}
So the matrices in \eqref{IJ-UW} represent the action of these maps on $V_{-\mu}$ and $V_\mu$.  Hence the Nahm data \eqref{caloron-alpha-beta}, \eqref{IJ-UW} is the direct image of the delayed Nahm data.

The Nahm data \eqref{caloron-alpha-beta} are periodic up to gauge transformation with the matrix $S\in \U(nk)$.  This is because $K_{s+\frac{T}{n},j}=K_{s,j+1}$ in \eqref{V-fibers}.  In practice, it is sometimes useful to make a gauge transformation so that they are strictly periodic.  To do so, we define 
\begin{align}
    \begin{aligned}
        \wh{\alpha}_p(s)&=g_p(s)^{-1}\diag\{\alpha(s+\tfrac{\theta}{m}),\alpha(s+\tfrac{2\theta}{m}),\dots,\alpha(s+\tfrac{n\theta}{m})\}g_p(s)+g_p(s)^{-1}\frac{\d g_p(s)}{\d s},\\
        \wh{\beta}_p(s)&=g_p(s)^{-1}\diag\{\beta(s+\tfrac{\theta}{m}),\beta(s+\tfrac{2\theta}{m}),\dots,\beta(s+\tfrac{n\theta}{m})\}S^mg_p(s),
    \end{aligned}\label{caloron_ND-ps}
\end{align}
where $g_1:[-\mu,\mu]\lto\U(nk)$ and $g_2:[\mu,\tfrac{\theta}{m}-\mu]\lto\U(nk)$ are such that $g_1(\pm\mu)=\id$, and $g_2(\mu)=\id$ and $g_2(\tfrac{\theta}{m}-\mu)=S$.  These are periodic in $s$ and gauge equivalent to \eqref{caloron-alpha-beta}.

As such $(\alpha,\beta,U_\pm,W_\pm)$ solve the delayed equations \eqref{delay-eqs} if and only if $(\wh{\alpha}_p,\wh{\beta}_p,I_\pm,J_\pm)$ defined by \eqref{caloron_ND-ps} and \eqref{IJ-UW} solve the Nahm equations
\begin{align}
\begin{aligned}
    \frac{\d(\wh{\alpha}_p+\wh{\alpha}_p^\dagger)}{\d s}+[\wh{\alpha}_p,\wh{\alpha}_p^\dagger]+[\wh{\beta}_p,\wh{\beta}_p^\dagger]&=0,\\
    \frac{\d\wh{\beta}_p}{\d s}+[\wh{\alpha}_p,\wh{\beta}_p]&=0,
\end{aligned}
\end{align}
and the matching conditions
\begin{align}
\begin{aligned}
    \wh{\gamma}_1(\mu)-\wh{\gamma}_2(\mu)&=I_+^\dagger I_+-J_+^\dagger J_+,&
    \wh{\gamma}_2(\tfrac{\theta}{m}-\mu)-\wh{\gamma}_1(-\mu)&=I_-^\dagger I_--J_-^\dagger J_+,\\
    \wh{\beta}_1(\mu)-\wh{\beta}_2(\mu)&=I_+^\dagger J_+,&\wh{\beta}_2(\tfrac{\theta}{m}-\mu)-\wh{\beta}_1(-\mu)&=I_-^\dagger J_-,
\end{aligned}
\end{align}
where $\wh{\gamma}_p=\wh{\alpha}_p+\wh{\alpha}_p^\dagger$, which are equivalent to \eqref{NE-cal1}--\eqref{NE-cal2}.
\subsection{Cyclic symmetry of the caloron Nahm data}\label{Sec-glide-symmetry-cal}
The $\frac{2\pi n}{T}$-periodic caloron obtained from a rotating caloron will be invariant under the action of a glide rotation.  This is reflected in the Nahm data \eqref{caloron-alpha-beta}, \eqref{IJ-UW}, as we now explain.  The translation $(t,\vec{x})\mapsto (t+\frac{2\pi}{T},\vec{x})$ acts on the caloron Nahm data as
\begin{equation}
(\wh{\alpha},\wh{\beta},I_\pm,J_\pm)\mapsto (\wh{\alpha}+\ii\tfrac{2\pi}{T}\id_{nk},\wh{\beta},I_\pm,J_\pm)
\end{equation}
The rotation $R$ through angle $\frac{2\pi\theta}{T}=\frac{2\pi m}{n}$ acts on the Nahm data as 
\begin{equation}
(\wh{\alpha},\wh{\beta},I_\pm,J_\pm)\mapsto (\wh{\alpha},\e^{-\ii\tfrac{2m\pi}{n}}\wh{\beta},\e^{\ii\tfrac{m\pi}{n}}I_\pm,\e^{-\ii\tfrac{m\pi}{n}}J_\pm)
\end{equation}
So we expect the caloron Nahm data to satisfy
\begin{align}
\begin{aligned}
    \wh{\alpha}(s)&=G(s)^{-1}\wh{\alpha}_p(s)G(s)+\ii\frac{2m\pi}{n\theta}+G(s)^{-1}\frac{\d G}{\d s}(s),\\
    \wh{\beta}(s)&=\e^{-\ii\tfrac{2m\pi}{n}}G(s)^{-1}\wh{\beta}_p(s)G(s),\\
    I_\pm&=\e^{\ii(\tfrac{m\pi}{n}\pm\omega)}I_\pm G(\pm\mu),\\
    J_\pm&=\e^{\ii(-\tfrac{m\pi}{n}\pm\omega)}J_\pm G(\pm\mu),
\end{aligned}\label{glide-symmetry}
\end{align}
for some gauge transformations $G(s)$ and $e^{i\omega}$.  This expectation is correct.  The gauge transformations are given by $\omega=\tfrac{m\pi}{n\theta}(\theta+2\mu)$ and
\begin{align}
    G(s)=\e^{-\ii\tfrac{2m\pi}{n\theta}s}\diag\{\e^{\ii\tfrac{2\pi}{n}},\e^{\ii\tfrac{4\pi}{n}},\dots,1\}\otimes \id_k.
\end{align}
This confirms that caloron Nahm data are invariant under the action of the cyclic group $C_n$.  

Caloron Nahm data invariant under cyclic group actions have been investigated elsewhere. The trivial asymptotic holonomy case was considered in \cite{NakamulaSawado2013cyclic} by generalising Sutcliffe's ansatz \cite{Sutcliffe1996cyclic} for cyclic monopoles. Non-trivial asymptotic holonomy charge $N$ $\SU(2)$ calorons invariant under an action of $C_{2N}$ were classified in \cite{cork2018symmrot}; this result was extended to charge $N$ $\SU(M)$ calorons invariant under an analogous action of $C_{NM}$ in \cite{cork2018thesis}. The cyclic group actions in \cite{cork2018symmrot,cork2018thesis} were generated by a combination of a glide rotation and the ``rotation map'', which changes the asymptotic holonomy and permutes the intervals on which the Nahm matrices $\wh{\alpha}_p$, $\wh{\beta}_p$ are defined. In the case of $\SU(2)$, all of the calorons constructed in \cite{cork2018symmrot} are invariant under the action \eqref{glide-symmetry}, which generates a $C_N$ subgroup of $C_{2N}$. So \cite{cork2018symmrot} gives existence of some very special solutions of the delayed Nahm equations. In the next section we construct more general solutions in the case $k=1$, none of which may be obtained from these examples.

\section{Solutions}\label{sec:solutions}
\subsection{Delayed Nahm equations for \texorpdfstring{$k=1$}{k=1} and \texorpdfstring{$\theta=\frac{T}{2}$}{theta=T/2}}
We now discuss some solutions of the delayed equations for rational rotation angle. We will focus our attention to the simplest nontrivial case, with charge $k=1$ and twist parameter $\theta=\frac{T}{2}$.  In this case the rotation angle in \eqref{twisted-gf} is $\pi$, and the $\frac{2\pi}{T}$-periodic charge 1 rotating caloron is equivalent to a $\frac{4\pi}{T}$-periodic charge 2 caloron.  We will assume that $0\leq\phi\leq\theta$ (which is equivalent to the assumption $\frac{m-1}{m}\theta\leq\phi\leq\theta$ in the previous section).

By $2\theta$-periodicity it is enough to look on a fundamental domain $[-\tfrac{\theta}{2}-\tfrac{\phi}{2},\tfrac{3\theta}{2}-\tfrac{\phi}{2}]$. We split up this domain around the singular points in \eqref{delay-eqs} into four sub-intervals
\begin{align}
\begin{aligned}
    I_1&=[-\tfrac{\theta}{2}-\tfrac{\phi}{2},-\tfrac{\theta}{2}+\tfrac{\phi}{2}],& I_2&=[-\tfrac{\theta}{2}+\tfrac{\phi}{2},\tfrac{\theta}{2}-\tfrac{\phi}{2}],\\ I_3&=[\tfrac{\theta}{2}-\tfrac{\phi}{2},\tfrac{\theta}{2}+\tfrac{\phi}{2}],& I_4&=[\tfrac{\theta}{2}+\tfrac{\phi}{2},\tfrac{3\theta}{2}-\tfrac{\phi}{2}].
    \end{aligned}
\end{align}
Let $(\alpha_p,\beta_p)$ denote the values of $(\alpha,\beta)$ on each interval $I_p$.  These are four pairs of functions $\alpha_p,\beta_p:I_p\lto\C$.  The matching data consist of four complex numbers $U_\pm,W_\pm$; each of the pairs $(U_+,U_-)$, $(W_+,W_-)$, $(U_\pm,W_\pm)$ must be nonzero (otherwise the associated caloron is trivial). Then the delayed Nahm equations \eqref{delay-eqs} are equivalent to the ODEs
\begin{align}\label{n=2-delay}
    \begin{aligned}
        \frac{\d\gamma_i}{\d s}(s)+\beta_i(s)\beta_i^\dagger(s)-\beta_{i+2}^\dagger(s+\theta)\beta_{i+2}(s+\theta)&=0,&s&\in I_i,\,i=1,2\\
        \frac{\d\beta_i}{\d s}(s)+\alpha_i(s)\beta_i(s)-\beta_i(s)\alpha_{i+2}(s+\theta)&=0,&s&\in I_i,\,i=1,2,\\
        \frac{\d \gamma_i}{\d s}(s)+\beta_i(s)\beta_i^\dagger(s)-\beta_{i-2}^\dagger(s-\theta)\beta_{i-2}(s-\theta)&=0,&s&\in I_i,\,i=3,4,\\
        \frac{\d\beta_i}{\d s}(s)+\alpha_i(s)\beta_i(s)-\beta_i(s)\alpha_{i-2}(s-\theta)&=0,&s&\in I_i,\,i=3,4,
    \end{aligned}
\end{align}
and the matching equations
\begin{align}\label{match-case0}
\begin{aligned}
        \gamma_4(\tfrac{3\theta}{2}-\tfrac{\phi}{2})-\gamma_1(-\tfrac{\theta}{2}-\tfrac{\phi}{2})&=U_-^\dagger U_-,&
        \gamma_1(-\tfrac{\theta}{2}+\tfrac{\phi}{2})-\gamma_2(-\tfrac{\theta}{2}+\tfrac{\phi}{2})&=U_+^\dagger U_+,\\
        \gamma_3(\tfrac{\theta}{2}-\tfrac{\phi}{2})-\gamma_2(\tfrac{\theta}{2}-\tfrac{\phi}{2})&=W_-^\dagger W_-,&
        \gamma_4(\tfrac{\theta}{2}+\tfrac{\phi}{2})-\gamma_3(\tfrac{\theta}{2}+\tfrac{\phi}{2})&=W_+^\dagger W_+,\\
        \beta_4(\tfrac{3\theta}{2}-\tfrac{\phi}{2})-\beta_1(-\tfrac{\theta}{2}-\tfrac{\phi}{2})&=U_-^\dagger W_-,&
        \beta_1(-\tfrac{\theta}{2}+\tfrac{\phi}{2})-\beta_2(-\tfrac{\theta}{2}+\tfrac{\phi}{2})&=U_+^\dagger W_+,\\
        \beta_2(\tfrac{\theta}{2}-\tfrac{\phi}{2})-\beta_3(\tfrac{\theta}{2}-\tfrac{\phi}{2})&=0,&
        \beta_3(\tfrac{\theta}{2}+\tfrac{\phi}{2})-\beta_4(\tfrac{\theta}{2}+\tfrac{\phi}{2})&=0,
    \end{aligned}
\end{align}
where $\gamma_p=\alpha_p+\alpha^\dagger_p$.
\subsection{Constant solutions}\label{sec:constant}
First we look for solutions of \eqref{n=2-delay}, \eqref{match-case0} where the functions $\alpha_p,\beta_p$ are constant on each interval.  With this assumption, the matching conditions \eqref{match-case0} imply that
\begin{equation}\label{constant gamma4-gamma2}
    \gamma_4-\gamma_2=|U_+|^2+|U_-|^2=|W_+|^2+|W_-|^2\neq0.
\end{equation}
The delayed Nahm equations \eqref{n=2-delay} imply that
\begin{equation}
    (\alpha_2-\alpha_4)\beta_2=(\alpha_2-\alpha_4)\beta_4=0.
\end{equation}
Since \eqref{constant gamma4-gamma2} implies $\alpha_2-\alpha_4\neq0$, these imply that $\beta_2=\beta_4=0$.  The matching conditions \eqref{match-case0} then imply that
\begin{equation}
    \beta_2=\beta_3=\beta_4=0.
\end{equation}
The delayed Nahm equations \eqref{n=2-delay} imply that $|\beta_1|^2=|\beta_3|^2$ and so $\beta_1=0$ also.  The matching conditions \eqref{match-case0} then give
\begin{equation}
    \overline{U}_-W_-=\overline{U}_+W_+=0.
\end{equation}
In view of \eqref{constant gamma4-gamma2} this means that either $U_+=W_-=0$ or $U_-=W_+=0$.  We focus on the former case (the latter case is equivalent up to a simple transformation).  Then $|U_-|=|W_+|$ by \eqref{constant gamma4-gamma2}, so up to a gauge transformation \eqref{gauge transformations} we may take $U_-=\lambda>0$ and $W_+=\lambda e^{\ii\psi}$.  With the remaining gauge freedom we can arrange that the imaginary parts of $\alpha_1,\alpha_2,\alpha_3,\alpha_4$ are equal.  The matching conditions \eqref{match-case0} then tell us that 
\begin{equation}
    \alpha_1=\alpha_2=\alpha_3=\alpha_4-\frac{\lambda^2}{2}.
\end{equation}
We thus obtain
\begin{align}\label{const-sol}
        \begin{aligned}
            \alpha_1=\alpha_2=\alpha_3&=\xi,&\alpha_4&=\xi+\tfrac{1}{2}\lambda^2,\\
            \beta_i&=0,\\
            (U_+,U_-)&=(0,\lambda),&(W_+,W_-)&=(\lambda\e^{\ii\psi},0),
        \end{aligned}
    \end{align}
where $\lambda>0$, $\xi\in\C$, and $\psi\in[0,2\pi)$.  This is the most general constant solution of \eqref{n=2-delay} and \eqref{match-case0}.

    The associated caloron Nahm data for the solution \eqref{const-sol} is given by
    \begin{align}\label{const.cal1}
    \begin{aligned}
        \wh{\alpha}_1(s)&=\begin{pmatrix}
            \xi+\tfrac{\lambda^2}{2}&0\\
            0&\xi
        \end{pmatrix},&\wh{\beta}_1(s)&=\begin{pmatrix}
            0&0\\0&0
        \end{pmatrix}\\
        \wh{\alpha}_2(s)&=\begin{pmatrix}\xi&\tfrac{\pi}{2\phi}\ii\\
        \tfrac{\pi}{2\phi}\ii&\xi
        \end{pmatrix},&\wh{\beta}_2(s)&=\begin{pmatrix}
            0&0\\
            0&0
        \end{pmatrix},\\
        (I_+,J_+)&=(\begin{pmatrix}
            \lambda&0
        \end{pmatrix},\begin{pmatrix}
            0&0
        \end{pmatrix}),&(I_-,J_-)&=(\begin{pmatrix}
            0&0
        \end{pmatrix},\begin{pmatrix}
            \lambda\e^{\ii\psi}&0
        \end{pmatrix}).
    \end{aligned}
    \end{align}
These Nahm data appeared earlier in \cite{Harland2007}.  The corresponding gauge fields are invariant under $\SO(2)$ rotations of $\R^3$ about the $x_3$-axis.  This means that they are invariant under $(t,\vec{x})\mapsto (t+\frac{2\pi}{T},\vec{x})$ in addition to \eqref{twisted-gf}.  So they are trivial examples of charge 1 rotating calorons, being invariant under both $\frac{2\pi}{T}$ shifts of $t$ and rotations about the $x_3$-axis.  These calorons were studied in detail in \cite{KraanVanBaal1998,LeeLu1998}.
\subsection{Calorons and the reality condition}
In order to simplify the search for non-constant solutions of the delayed Nahm equations, we recall here the reality constraint for caloron Nahm data.  The caloron Nahm matrices \eqref{caloron_ND-ps} associated with the delay Nahm data are defined on the intervals $[-\mu,\mu]$ and $[\mu,\theta-\mu]$ and may be cast in the form
\begin{align}\label{caloron-Nahm-from-delay-n=2}
    \begin{aligned}
        \wh{\alpha}_1(s)&=\begin{pmatrix}
            \alpha_4(s+\theta)&0\\
            0&\alpha_2(s)
        \end{pmatrix},&\wh{\beta}_1(s)&=\begin{pmatrix}
            0&\beta_4(s+\theta)\\
            \beta_2(s)&0
        \end{pmatrix},\\
        \wh{\alpha}_2(s)&=g\cdot\begin{pmatrix}
            \alpha_1(s-\theta)&0\\
            0&\alpha_3(s)
        \end{pmatrix},&\wh{\beta}_2(s)&=g\cdot\begin{pmatrix}
            0&\beta_1(s-\theta)\\
            \beta_3(s)&0
        \end{pmatrix},
    \end{aligned}
\end{align}
where $g:[\mu,\theta-\mu]\lto\U(2)$ is a gauge transformation on the second interval with $g(\mu)=1$ and $g(\theta-\mu)=\ii\sigma^1$; we can choose $g(s)=\begin{pmatrix}
    \cos(\tfrac{\pi}{2}\tfrac{s-\mu}{\theta-2\mu})&\ii\sin(\tfrac{\pi}{2}\tfrac{s-\mu}{\theta-2\mu})\\
    \ii\sin(\tfrac{\pi}{2}\tfrac{s-\mu}{\theta-2\mu})&\cos(\tfrac{\pi}{2}\tfrac{s-\mu}{\theta-2\mu})
\end{pmatrix}$.
The matching data is likewise defined using \eqref{IJ-UW} as
\begin{align}\label{caloron-match-n=2}
    I_+&=\begin{pmatrix}
        U_-&0
    \end{pmatrix},&J_+&=\begin{pmatrix}
        0&W_-
    \end{pmatrix},&I_-&=\begin{pmatrix}
        0&U_+
    \end{pmatrix},&J_-&=\begin{pmatrix}
        W_+&0
    \end{pmatrix}.
\end{align}
Due to the isomorphism $\SU(2)\cong{\rm Sp}(1)$, caloron Nahm data may always be put into a gauge where the reality condition holds \cite{CharbonneauHurtubise2007nahm.tfm}, that is $X_1(-s)=X_1(s)^t$ and $X_2(\theta-s)=X_2(s)^t$ for $X=\wh{\alpha},\wh{\beta}$.

Imposing the reality condition on \eqref{caloron-Nahm-from-delay-n=2} we obtain the conditions
    \begin{align}\label{reality-cond-delay-data}
    \begin{aligned}
        \alpha_1(-s)&=\alpha_3(s),&\beta_1(-s)&=\beta_1(s-\theta),&\beta_3(s)&=\beta_3(\theta-s),&\forall\,s\in I_3,\\
        \alpha_2(-s)&=\alpha_2(s),&\alpha_4(\theta-s)&=\alpha_4(\theta+s),&\beta_2(-s)&=\beta_4(\theta+s),&\forall\,s\in I_2.
    \end{aligned}
    \end{align}
    It is straightforward to check that these constraints are consistent with equations \eqref{n=2-delay}. These conditions also imply constraints on $U_\pm,W_\pm$ via \eqref{match-case0} which are
    \begin{align}\label{reality-cond-matching-data}
        \begin{aligned}
            U_-^\dagger U_-&=W_+^\dagger W_+,&U_+^\dagger U_+&=W_-^\dagger W_-,&U_-^\dagger W_-+U_+^\dagger W_+&=0.
        \end{aligned}
    \end{align}
\subsection{General solution of the delayed Nahm equations}
Let us return to the problem of solving the delayed Nahm equations \eqref{n=2-delay}.  We may choose a gauge in which the imaginary parts of $\alpha_1,\alpha_2,\alpha_3,\alpha_4$ are constant and equal.  In this gauge, the equations imply that the phases of $\beta_1,\beta_2,\beta_3,\beta_4$ are constant on each interval.  The equations also imply that
\begin{equation}
    \alpha_1(s)+\alpha_3(s+\theta)\quad\text{ and }\quad\alpha_2(s)+\alpha_4(s+\theta)
\end{equation}
are constant.  So the equations reduce to equations for the differences $\alpha_i(s)-\alpha_{i+2}(s+\theta)$ and the moduli $|\beta_p(s)|$.  More precisely, the equations \eqref{n=2-delay} are equivalent to the Euler top equations
\begin{align}\label{Euler-eqns}
    a'=bc,\quad b'=ca,\quad c'=ab.
\end{align}
in which
\begin{equation}
\begin{aligned}
    (a,b,c)&=\big((\alpha_1(s)-\alpha_3(s+\theta),|\beta_1(s)|-|\beta_3(s+\theta)|,|\beta_1(s)|+|\beta_3(s+\theta)|\big)\\
    &\text{or }\big((\alpha_4(s+\theta)-\alpha_2(s),|\beta_4(s+\theta)|-|\beta_2(s)|,|\beta_4(s+\theta)|+|\beta_2(s)|\big).
\end{aligned}
\end{equation}
The Euler top equation is integrable.  Its solution is that, up to reordering, $a$, $b$ and $c$ are the three functions
\begin{align}
    D\kappa'\scn_\kappa(D(s+S)),\quad D\kappa'\nc_\kappa(D(s+S)),\quad D\dc_\kappa(D(s+S)).
\end{align}
Here $D$ and $S$ are constants, $\kappa\in[0,1]$ is an elliptic modulus, $\kappa'^2=1-\kappa^2$, and we have adopted Glaisher's notation $\scn=\sn/\cn$ for Jacobi elliptic functions.  So the general solution takes the form\footnote{Some of the sign choices here have been made for later convenience.}
\begin{align}\label{sol-ansatz}
    \begin{aligned}
        \alpha_1(s)&=\frac{D_1}{2}f_1(D_1(s+S_1))+\xi,&\beta_1(s)&=\frac{D_1}{2}{\rm e}^{\ii\psi_1}g_1(D_1(s+S_1)),\\
        \alpha_2(s)&=-\frac{D_2}{2}f_2(D_2(S_2-s))+C+\xi,&\beta_2(s)&=\frac{D_2}{2}{\rm e}^{\ii\psi_2}g_2(D_2(S_2-s)),\\
        \alpha_3(s)&=-\frac{D_1}{2}f_1(D_1(s+S_1-\theta))+\xi,&\beta_3(s)&=\frac{D_1}{2}\e^{\ii\psi_3}h_1(D_1(s+S_1-\theta)),\\
        \alpha_4(s)&=\frac{D_2}{2}f_2(D_2(S_2+\theta-s))+C+\xi,&\beta_4(s)&=\frac{D_2}{2}\e^{\ii\psi_4}h_2(D_2(S_2+\theta-s)),
    \end{aligned}
\end{align}
where $f_i$, $\frac12(h_i-g_i)$ and $\frac12(h_i+g_i))$ are permutations of the three functions 
\begin{align}\label{elliptic-fcs-NE}
    \kappa'\scn_\kappa(s),\quad \kappa'\nc_\kappa(s),\quad\dc_\kappa(s).
\end{align}

A priori there are 6 choices for $f_1,g_1,h_1$ and 6 for $f_2,g_2,h_2$, making a total of 36 cases to consider.  We claim that to solve the system \eqref{n=2-delay}, \eqref{match-case0} it is enough to consider just two cases:
\begin{align}\label{functions-from-reality}
\begin{aligned}
    (f_1,g_1,h_1)&=(\kappa_1'\scn_{\kappa_1},\kappa_1'\nc_{\kappa_1}-\dc_{\kappa_1},\kappa_1'\nc_{\kappa_1}+\dc_{\kappa_1}),\\
    (f_2,g_2,h_2)&=(\chi,\kappa_2'\scn_{\kappa_2}-\eta,\kappa_2'\scn_{\kappa_2}+\eta),
\end{aligned}
\end{align}
where $(\chi,\eta)\in\{(\dc_{\kappa_2},\kappa'_2\nc_{\kappa_2}),(\kappa'_2\nc_{\kappa_2},\dc_{\kappa_2})\}$. In particular, with these functions, we must also fix
\begin{align}\label{param-from-reality}
    S_1=\tfrac{\theta}{2},\quad S_2=0,\quad \psi_4=\psi_2+\pi.
\end{align}
The reason is twofold.  First, the reality conditions \eqref{reality-cond-delay-data} imply that the functions
\begin{equation}
    \begin{aligned}
        \alpha_1(s-\tfrac{\theta}{2})-\alpha_3(s+\tfrac{\theta}{2})&=D_1f_1(D_1(s+S_1-\tfrac{\theta}{2})\text{ and}\\
        \alpha_4(s+\theta)-\alpha_2(s)&=D_2f_2(D_2(S_2-s))
    \end{aligned}
\end{equation}
are odd and even respectively.  The function $\scn_\kappa(s-S)$ is never even and the functions $\nc_\kappa(s-S)$, $\dc_\kappa(s-S)$ are never odd, so the only choices for $f_1,f_2,S_1,S_2$ are those given in \eqref{functions-from-reality}, \eqref{param-from-reality}.  Second, swapping the two functions $\frac12(h_i\pm g_i)$ is equivalent to changing the sign of $g_i$.  This can be absorbed into the phases $\psi_1,\psi_2$ in \eqref{sol-ansatz}, so it is only necessary to consider one of the two possible orderings.

Using the reality condition for the matching data \eqref{reality-cond-matching-data}, we may also fix
\begin{align}\label{U-W-param}
        \begin{aligned}
            U_+&=\lambda\sin\tau\,\e^{\ii\sigma_+},&U_-&=\lambda\cos\tau\,\e^{\ii\sigma_-}&W_+&=\lambda\cos\tau\,\e^{\ii\psi_+}&W_-&=-\lambda\sin\tau\,\e^{\ii\psi_-},
        \end{aligned}
    \end{align}
    where $\lambda>0$ and $\tau\in[0,\tfrac{\pi}{2}]$ up to relabeling of the phases $\psi_\pm,\sigma_\pm$. We also require $\psi_--\sigma_-=\psi_+-\sigma_+$, in order to satisfy the final constraint of \eqref{reality-cond-matching-data}. There is additional gauge freedom to fix one of the phases which we will do later in context to provide a convenient parameterisation.

    Notice that as the period is $T=2\theta$ there is a symmetry of the equations given by $(\alpha(s),\beta(s),U_\pm,W_\pm)\mapsto(-\alpha(\theta-s),-\beta(\theta-s),U_\mp,W_\mp)$. Unlike the reality condition, this is not a gauge symmetry. However, it is straightforward to see that all solutions with $\tau\in[\tfrac{\pi}{4},\tfrac{\pi}{2}]$ may be obtained from those with $\tau\in[0,\tfrac{\pi}{4}]$ (after relabeling of the remaining parameters) under this symmetry. From the point of view of calorons on $S^1\times\R^3$, such solutions are related by the isometry $(t,\vec{x})\mapsto(-t,-\vec{x})$. As such, we shall focus only on the case $\tau\in[0,\tfrac{\pi}{4}]$.

The functions \eqref{elliptic-fcs-NE} have poles at $s=\pm K(\kappa)$, so we need to make sure that our Nahm data \eqref{sol-ansatz}, \eqref{functions-from-reality} is free of poles on the intervals $I_p$ where they are defined.  In terms of
\begin{align}
\nu=\frac{\phi}{2}\quad\text{and}\quad \mu=\frac{\theta-\phi}{2},
\end{align}
this means that $f_1,g_1,h_1$ cannot have poles in the interval $-\nu<D_1s<\nu$ and $f_2,g_2,h_2$ cannot have poles in the interval $-\mu<D_2s<\mu$.  So we require that $-K(\kappa_1)<D_1\nu<K(\kappa_1)$, and $-K(\kappa_2)<D_2\mu<K(\kappa_2)$.
   
    The sign of $D_1$ does not affect \eqref{sol-ansatz} and can be changed by relabeling the phases $\psi_1,\psi_3$. So we can without loss of generality fix $D_1\geq0$. A priori, $D_2$ could have any sign, however we show later (see Lemma \ref{lem.D2-positive}) that $D_2\leq0$ gives no valid solutions. So we will ultimately consider $D_i$ in the ranges
\begin{align}\label{D-intervals}
        0\leq D_1\nu<K(\kappa_1),\quad\text{and}\quad 0<D_2\mu<K(\kappa_2).
    \end{align}
\subsection{Solutions of matching conditions}
So far we have found the general solution of the delayed Nahm equations \eqref{n=2-delay} and reality conditions.  Next we aim to solve the matching conditions \eqref{match-case0}.  Before doing so, we make one more restriction: of the two solutions \eqref{reality-cond-delay-data}, we choose the solution $\chi=\dc_{\kappa_2}$ and $\eta=\kappa_2'\nc_{\kappa_2}$.  The reason is that this choice includes the constant solutions \eqref{const-sol} as a special case. We prove this below in Lemma \ref{lem.const-def}. We expect the other case to also yield solutions, however a comprehensive analysis of the moduli space is beyond the scope of this paper.

Plugging in the data \eqref{sol-ansatz} and \eqref{U-W-param}, with this choice from \eqref{functions-from-reality}, and the parameters \eqref{param-from-reality}, the matching conditions \eqref{match-case0} are equivalent to
\begin{align}\label{match-case0-k=1}
\begin{aligned}
        2C+D_2\dc_{\kappa_2}(D_2\mu)+D_1\kappa_1'\scn_{\kappa_1}(D_1\nu)&=\lambda^2\cos^2\tau,\\
        -2C+D_1\kappa_1'\scn_{\kappa_1}(D_1\nu)+D_2\dc_{\kappa_2}(D_2\mu)&=\lambda^2\sin^2\tau,\\
        D_1\e^{\ii\psi_1}(\kappa_1'\nc_{\kappa_1}(D_1\nu)-\dc_{\kappa_1}(D_1\nu))\qquad\qquad\qquad\qquad&\\
        -D_2\e^{\ii\psi_2}(\kappa_2'\scn_{\kappa_2}(D_2\mu)-\kappa_2'\nc_{\kappa_2}(D_2\mu))&=\lambda^2\sin2\tau\e^{\ii(\psi_+-\sigma_+)},\\        D_1\e^{\ii\psi_3}(\kappa_1'\nc_{\kappa_1}(D_1\nu)+\dc_{\kappa_1}(D_1\nu))\qquad\qquad\qquad\qquad&\\
        +D_2\e^{\ii\psi_2}(\kappa_2'\scn_{\kappa_2}(D_2\mu)+\kappa_2'\nc_{\kappa_2}(D_2\mu))&=0.
    \end{aligned}
\end{align}
The first and second of \eqref{match-case0-k=1} are solved by setting
\begin{align}\label{C-eq}
    C=\frac{\lambda^2}{4}\cos2\tau,
\end{align}
and
\begin{align}\label{lam-constraint}
    \lambda^2=2D_2\dc_{\kappa_2}(D_2\mu)+2D_1\kappa_1'\scn_{\kappa_1}(D_1\nu).
\end{align}
This and the final equation of \eqref{match-case0-k=1} give the following important results.
\begin{lemma}\label{lem.const-def}
    $D_1=0$ if and only if $\kappa_2=1$, moreover this gives the constant solutions \eqref{const-sol}.
\end{lemma}
\begin{proof}
It is clear from the final equation of \eqref{match-case0-k=1} that $\kappa_2=1$ implies $D_1=0$. Likewise, if $D_1=0$ either $\kappa_2=1$, or $D_2=0$. But the latter option then gives $\lambda=0$ by \eqref{lam-constraint}, and hence not valid Nahm data.

To see that this gives the constant solutions, let $D_1=0$ and $\kappa_2=1$. Then $g_2,h_2=0$ in \eqref{functions-from-reality}, so $\beta_p=0$ for all $p$ in \eqref{sol-ansatz}. Moreover, $f_2$ is constant in \eqref{functions-from-reality}, so $\alpha_p$ are constant for all $p$ in \eqref{sol-ansatz}. So the Nahm data are constant on each interval.
\end{proof}
\begin{lemma}\label{lem.D2-positive}
Any solution of the matching conditions \eqref{match-case0-k=1} with $-K(\kappa_2)/\mu<D_2\leq0$ gives $\lambda^2\leq0$, and hence invalid Nahm data.
\end{lemma}
\begin{proof}
    Suppose $D_2\leq0$. We will assume $D_1\neq0$ and $\kappa_2\neq1$ as the cases $D_1=0$ and $\kappa_2=1$ are covered by Lemma \ref{lem.const-def}. Then the final equation of \eqref{match-case0-k=1} gives $\psi_3=\psi_2$, and may be rearranged and substituted into \eqref{lam-constraint} to give
    \begin{align}\label{lam-sq-D2-input}
    \frac{\lambda^2}{2}=D_1\left(\kappa_1'\sn_{\kappa_1}(D_1\nu)-\frac{(\kappa_1'+\dn_{\kappa_1}(D_1\nu))\dn_{\kappa_2}(D_2\mu)}{\kappa_2'(\sn_{\kappa_2}(D_2\mu)+1)}\right)\nc_{\kappa_1}(D_1\nu).
\end{align}
The elliptic functions satisfy $\kappa'\leq\dn_\kappa\leq1$, and $-1<\sn_{\kappa_2}(D_2\mu)\leq0$ because $-K(\kappa_2)/\mu<D_2\leq0$.  So
\begin{align*}
    \frac{(\kappa_1'+\dn_{\kappa_1}(D_1\nu))\dn_{\kappa_2}(D_2\mu)}{\kappa_2'(\sn_{\kappa_2}(D_2\mu)+1)}\geq\frac{(\kappa_1'+\kappa_1')\kappa_2'}{\kappa_2'}=2\kappa_1'.
\end{align*}
Moreover, $\kappa_1'\sn_{\kappa_1}(D_1\nu)\leq\kappa_1'$ and $D_1\nc_{\kappa_1}(D_1\nu)>0$, because $0<D_1<K(\kappa_1)/\nu$.  So from \eqref{lam-sq-D2-input} we would obtain $\lambda^2\leq0$.
\end{proof}
This Lemma establishes the range restriction \eqref{D-intervals}, and in particular this and the final equation of \eqref{match-case0-k=1} forces $\psi_2=\psi_3+\pi$.
Putting this all together, our solution must therefore take the form
\begin{align}\label{sol-deformed}
    \begin{aligned}
        \alpha_1(s)&=\frac{D_1}{2}\kappa_1'\scn_{\kappa_1}(D_1(s+\tfrac{\theta}{2}))+\xi,\\\beta_1(s)&=\frac{D_1}{2}{\rm e}^{\ii\psi_1}(\kappa_1'\nc_{\kappa_1}(D_1(s+\tfrac{\theta}{2}))-\dc_{\kappa_1}(D_1(s+\tfrac{\theta}{2}))),\\
        \alpha_2(s)&=-\frac{D_2}{2}\dc_{\kappa_2}(D_2s)+\tfrac{\lambda^2}{4}\cos2\tau+\xi,\\\beta_2(s)&=\frac{D_2}{2}{\rm e}^{\ii\psi_3}\kappa_2'(\scn_{\kappa_2}(D_2s)+\nc_{\kappa_2}(D_2s)),\\
        \alpha_3(s)&=-\frac{D_1}{2}\kappa_1'\scn_{\kappa_1}(D_1(s-\tfrac{\theta}{2}))+\xi,\\\beta_3(s)&=\frac{D_1}{2}\e^{\ii\psi_3}(\kappa_1'\nc_{\kappa_1}(D_1(s-\tfrac{\theta}{2}))+\dc_{\kappa_1}(D_1(s-\tfrac{\theta}{2}))),\\
        \alpha_4(s)&=\frac{D_2}{2}\dc_{\kappa_2}(D_2(s-\theta))+\tfrac{\lambda^2}{4}\cos2\tau+\xi,\\\beta_4(s)&=\frac{D_2}{2}\e^{\ii\psi_3}\kappa_2'(\nc_{\kappa_2}(D_2(s-\theta))-\scn_{\kappa_2}(D_2(s-\theta))),
    \end{aligned}
\end{align}
with matching data
\begin{align}\label{U-W-param-sol}
        \begin{aligned}
            U_+&=\lambda\sin\tau\,\e^{\ii\sigma_+},&U_-&=\lambda\cos\tau\,\e^{\ii\sigma_-},\\W_+&=\lambda\cos\tau\,\e^{\ii\psi_+},&W_-&=-\lambda\sin\tau\,\e^{\ii(\psi_+-\sigma_++\sigma_-)},
        \end{aligned}
    \end{align}
    and the matching conditions \eqref{match-case0-k=1} reduce to the two equations
\begin{align}
    D_1(\kappa_1'\nc_{\kappa_1}(D_1\nu)+\dc_{\kappa_1}(D_1\nu))&=D_2(\kappa_2'\scn_{\kappa_2}(D_2\mu)+\kappa_2'\nc_{\kappa_2}(D_2\mu)),\label{match_Ds-ks}\\
    \ell_1\e^{\ii\psi_1}+\ell_2\e^{\ii\psi_3}+\ell_3\e^{\ii\omega}\sin2\tau&=0,\label{match-triangle}
\end{align}
where we have denoted $\omega=\psi_+-\sigma_+$ and
\begin{align}\label{ells}
    \begin{aligned}
        \ell_1&=D_1(\dc_{\kappa_1}(D_1\nu)-\kappa_1'\nc_{\kappa_1}(D_1\nu)),\\
        \ell_2&=D_2(\kappa_2'\nc_{\kappa_2}(D_2\mu)-\kappa_2'\scn_{\kappa_2}(D_2\mu)),\\
        \ell_3&=2D_2\dc_{\kappa_2}(D_2\mu)+2D_1\kappa_1'\scn_{\kappa_1}(D_1\nu).
    \end{aligned}
\end{align}
By the inequalities \eqref{D-intervals} we automatically have $\ell_1,\ell_2\geq0$, and $\ell_3=\lambda^2>0$.

We shall now establish existence and uniqueness of solutions of the matching conditions \eqref{match_Ds-ks}-\eqref{match-triangle} in general.
\begin{theorem}\label{thm.gen-sol}
    Let $\mu,\nu>0$, $\kappa_1\in[0,1)$, $\kappa_2\in[0,1)$, and $D_1\in(0,K(\kappa_1)/\nu)$. Then there exists a unique $D_2\in(0,K(\kappa_2)/\mu)$ satisfying \eqref{match_Ds-ks}. Moreover for all such $\mu,\nu,\kappa_i,D_i$, there exists a range $(\tau_-,\tau_+)\subset[0,\tfrac{\pi}{4}]$, such that, for each $\tau\in(\tau_-,\tau_+)$, \eqref{match-triangle} admits a unique solution $(\e^{\ii\psi_1},\e^{\ii\psi_3})$ modulo the choice of sign for both the angles $(\psi_1,\psi_3)$. Therefore the data \eqref{sol-deformed}-\eqref{U-W-param-sol} gives rise to an $8$-parameter family of solutions of the delayed Nahm equations with $k=1$.
\end{theorem}
\begin{proof}
Consider the functions $F:(0,K(\kappa_1)/\nu)\lto(0,\infty)$ and $G:(0,K(\kappa_2)/\mu)\lto(0,\infty)$ given by
\begin{align}\label{F-G-elliptic}
\begin{aligned}
    F(u)&=u(\kappa_1'\nc_{\kappa_1}(u\nu)+\dc_{\kappa_1}(u\nu)),\\
    G(u)&=u(\kappa_2'\scn_{\kappa_2}(u\mu)+\kappa_2'\nc_{\kappa_2}(u\mu)).
\end{aligned}
\end{align}
They are bijective for all $\mu,\nu>0$ and $\kappa_1\in[0,1)$, $\kappa_2\in[0,1)$. As such, for each $D_1$ there exists a unique $D_2\in(0,K(\kappa_2)/\mu)$ with $F(D_1)=G(D_2)$, i.e., \eqref{match_Ds-ks} admits a unique solution.

We can view the remaining equation \eqref{match-triangle} as an equation relating three points in the complex plane. This admits at most two solutions for the pair $(\e^{\ii\psi_1},\e^{\ii\psi_3})$, with the two solutions related by complex conjugation, if and only if the triangle inequalities
\begin{align}
    \label{tau-constraint}\ell_1\leq\ell_2+\ell_3\sin2\tau,\quad\ell_2\leq\ell_3\sin2\tau+\ell_1,\quad\ell_3\sin2\tau\leq\ell_1+\ell_2
\end{align}
hold. These are satisfied on an open interval $(\tau_-,\tau_+)\subset[0,\tfrac{\pi}{4}]$ if and only if
\begin{align}\label{ell-constraint}
    \ell_3>|\ell_2-\ell_1|.
\end{align}
It thus remains to prove that this is implied by \eqref{match_Ds-ks}. Using equation \eqref{match_Ds-ks} we can write \begin{align}\label{D2-eqn}
    D_2=\frac{D_1(\kappa_1'+\dn_{\kappa_1}(D_1\nu))\cn_{\kappa_2}(D_2\mu)}{\kappa_2'(\sn_{\kappa_2}(D_2\mu)+1)\cn_{\kappa_1}(D_1\nu)}.
\end{align}
The constraint \eqref{ell-constraint} thus becomes equivalent to the condition
\begin{multline}\label{ell-constraint-expanded}
    (\kappa_1'+\dn_{\kappa_1}(D_1\nu))\dn_{\kappa_2}(D_2\mu)\\+\kappa_1'\kappa_2'(\sn_{\kappa_2}(D_2\mu)+1)\sn_{\kappa_1}(D_1\nu)>\kappa_2'\left|\kappa_1'-\dn_{\kappa_1}(D_1\nu)\sn_{\kappa_2}(D_2\mu)\right|.
\end{multline}
Since $D_2\in(0,K(\kappa_2)/\mu)$, $\dn_{\kappa_2}(D_2\mu)>\kappa_2'$ and $\sn_{\kappa_2}(D_2\mu)<1$, so
\begin{multline*}
    (\kappa_1'+\dn_{\kappa_1}(D_1\nu))\dn_{\kappa_2}(D_2\mu)-\kappa_2'\left|\kappa_1'-\dn_{\kappa_1}(D_1\nu)\sn_{\kappa_2}(D_2\mu)\right|\\>(\kappa_1'+\dn_{\kappa_1}(D_1\nu))\kappa_2'-\kappa_2'(\kappa_1'+\dn_{\kappa_1}(D_1\nu)\sn_{\kappa_2}(D_2\mu))\\=\kappa_2'\dn_{\kappa_1}(D_1\nu)(1-\sn_{\kappa_2}(D_2\mu))>0.
\end{multline*}
We also have $\kappa_1'\kappa_2'(\sn_{\kappa_2}(D_2\mu)+1)\sn_{\kappa_1}(D_1\nu)>0$ because, by assumption, $\kappa_i'>0$.  It follows that \eqref{ell-constraint-expanded} is upheld for all values of the parameters.

The data \eqref{sol-deformed}-\eqref{U-W-param-sol} gives an $8$-parameter family. The parameter $D_2\in(0,K(\kappa_2)/\mu)$ is fixed by the equation $G(D_2)=F(D_1)$, $\lambda>0$ by the condition \eqref{lam-constraint}, and the phases $\psi_1,\psi_3$ by the equation \eqref{match-triangle}. We may also fix one of the phases (either $\sigma_+$ or $\psi_+$) to $0$ by gauge transformation. The free parameters are then $\kappa_i\in[0,1)$, $D_1\in[0,K(\kappa_1)/\nu)$, $\xi\in\C$, $\psi_+,\sigma_-\in[0,2\pi)$, and $\tau\in(\tau_-,\tau_+)$ constrained by \eqref{tau-constraint}. 
\end{proof}

The fact this is an $8$-parameter family fits with physical intuition. Physically, the caloron is a unit charge instanton on the manifold $\R^4/\mathbb{Z}$ (where the action of $\Z$ is given in \eqref{twisted-gf}).  An instanton is determined by eight parameters: four position coordinates, one scale coordinate, and three orientation coordinates.  So our family has the expected number of parameters.

The solution of the matching conditions is implicit, not explicit.  However, an explicit solution is easily obtained when the holonomy of the caloron is trivial, i.e.\ when $\mu=0$ or $\nu=0$.   
For example, when $\mu=0$ the solution \eqref{sol-deformed} is valid with the variable $D_2$ determined explicitly from \eqref{match_Ds-ks} as $D_2=\frac{D_1}{\kappa_2'}\left(\kappa_1'\nc_{\kappa_1}(D_1\tfrac{\theta}{2})+\dc_{\kappa_1}(D_1\tfrac{\theta}{2})\right)$. The trivial holonomy case again gives an $8$-parameter family, which reproduces the trivial holonomy solutions found in \cite{chernodub2022instantons}.

\section{Visualisation and numerical Nahm transform}\label{sec:pictures}

Having derived a general solution of the delayed Nahm equations, we now construct the associated rotating calorons.
We do this in two cases for which the matching conditions \eqref{match_Ds-ks}-\eqref{match-triangle} simplify, namely when $\tau=0$ and $\kappa_1=1$. We note that these are edge cases not explicitly covered by the hypotheses of Theorem \ref{thm.gen-sol}.

By performing a numerical Nahm transform, we can obtain approximations of the curvature (field strength) $F$ of the associated caloron, and using these visualise the solutions by plotting isosurfaces of the Yang--Mills action density $\EE=|F|^2$ at individual time slices. 
We focus in particular on the qualitative structure of the solutions. Heuristically, calorons can be thought of as being made up of either instantons or monopoles (also known as dyons). Our charge 1 rotating calorons correspond to charge 2 calorons and, depending on the choice of parameters, they may appear to be more instanton-like or more monopole-like.  We find that when the parameter $\lambda$ is small, the calorons appear to be made up of instantons which are localised at points in space-time and are small compared with the period $\frac{2\pi}{T}$.  Conversely, for large $\lambda$ they appear to be made up of monopoles, which are localised in space and approximately constant in (Euclidean) time.  This interpretation for $\lambda$ is consistent with what is known about the ADHM construction of instantons, where the magnitude of the vector $L$ controls the size of the instanton.
We take this a step further in Section \ref{sec-parameters} by showing how to extract constituent monopole locations directly from the Nahm data.

Our numerical Nahm transform constructs the field strength tensor at a point $(t,x_1,x_2,x_3)$ as follows.  First, it constructs $\psi$ in equation \eqref{Nahm operator} on each interval using a Runge--Kutta order 4 method.  Then the equation $\Delta_x^\dagger(\psi,\zeta_-,\zeta_+)=0$ becomes a finite system of linear equations for $\zeta_\pm$ and the initial conditions for $\psi$ on each interval.  Solving and normalising as in \eqref{NT normalisation} gives a pair of solutions, from which the gauge field can be constructed using \eqref{NT gauge field}, with derivatives approximated by finite differences.  Doing this at all points $(t,x_1,x_2,x_3)$ in a lattice allows us to reconstruct the field strength tensor $F_{\mu\nu}$.  We numerically check that the field strength tensor is anti-self-dual, as it should be.

For all numerically-generated solutions in what follows we have fixed $\theta=2\pi$.  So our plotted solutions satisfy \eqref{twisted-gf} with $\frac{2\pi}{T}=\frac{1}{2}$, and they are therefore invariant under $(t,\vec{x})\mapsto (t+1,\vec{x})$.  We have also chosen $\mu=\nu=\tfrac{\theta}{4}$.  This choice gives the constituent monopoles of the charge 2 caloron equal mass, and means that the asymptotic holonomy of the charge 2 caloron has eigenvalues $\pm \ii$ (the asymptotic holonomy of the charge 1 rotating caloron therefore has eigenvalues $(1\pm \ii)/\sqrt{2}$). In each case we plot examples at various fixed time slices (like stills from a movie), and observe as expected that each $0.5$ shift in $t$ is related by a $\pi$ rotation around the vertical $x_3$-axis.

\subsection{The case \texorpdfstring{$\tau=0$}{tau=0}}
Because of the constraint \eqref{match-triangle} the case $\tau=0$ is only valid when $\psi_1=\psi_3+\pi$ and $\ell_1=\ell_2$. Using elliptic function identities and the constraint \eqref{D-intervals}, one may show this and the other matching equation \eqref{match_Ds-ks} are equivalent to
\begin{align}\label{tau-0-eqns}
\begin{aligned}
    \kappa_1'\nd_{\kappa_1}(D_1\nu)&=\sn_{\kappa_2}(D_2\mu),\\
    D_2\kappa_2'&=D_1\kappa_1.
\end{aligned}
\end{align}
We simplify the solution ansatz \eqref{sol-deformed}--\eqref{U-W-param-sol} as follows. We can use a gauge transformation to fix $\psi_+=0$. The parameters $\xi\in\C$, $\psi_3\in[0,2\pi]$ and $\sigma_-\in[0,2\pi]$ represent an overall translation in the $(x_3,t)$-plane, an overall rotation around the $x_3$ axis, and a global phase, so we can ignore these here and set them to $0$. $D_1,D_2$ are determined by the above equations \eqref{tau-0-eqns}. This leaves two true physical parameters, $\kappa_1,\kappa_2$, to vary and solve \eqref{tau-0-eqns}. This simplifies the implicit problem, however note that because we fixed the constraint $\tau=0$, we cannot use Theorem \ref{thm.gen-sol} to guarantee existence of solutions for all $(\kappa_1,\kappa_2)$.

We shall assume $\kappa_2<1$ as we know that $\kappa_2=1$ gives the constant solutions (see Lemma \ref{lem.D2-positive}). It is clear that \eqref{tau-0-eqns} admit no solutions satisfying \eqref{D-intervals} when $\kappa_1\in\{0,1\}$, so we also assume $0<\kappa_1<1$. Thus we can use the linear constraint of \eqref{tau-0-eqns} to fix $D_2=D_1\kappa_1/\kappa_2'$, and it remains to solve a single nonlinear equation for $D_1$. This takes the form
\begin{align}
    \kappa_1'\nd_{\kappa_1}(x)-\sn_{\kappa_2}(\rho x)=0,\label{tau_0-single-eq}
\end{align}
where $x=D_1\nu$ and $\rho=\frac{\kappa_1\mu}{\kappa_2'\nu}$. We seek a solution which adheres to the constraints \eqref{D-intervals} (with $D_1\neq0$), which requires $0<x<\min\{K(\kappa_1),K(\kappa_2)/\rho\}$. We'll call these admissible solutions.

\begin{lemma}
    \eqref{tau_0-single-eq} has an admissible solution if $K(\kappa_2)/\rho<K(\kappa_1)$.
\end{lemma}
\begin{proof}
    Let $H(x)=\kappa_1'\nd_{\kappa_1}(x)-\sn_{\kappa_2}(\rho x)$. Then $H(0)=\kappa_1'>0$, and $H(K(\kappa_2)/\rho))=\kappa_1'\nd_{\kappa_1}(K(\kappa_2)/\rho)-1<0$ as $\kappa'<\dn_\kappa(u)<1$ for all $u\in(0,K(\kappa))$. $H$ is clearly continuous, so there exists $x\in(0,K(\kappa_2)/\rho)$ with $H(x)=0$ by the intermediate value theorem.
\end{proof}
As discussed above, the scale parameter $\lambda$ in \eqref{lam-constraint} indicates the regions where the caloron solution is more monopole-like or instanton-like. We probe this for these admissible solutions by numerically solving \eqref{tau_0-single-eq} in the case $\mu=\nu=\tfrac{\pi}{2}$ and computing $\lambda$ as a function of $(\kappa_1,\kappa_2)$ in the region $K(\kappa_2)/\rho<K(\kappa_1)$. This we plot in Figure \ref{fig:lambda-map}. This illustrates how the scale is impacted by proximity to the curve $K(\kappa_1)=K(\kappa_2)/\rho$, with the large scale monopole-like regions being those values near this curve, and lower scale instanton-like regions away from the curve. This matches the following analytic observation: for $(\kappa_1,\kappa_2)$ on this curve, the equation \eqref{tau_0-single-eq} is solved by $x=K(\kappa_1)$, in which case the scale \eqref{lam-constraint} diverges.

\begin{figure}
    \centering
    \includegraphics[width=0.8\linewidth]{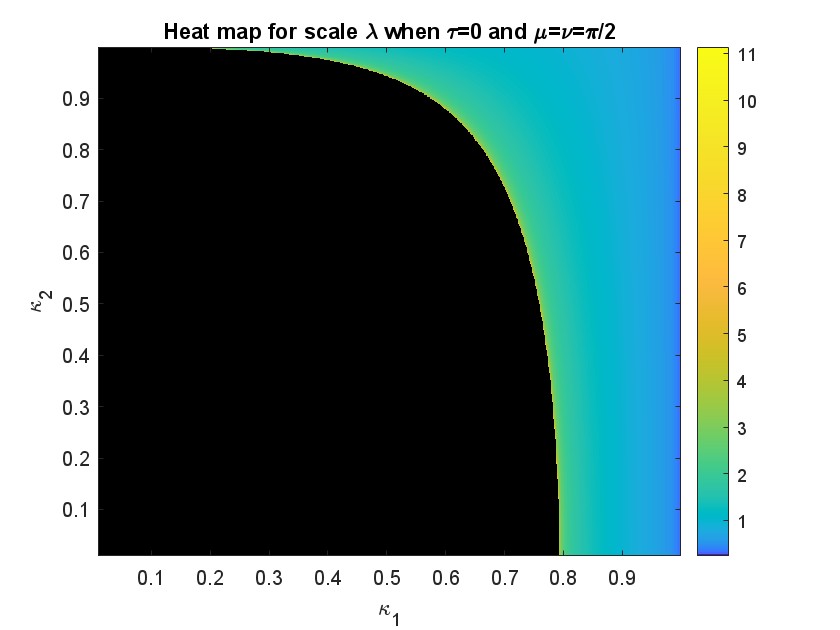}
    \caption{The scale $\lambda$ for the rotating caloron solutions with $\tau=0$ as a function of $(\kappa_1,\kappa_2)$ in the case $\theta=2\pi$ and $\mu=\nu=\tfrac{\pi}{2}$. The black region indicates $K(\kappa_1)\leq K(\kappa_2)/\rho$ where we did not seek solutions.}
    \label{fig:lambda-map}
\end{figure}

We visualise the monopole-like regime by numerically constructing the calorons for $(\kappa_1,\kappa_2)$ along the curve $K(\kappa_1)-K(\kappa_2)/\rho=1/30$, i.e. close to the critical curve where the scale diverges. Plots of this family taken as action density isosurfaces for the fixed time slices $t=0,0.225,0.5$ and $0.725$ are given in Figure \ref{fig:monopole-like-family-tau0}. We interpret these images schematically as follows. A charge 2 caloron is made up of two constituent 2-monopoles. For $\kappa_2\approx 1$ (and thus $\kappa_1\approx0$) one constituent forms a toroidal $2$-monopole in the $(x_1,x_2)$-plane, and the other constituent forms a pair of separated $1$-monopoles on the positive $x_3$-axis. As $\kappa_2$ decreases towards $0$ (and thus $\kappa_1$ increases towards $1$) the toroidal 2-monopole splits into two $1$-monopoles along the $x_1$-axis, and the separated constituent pair coalesce into a torus located on the positive $x_3$-axis and oriented with the axis of symmetry parallel to the $x_2$-axis. These observations may be replicated analytically via the Nahm data, which we discuss in Section \ref{sec-parameters}.

\begin{figure}
    \centering
    \begin{subfigure}[c]{0.24\textwidth}
         \centering
         \includegraphics[trim= 70 15 70 40,clip,width=\textwidth]{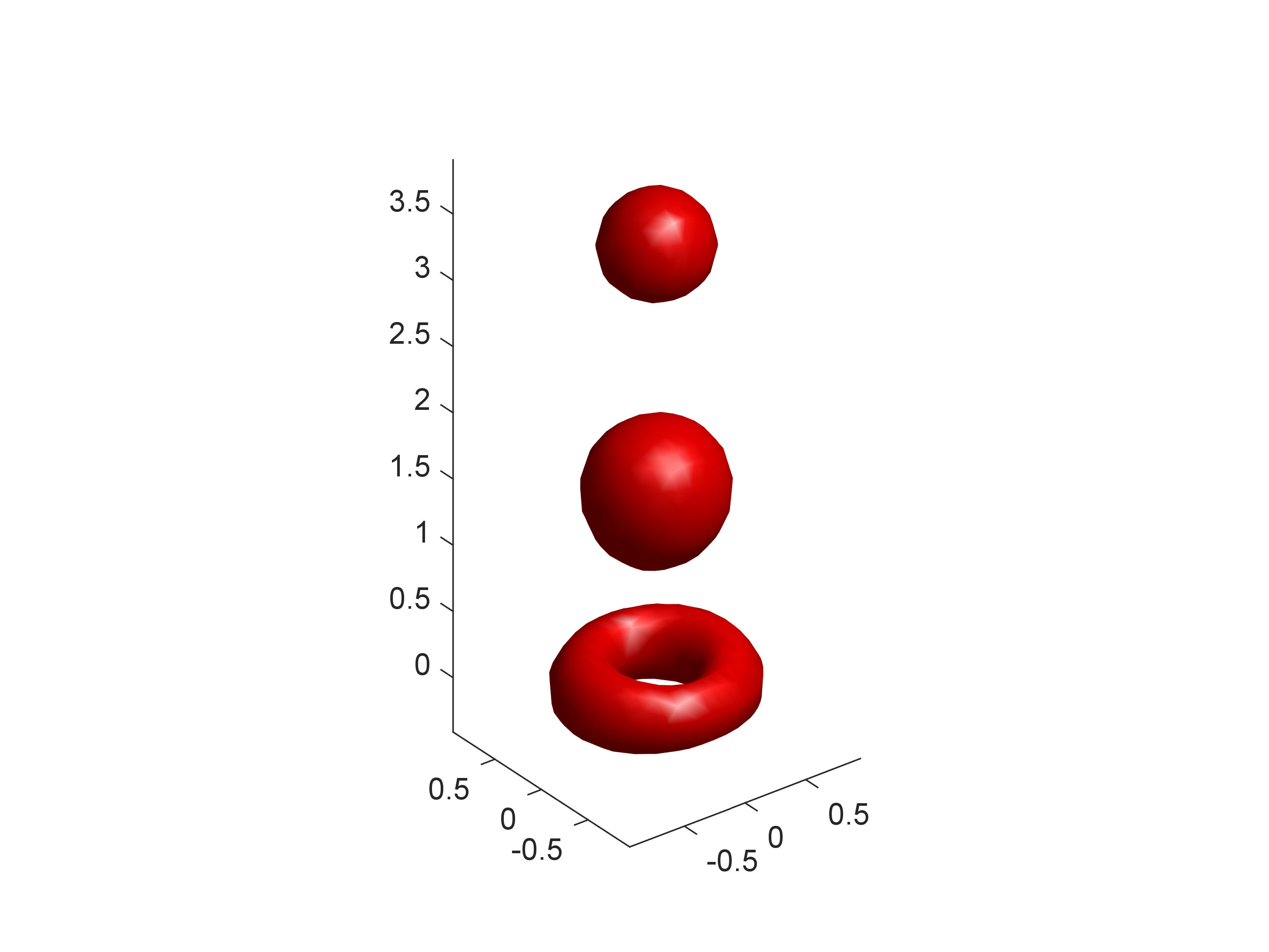}
         \caption{}
         \label{subfig:tau0-mon t=0,k2=0.99}
     \end{subfigure}
     \begin{subfigure}[c]{0.24\textwidth}
         \centering
         \includegraphics[trim= 70 15 70 40,clip,width=\textwidth]{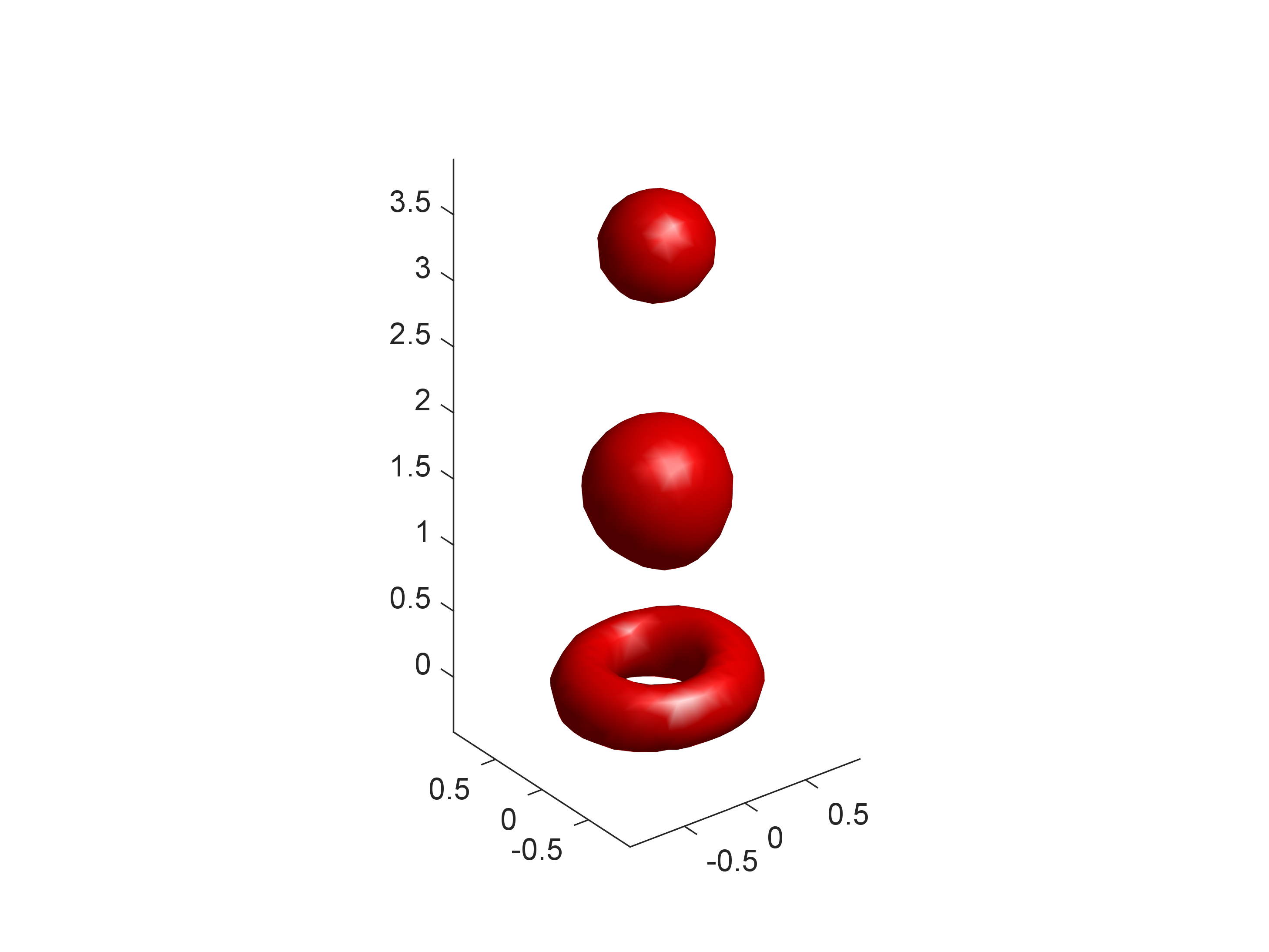}
         \caption{}
         \label{subfig:tau0-mon t=0.225,k2=0.99}
     \end{subfigure}
     \begin{subfigure}[c]{0.24\textwidth}
         \centering
         \includegraphics[trim= 70 15 70 40,clip,width=\textwidth]{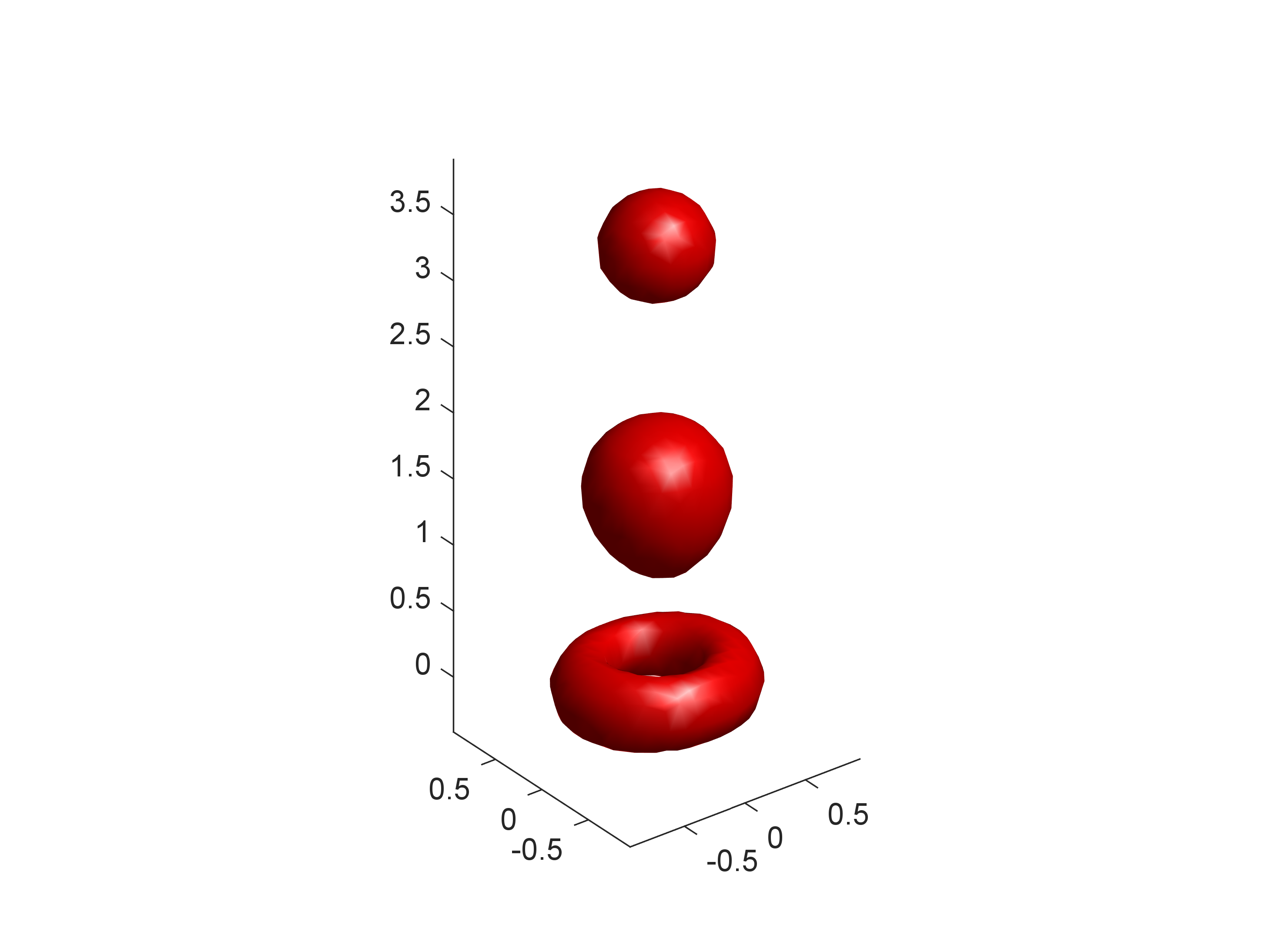}
         \caption{}
         \label{subfig:tau0-mon t=0.5,k2=0.99}
     \end{subfigure}
     \begin{subfigure}[c]{0.24\textwidth}
         \centering
         \includegraphics[trim= 70 15 70 40,clip,width=\textwidth]{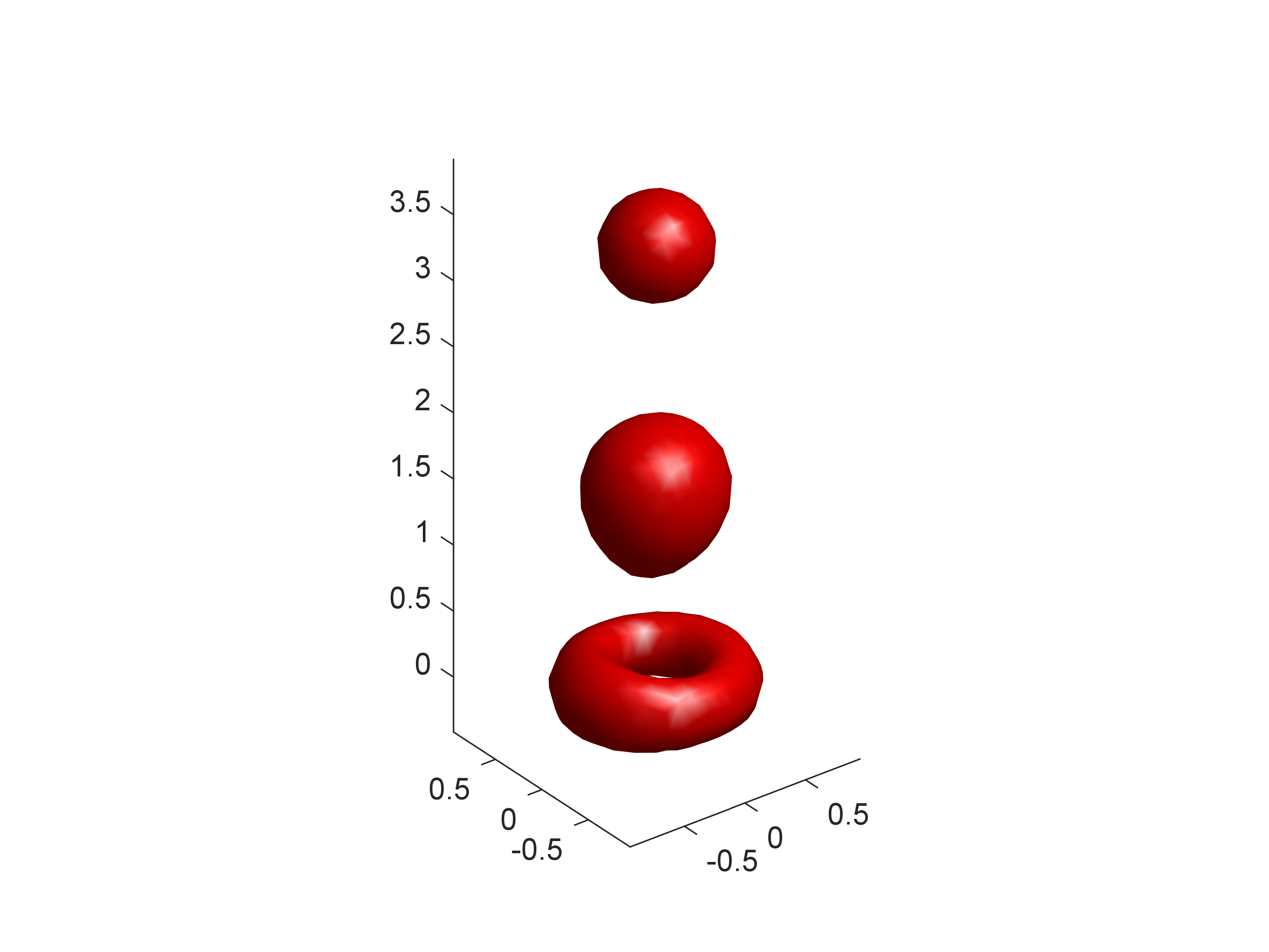}
         \caption{}
         \label{subfig:tau0-mon t=0.725,k2=0.99}
     \end{subfigure}
          \begin{subfigure}[c]{0.24\textwidth}
         \centering
         \includegraphics[trim= 70 15 70 40,clip,width=\textwidth]{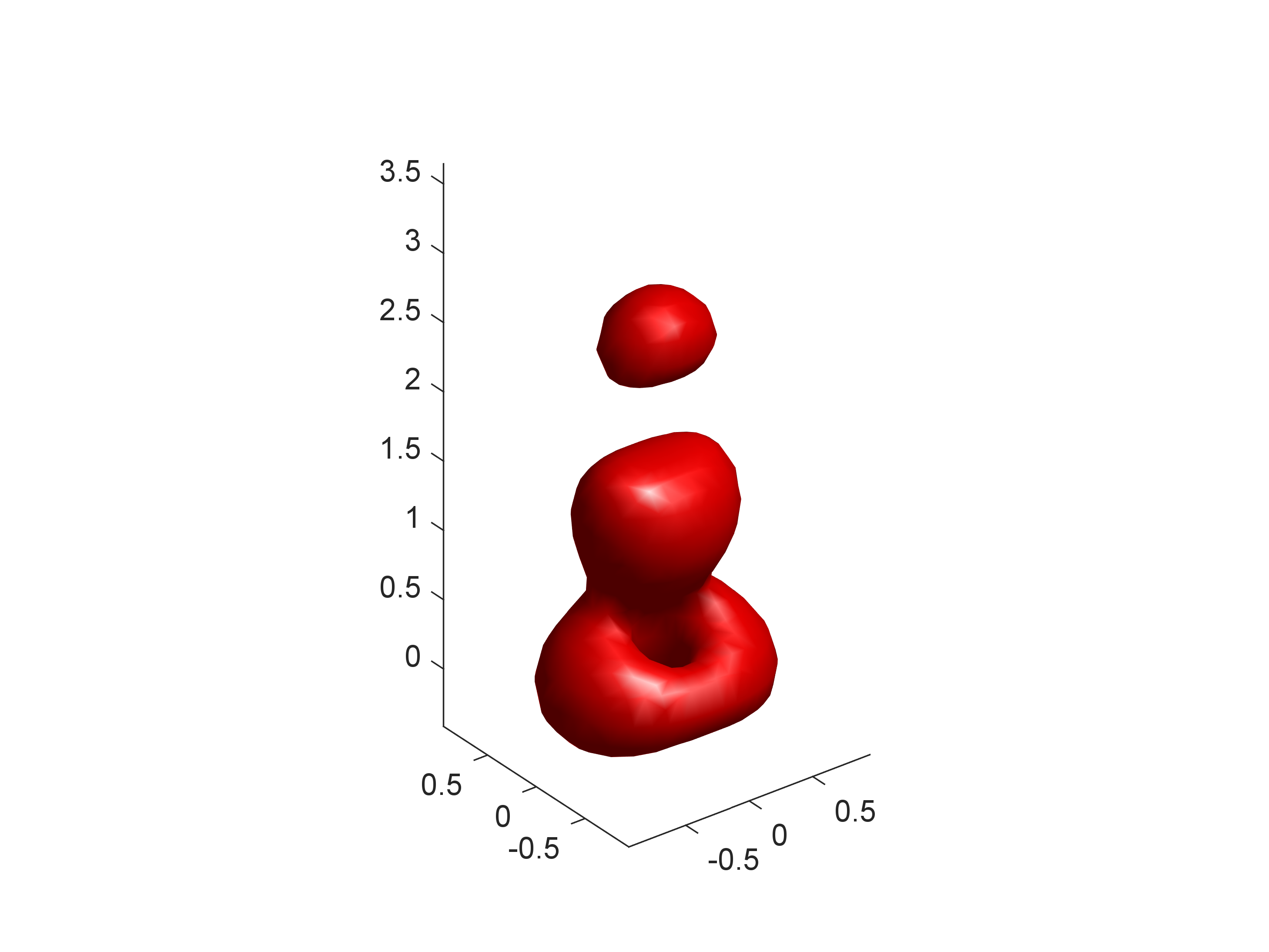}
         \caption{}
         \label{subfig:tau0-mon t=0,k2=0.9}
     \end{subfigure}
     \begin{subfigure}[c]{0.24\textwidth}
         \centering
         \includegraphics[trim= 70 15 70 40,clip,width=\textwidth]{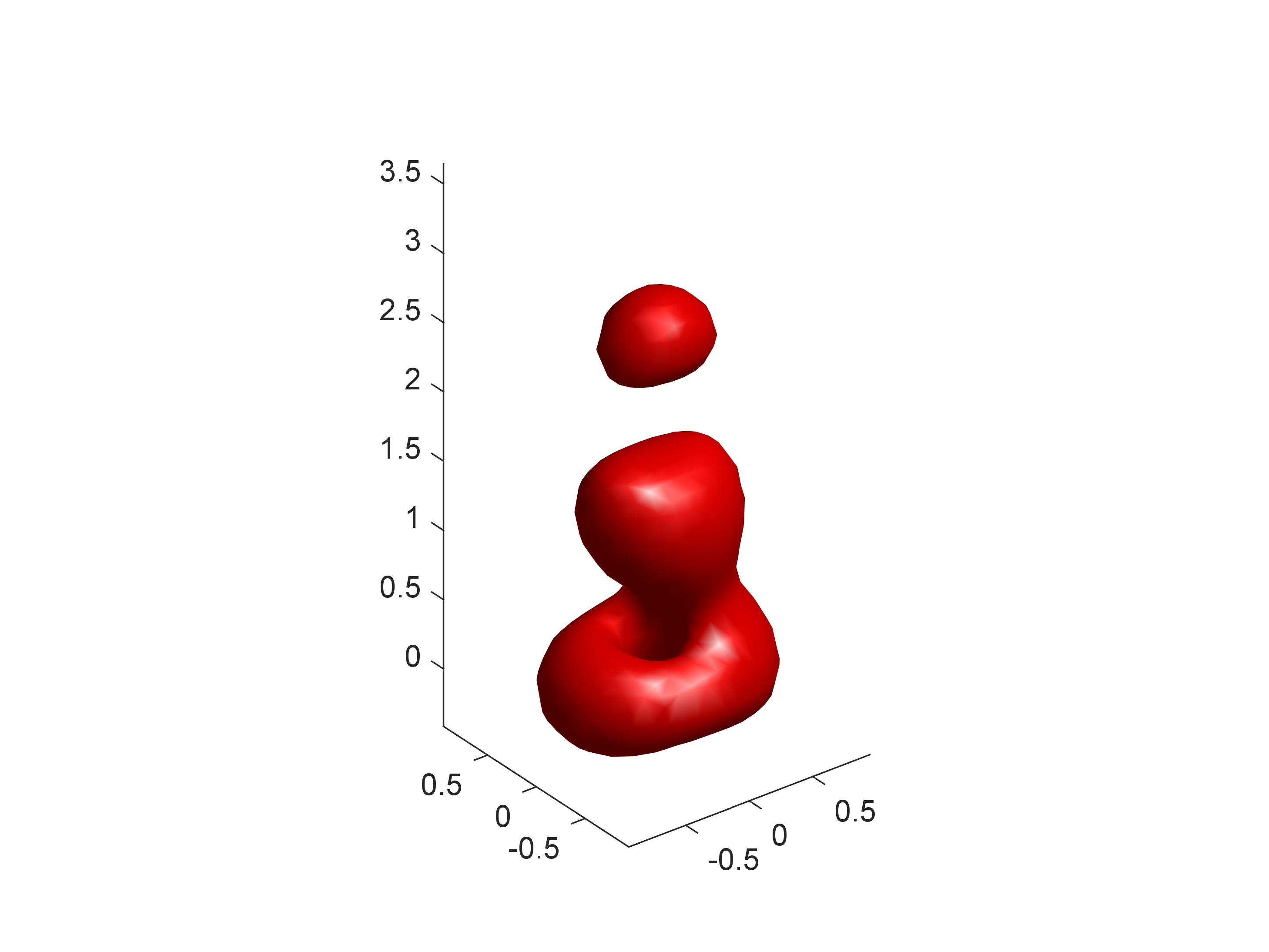}
         \caption{}
         \label{subfig:tau0-mon t=0.225,k2=0.9}
     \end{subfigure}
     \begin{subfigure}[c]{0.24\textwidth}
         \centering
         \includegraphics[trim= 70 15 70 40,clip,width=\textwidth]{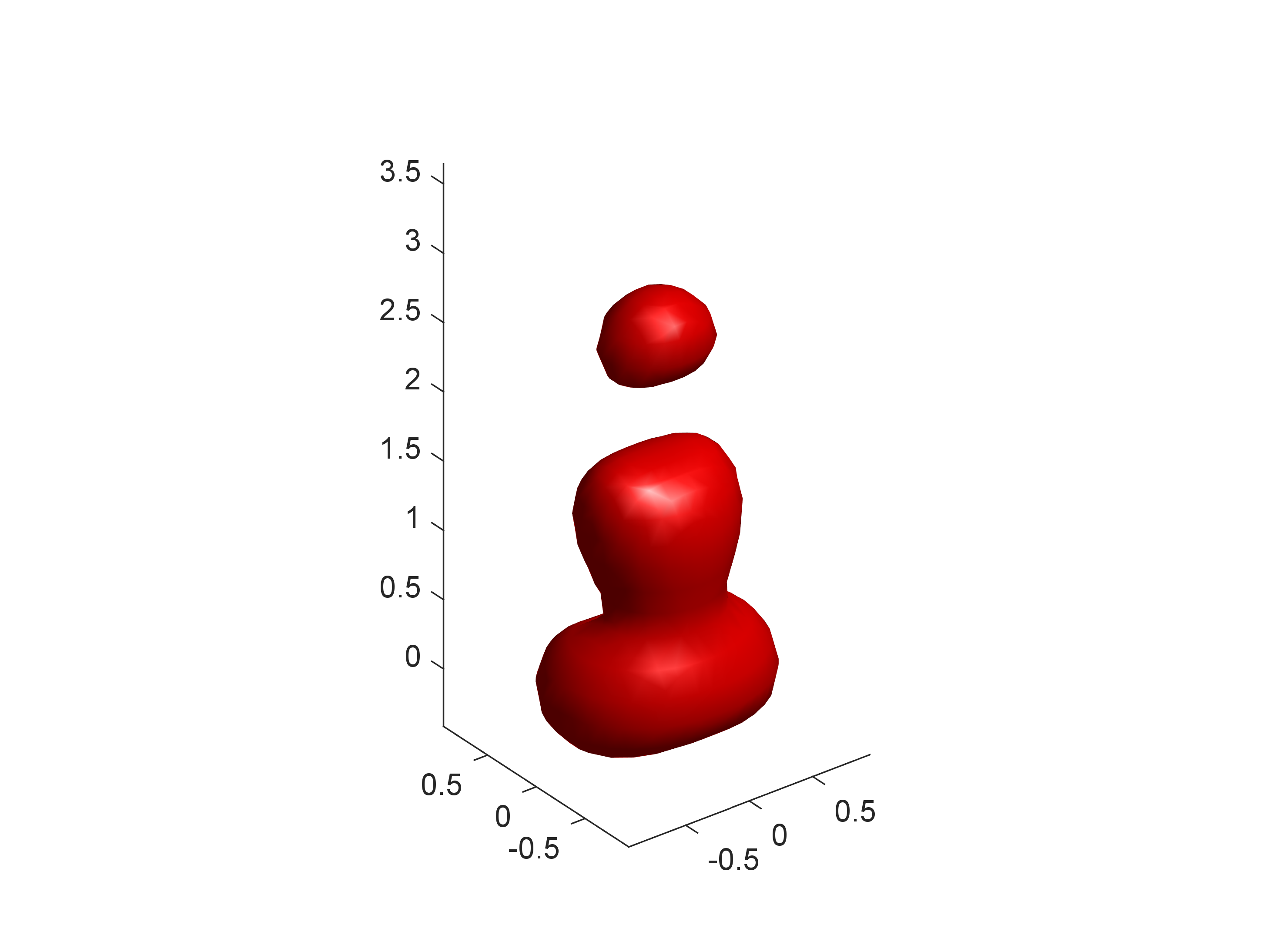}
         \caption{}
         \label{subfig:tau0-mon t=0.5,k2=0.9}
     \end{subfigure}
     \begin{subfigure}[c]{0.24\textwidth}
         \centering
         \includegraphics[trim= 70 15 70 40,clip,width=\textwidth]{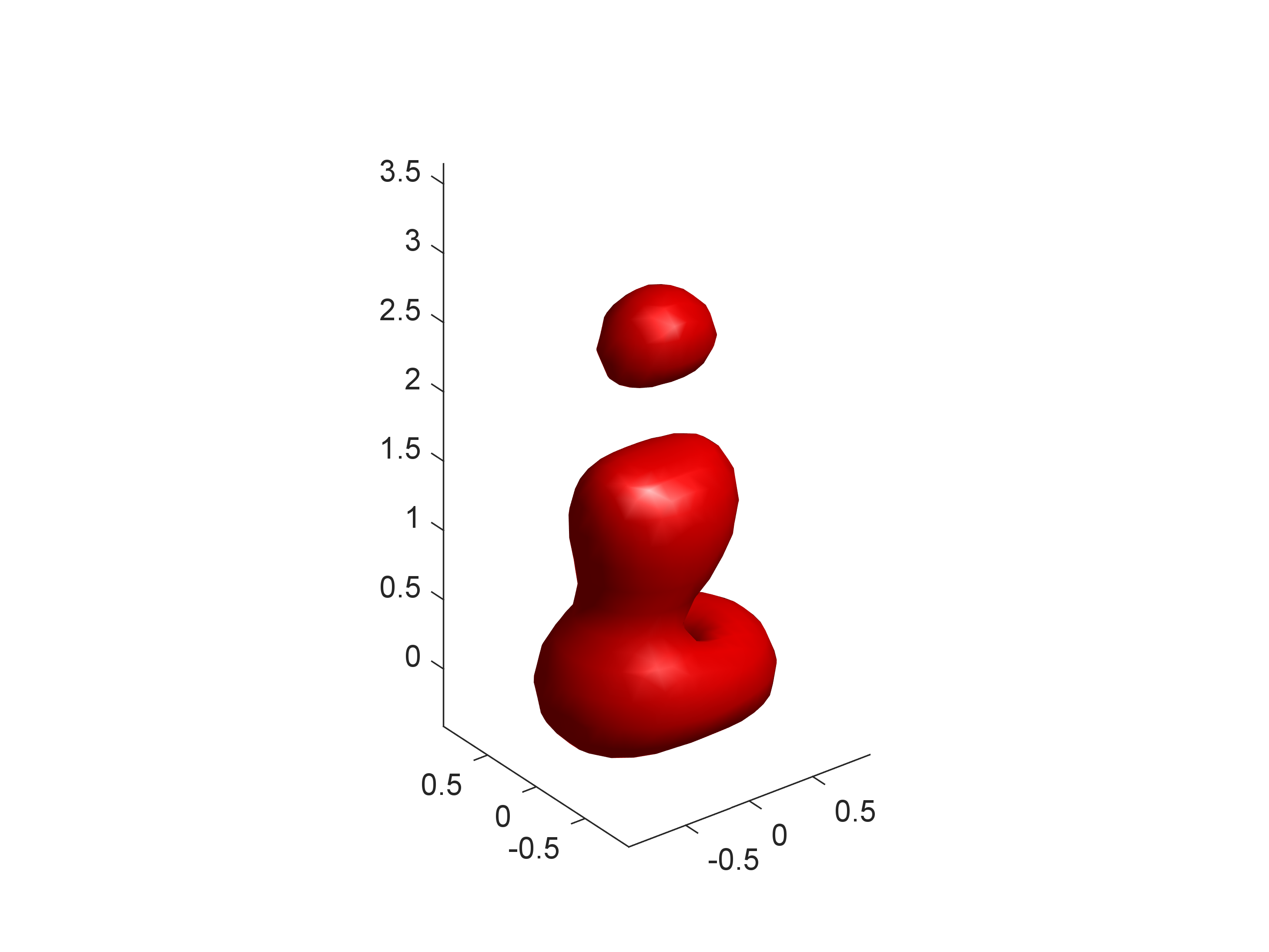}
         \caption{}
         \label{subfig:tau0-mon t=0.725,k2=0.9}
     \end{subfigure}
          \begin{subfigure}[c]{0.24\textwidth}
         \centering
         \includegraphics[trim= 70 15 70 40,clip,width=\textwidth]{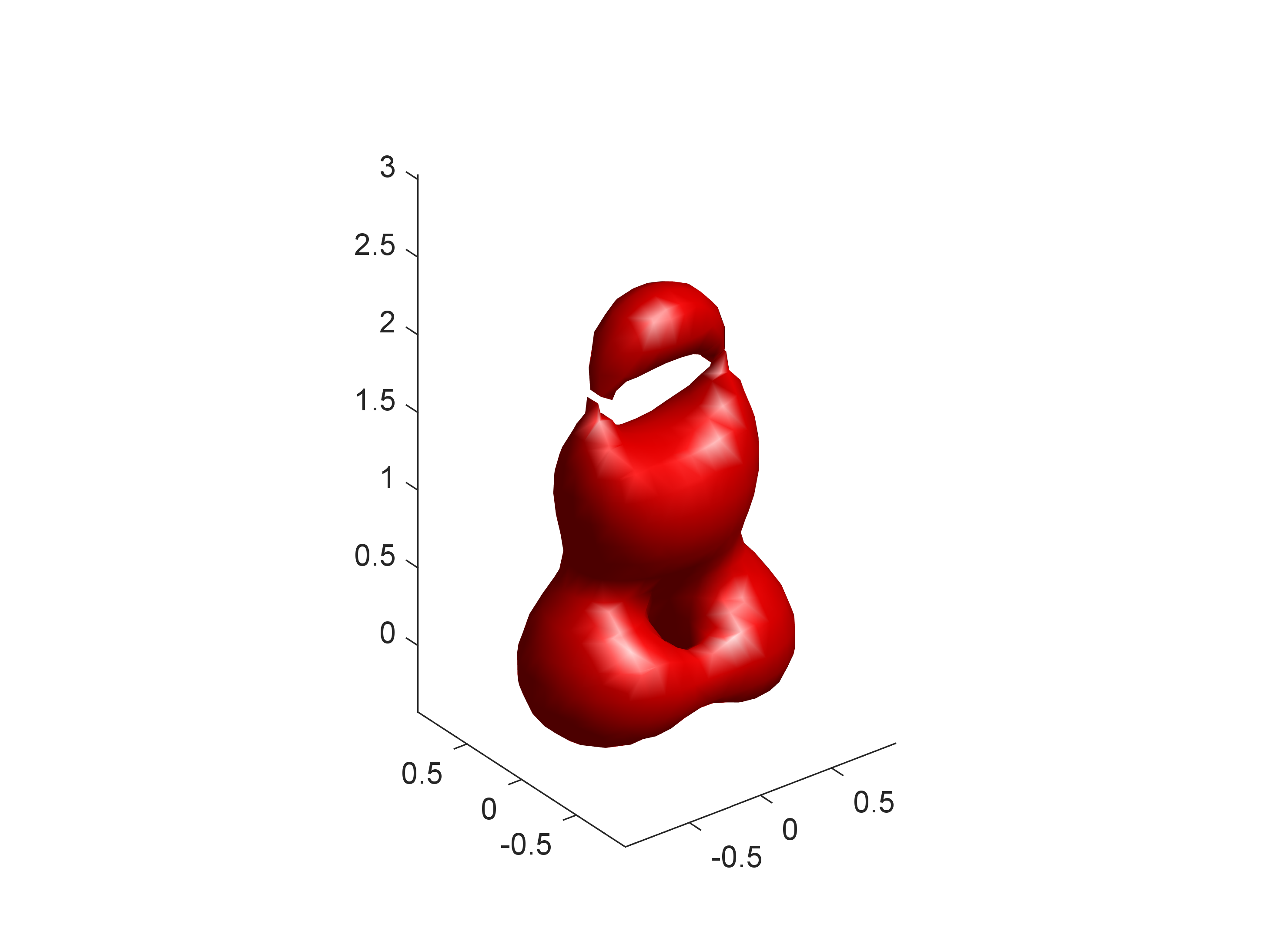}
         \caption{}
         \label{subfig:tau0-mon t=0,k2=0.7}
     \end{subfigure}
     \begin{subfigure}[c]{0.24\textwidth}
         \centering
         \includegraphics[trim= 70 15 70 40,clip,width=\textwidth]{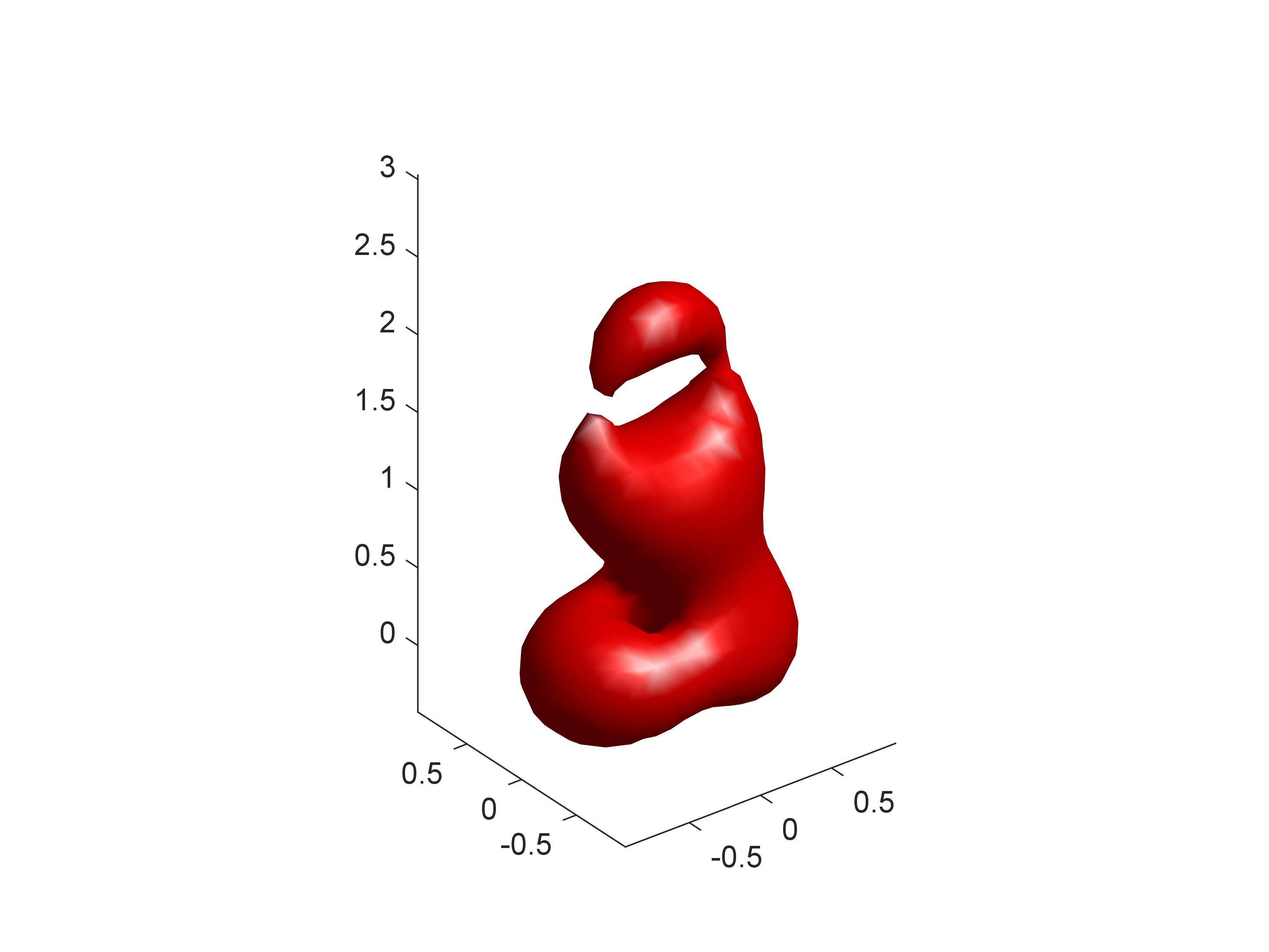}
         \caption{}
         \label{subfig:tau0-mon t=0.225,k2=0.7}
     \end{subfigure}
     \begin{subfigure}[c]{0.24\textwidth}
         \centering
         \includegraphics[trim= 70 15 70 40,clip,width=\textwidth]{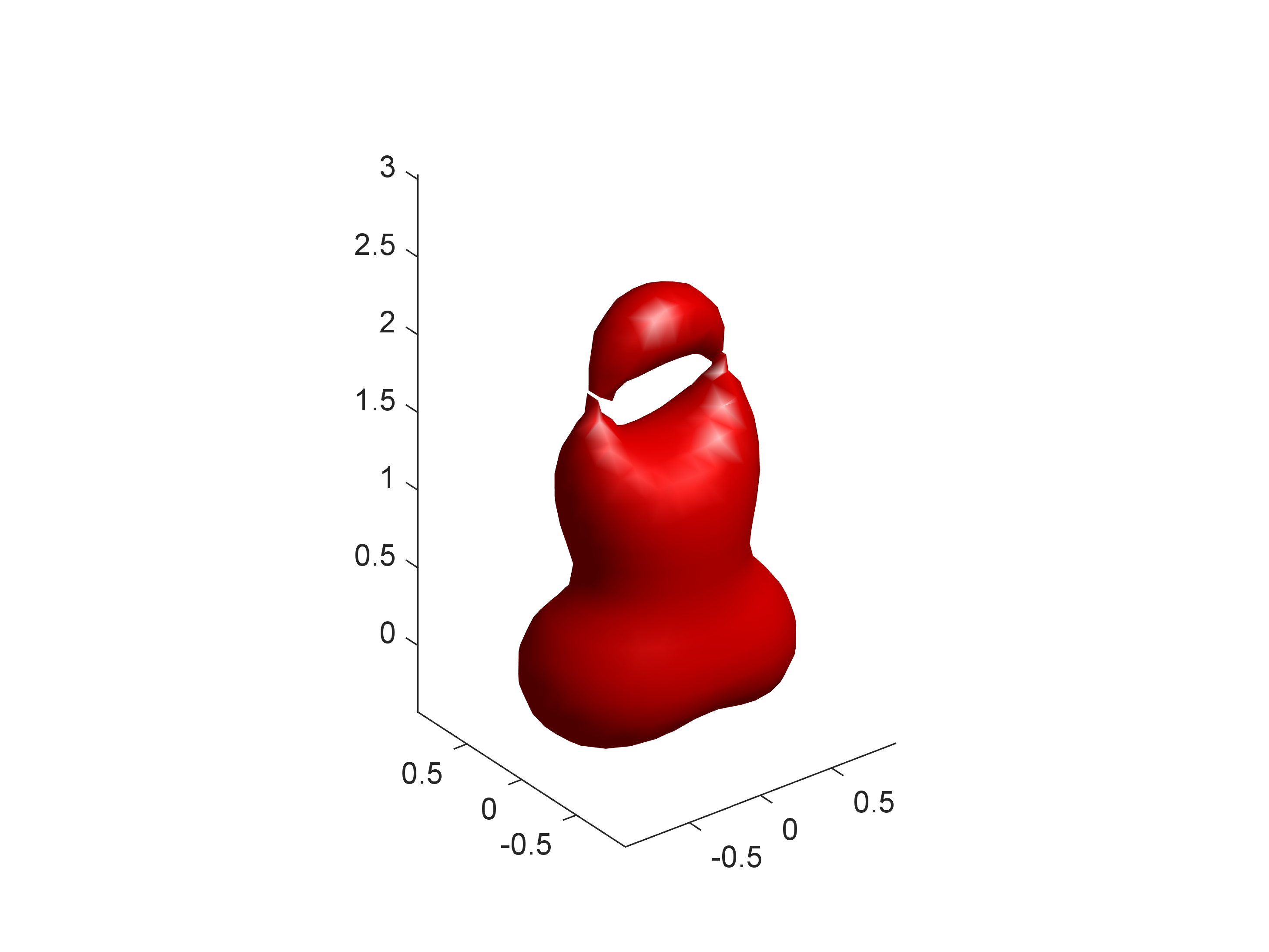}
         \caption{}
         \label{subfig:tau0-mon t=0.5,k2=0.7}
     \end{subfigure}
     \begin{subfigure}[c]{0.24\textwidth}
         \centering
         \includegraphics[trim= 70 15 70 40,clip,width=\textwidth]{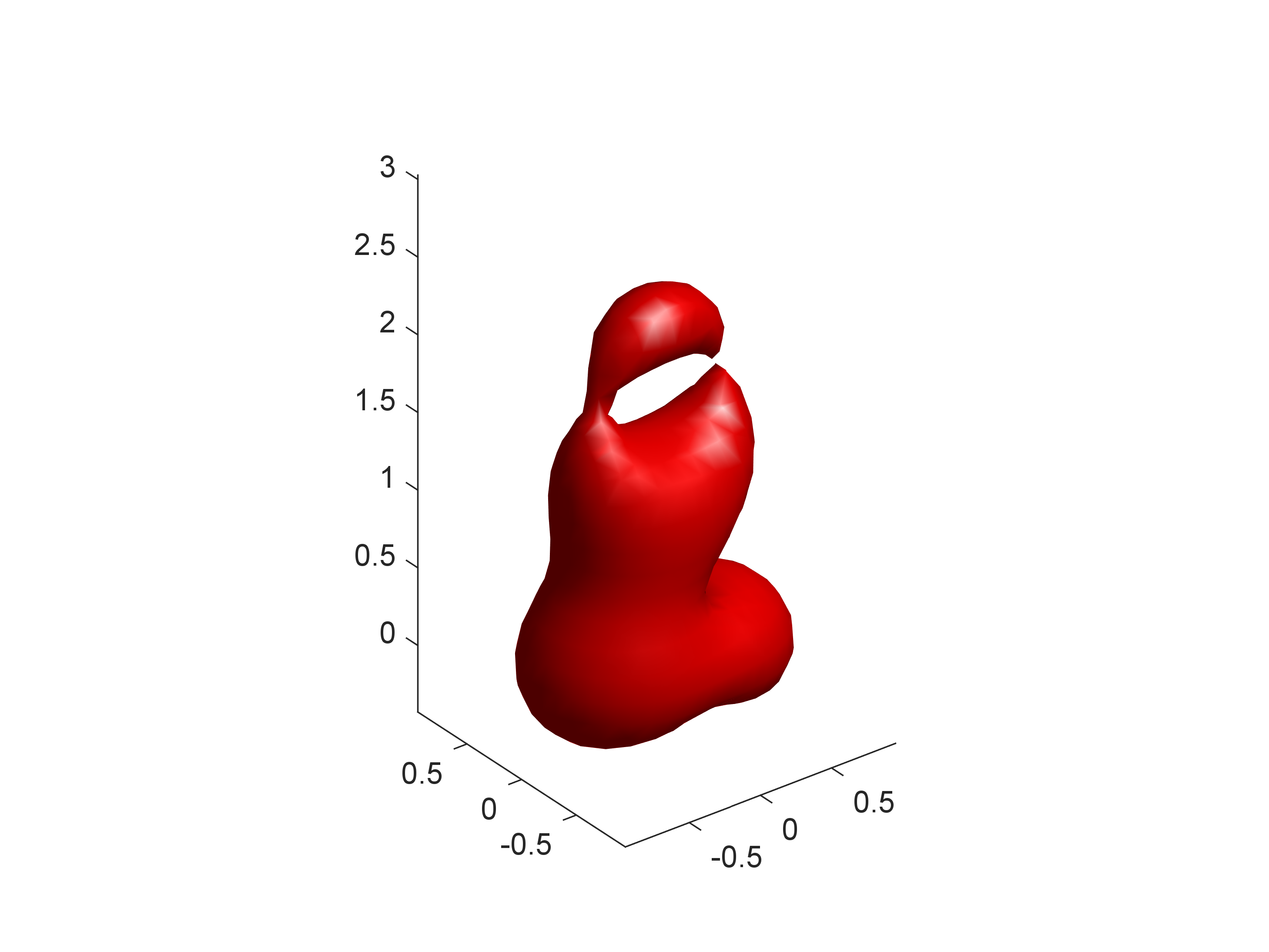}
         \caption{}
         \label{subfig:tau0-mon t=0.725,k2=0.7}
     \end{subfigure}
          \begin{subfigure}[c]{0.24\textwidth}
         \centering
         \includegraphics[trim= 70 15 70 40,clip,width=\textwidth]{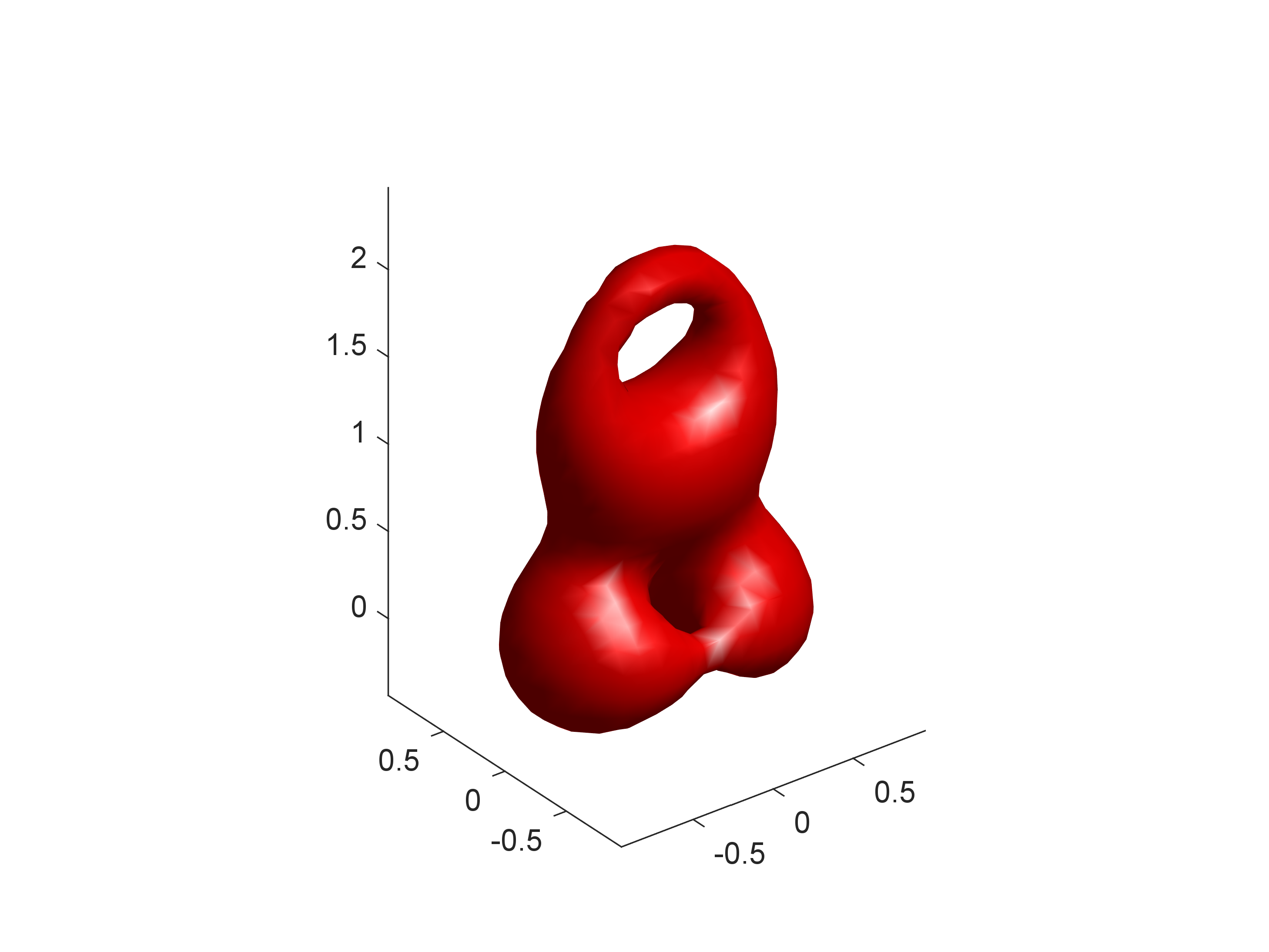}
         \caption{}
         \label{subfig:tau0-mon t=0,k2=0.5}
     \end{subfigure}
     \begin{subfigure}[c]{0.24\textwidth}
         \centering
         \includegraphics[trim= 70 15 70 40,clip,width=\textwidth]{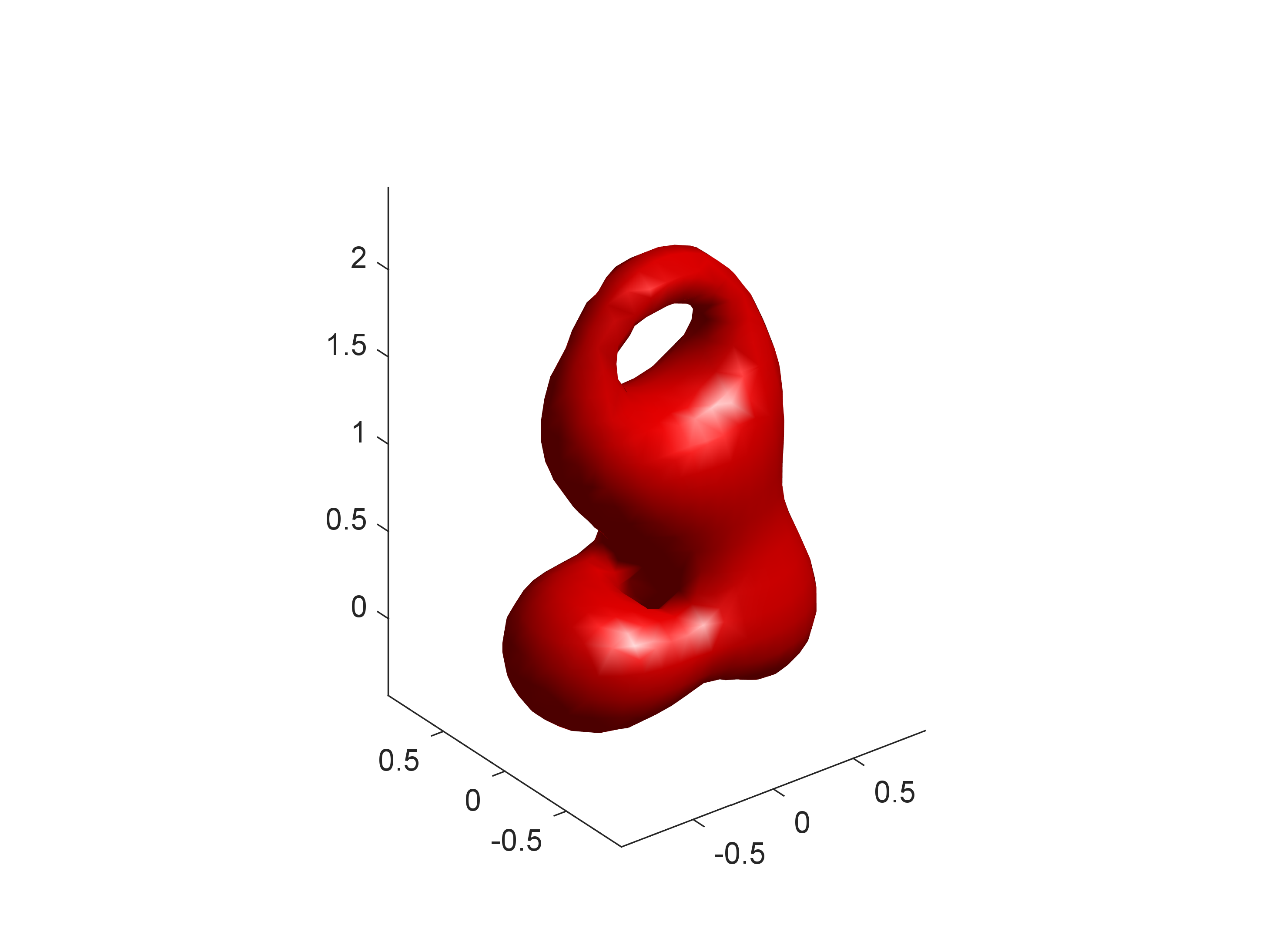}
         \caption{}
         \label{subfig:tau0-mon t=0.225,k2=0.5}
     \end{subfigure}
     \begin{subfigure}[c]{0.24\textwidth}
         \centering
         \includegraphics[trim= 70 15 70 40,clip,width=\textwidth]{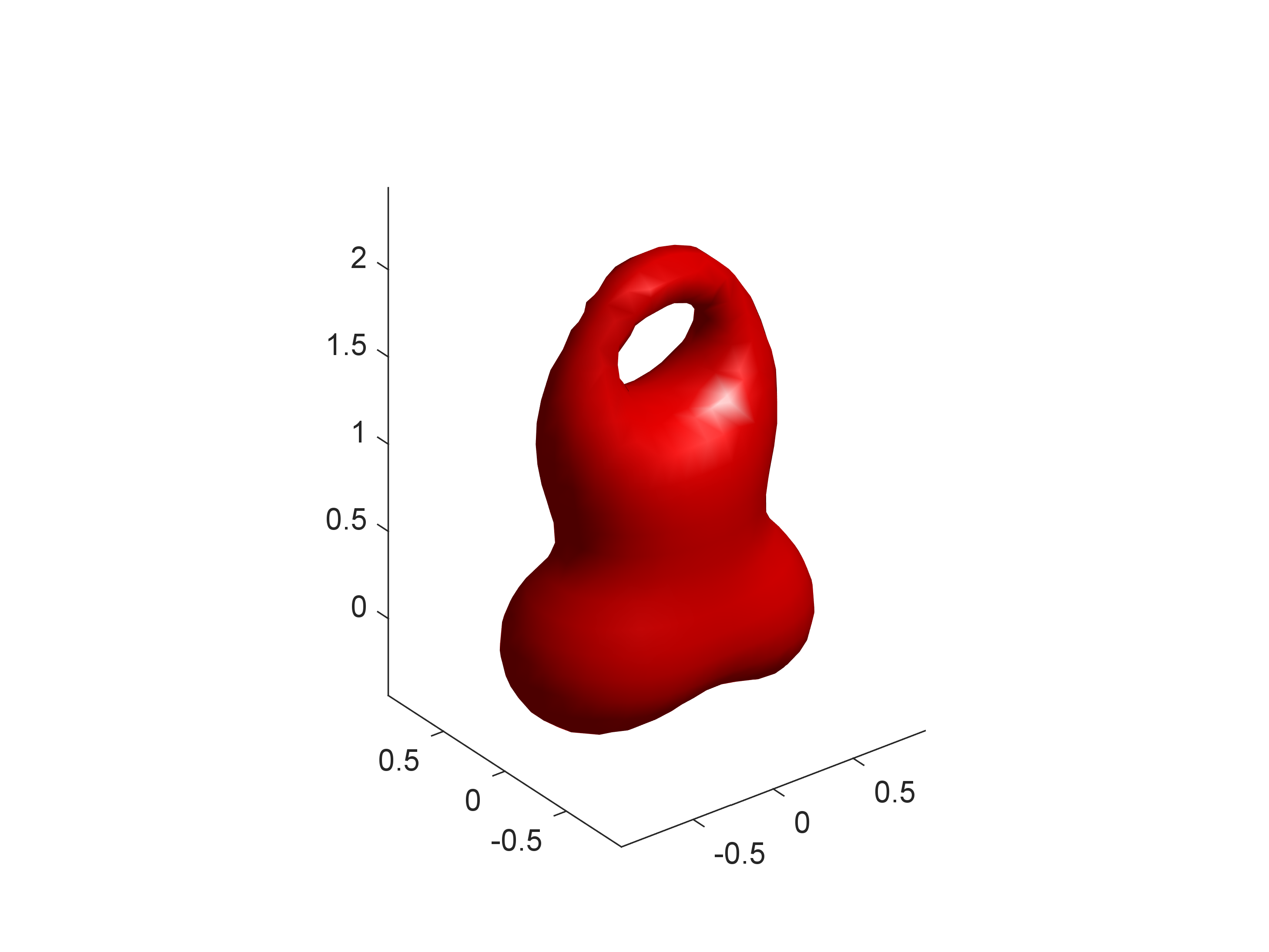}
         \caption{}
         \label{subfig:tau0-mon t=0.5,k2=0.5}
     \end{subfigure}
     \begin{subfigure}[c]{0.24\textwidth}
         \centering
         \includegraphics[trim= 70 15 70 40,clip,width=\textwidth]{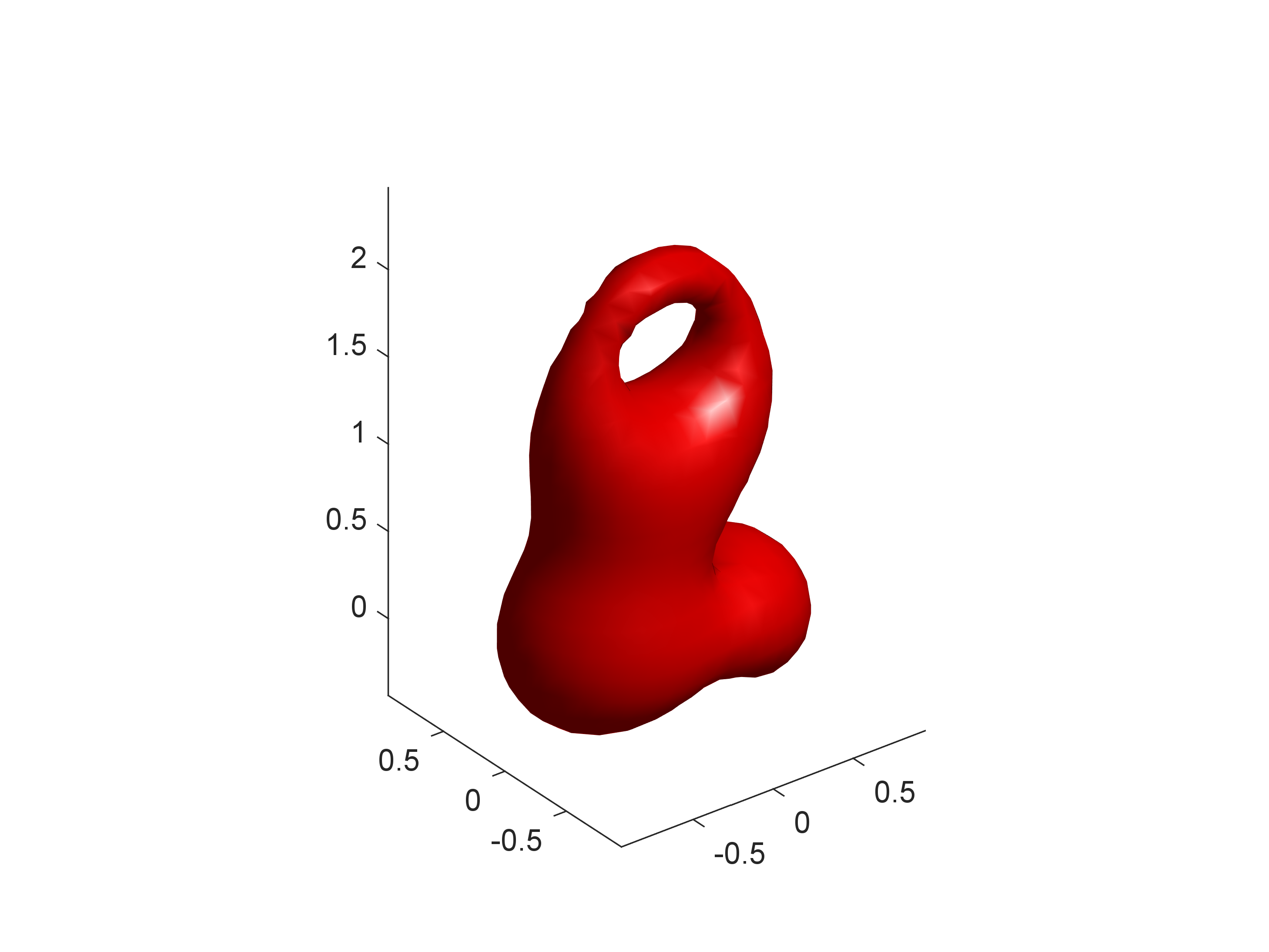}
         \caption{}
         \label{subfig:tau0-mon t=0.725,k2=0.5}
     \end{subfigure}
    \caption{Yang--Mills action density isosurface plots of a monopole-like family within the subfamily of $\tau=0$ rotating calorons. The plots in each column are made at fixed time slices $t=0$, $t=0.225$, $t=0.5$, and $t=0.725$ respectively. Plots (a)--(d), (e)--(h), (i)--(l), and (m)--(p) show the configurations with $\kappa_2=0.99,$ $0.9,$ $0.7$, and $0.5$ respectively. $\kappa_1$ is fixed by the condition $K(\kappa_1)-K(\kappa_2)/\rho=1/30$ giving $\kappa_1=0.30$, $0.58$, $0.72$, and $0.77$ respectively to $2$ decimal places. The pictures are qualitatively the same for all $\kappa_2\leq0.5$, and therefore omitted.}
    \label{fig:monopole-like-family-tau0}
\end{figure}

The instanton-like regime may be visualised by constructing configurations with small scale. An example is plotted in Figure \ref{fig:instanton-like-config-tau0} with $(\kappa_1,\kappa_2)=(0.99,0.5)$ yielding $(D_1,D_2)=(0.08,0.09)$ and scale $\lambda=0.43$. This configuration has no distinctive constituent monopole structure, rather the instantons appear as spherical lumps of maximal action density at isolated spacetime locations. 
Numerically the locations are found to be $(t,x_1,x_2,x_3)\approx(0.25,0.04,0,0)$ and $(t,x_1,x_2,x_3)\approx(0.75,-0.04,0,0)$. We remark that the $x_1$-coordinates are approximately $\pm\frac{D_1}{2}$.
\begin{figure}
    \centering
     \begin{subfigure}[c]{0.3\textwidth}
         \centering
         \includegraphics[trim= 70 15 70 40,clip,width=\textwidth]{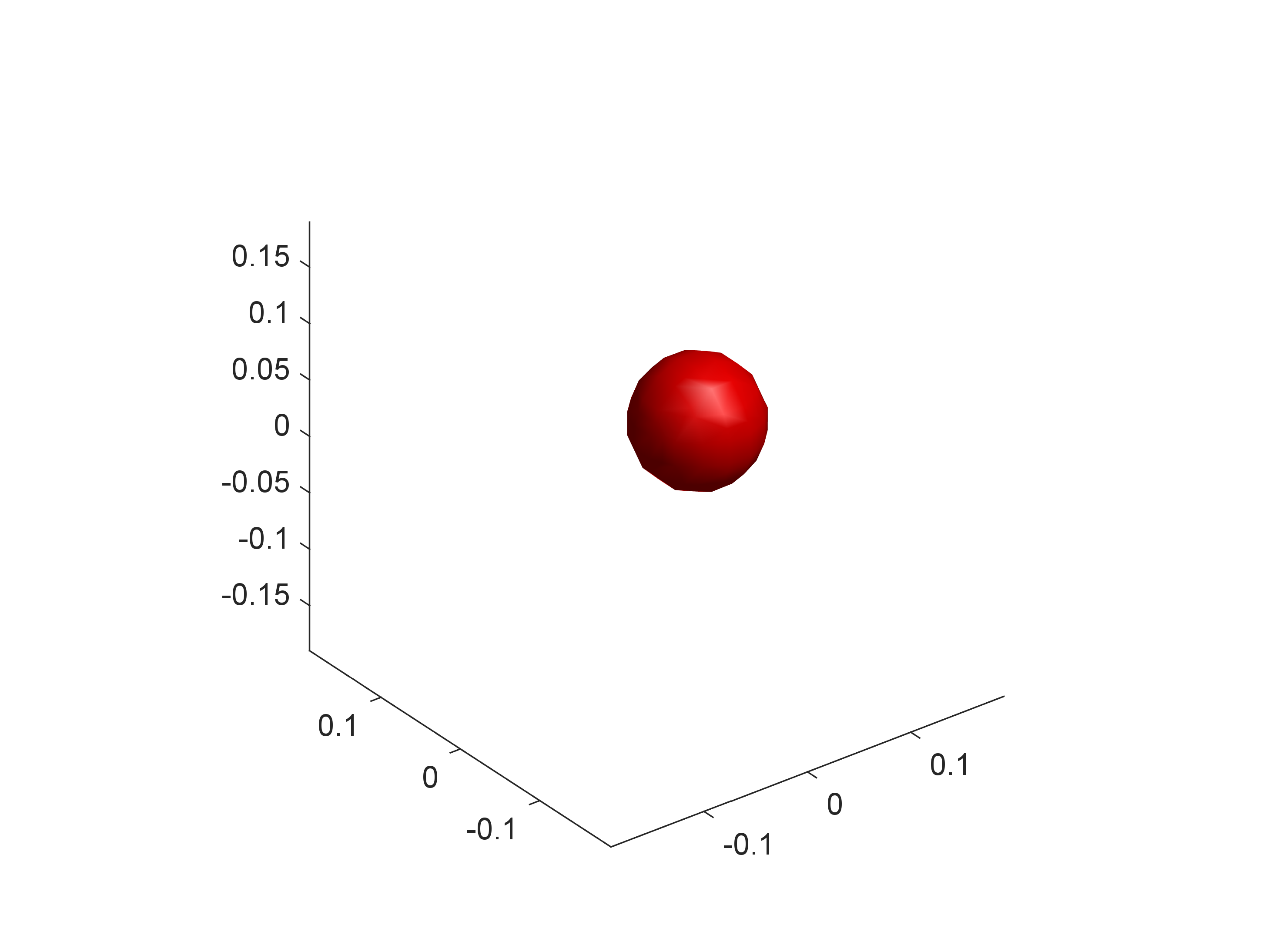}
         \caption{}
         \label{subfig:tau0-inst t=0.25}
     \end{subfigure}
     \begin{subfigure}[c]{0.3\textwidth}
         \centering
         \includegraphics[trim= 70 15 70 40,clip,width=\textwidth]{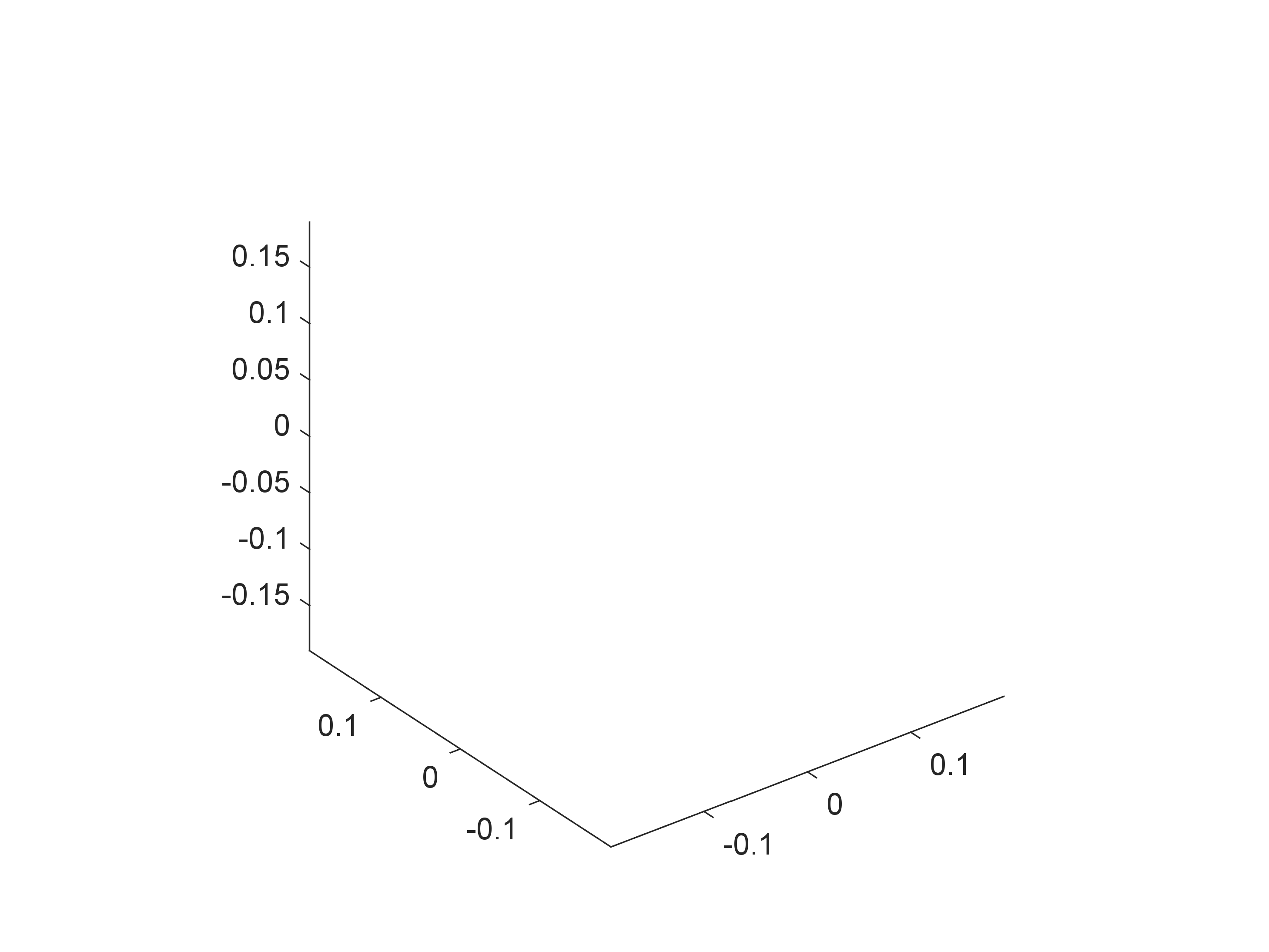}
         \caption{}
         \label{subfig:tau0-inst t=0.5}
     \end{subfigure}
     \begin{subfigure}[c]{0.3\textwidth}
         \centering
         \includegraphics[trim= 70 15 70 40,clip,width=\textwidth]{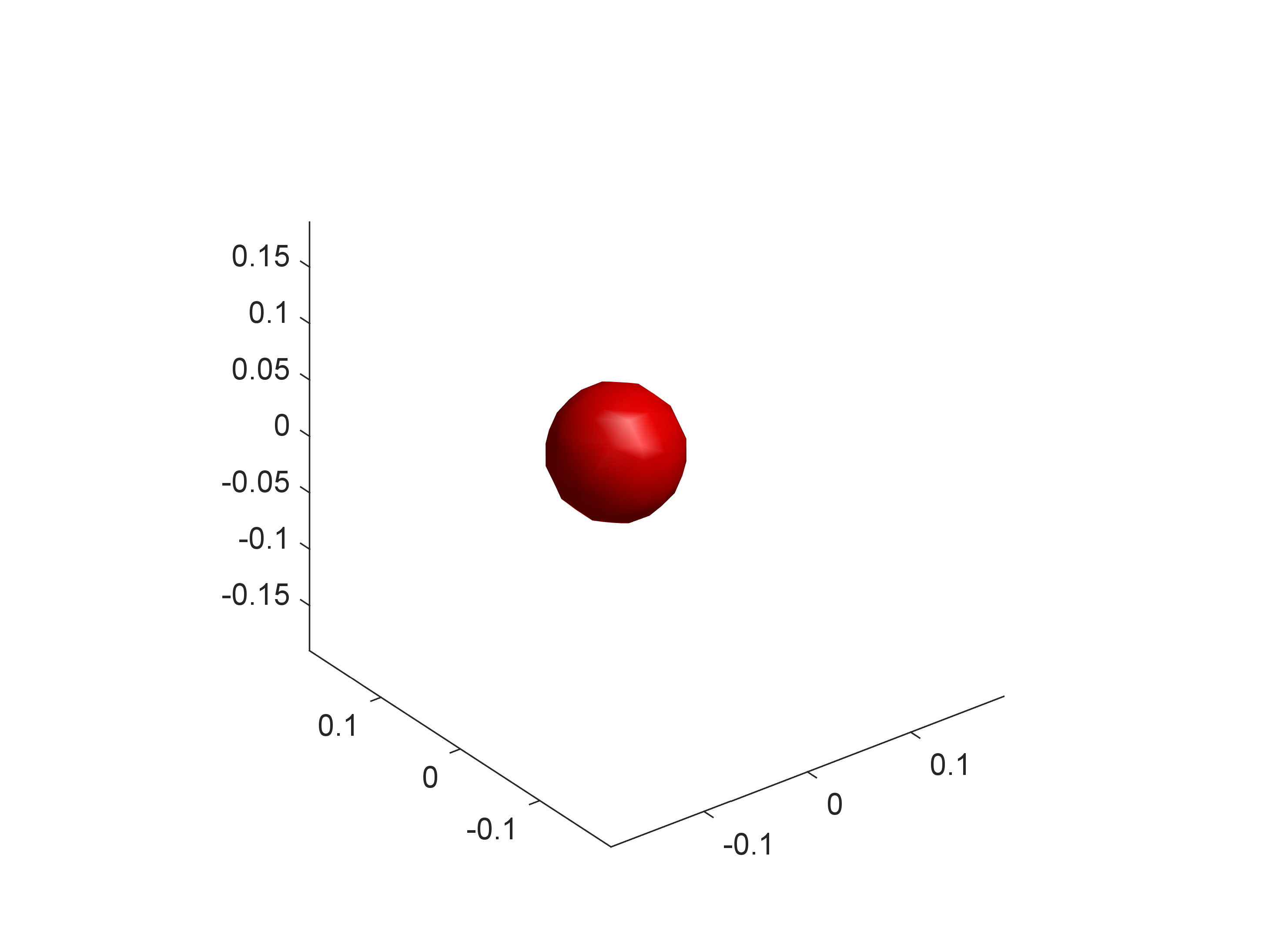}
         \caption{}
         \label{subfig:tau0-inst t=0.75}
     \end{subfigure}
    \caption{Yang--Mills action density isosurface plots of an instanton-like configuration within the subfamily of $\tau=0$ rotating calorons. Plots (a)--(c) are made at fixed time slices $t=0.25$, $t=0.5$, and $t=0.75$, respectively. The configuration has parameters $(\kappa_1,\kappa_2)=(0.99,0.5)$, giving scale $\lambda=0.43$.}
    \label{fig:instanton-like-config-tau0}
\end{figure}

\subsection{The case \texorpdfstring{$\kappa_1=1$}{kappa1=1}}
When $\kappa_1=1$, the equation \eqref{match_Ds-ks} immediately gives \begin{align}\label{D1-kappa1=1}
D_1=D_2(\kappa_2'\scn_{\kappa_2}(D_2\mu)+\kappa_2'\nc_{\kappa_2}(D_2\mu)).
\end{align}
Moreover, when $\kappa_1=1$ half of the data \eqref{sol-deformed} is constant.\footnote{Note that this is not the case for the caloron Nahm data in the periodic gauge \eqref{caloron-Nahm-from-delay-n=2} due to the non-constant transformation $g$.}

Except in the degenerate case $\tau=0$ already considered, as explained in the proof of Theorem \ref{thm.gen-sol}, the second equation \eqref{match-triangle} may be used to determine $\psi_1$ and $\psi_3$, and provides a constraint on $\tau$ given by \eqref{tau-constraint}. With $D_1$ fixed via \eqref{D1-kappa1=1} this constraint simplifies to
\begin{align}\label{tau-constraint-kappa1=1}
    \kappa_2'{\rm sd}_{\kappa_2}(D_2\mu)\leq\sin2\tau\leq\kappa_2'\nd_{\kappa_2}(D_2\mu).
\end{align}
The phases $(\e^{\ii\psi_1},\e^{\ii\psi_3})$ are then found via \eqref{match-triangle} up to a sign choice which we shall fix here by setting $\sin\psi_1\geq0$. As in the $\tau=0$ case, some parameters may be disregarded, namely $\xi$ (as an overall translation), and $\psi_+,\sigma_+,\sigma_-$ (as spatial, global, or gauge phases). This leaves three physical parameters to vary: $\kappa_2$, $D_2$, and $\tau$ (constrained \eqref{tau-constraint-kappa1=1}).

We visualise some of these calorons using our numerical Nahm transform for $\kappa_2\in(0,1)$ and fixing $D_2=0.68K(\kappa_2)/\mu$.
We solve the constraint \eqref{tau-constraint-kappa1=1} by fixing $\tau=0.1\tau_{-}+0.9\tau_{+}$, where $\tau_\pm$ are the minimum and maximum values for $\tau$ from the constraint \eqref{tau-constraint-kappa1=1}.
These choices lead to calorons that are monopole-like.
Plots for this family taken as action density isosurfaces for the fixed time slices $t=0,0.225,0.5$ and $0.725$ are given in Figure \ref{fig: kappa1=1 plots}. Again, the calorons are made up of two constituent 2-monopoles.
These images show that one of the constituents is located on the $x_3$-axis, and the other constituent is separated into two 1-monopoles on either side of the central constituent (approximately aligned with the $x_2$-axis).  As $\kappa_2$ varies the central constituent changes its shape but the outer constituents remain fixed.  When $\kappa_2\approx 1$ the central constituent is a pair of 1-monopoles separated on the $x_3$-axis, while when $\kappa_2\approx0$ it is a toroidal 2-monopole.
These schematic qualitative properties can also be deduced from the Nahm data, which we discuss below in Section \ref{sec-parameters}.
\begin{figure}
    \centering
    \begin{subfigure}[c]{0.24\textwidth}
         \centering
         \includegraphics[trim= 70 15 70 40,clip,width=\textwidth]{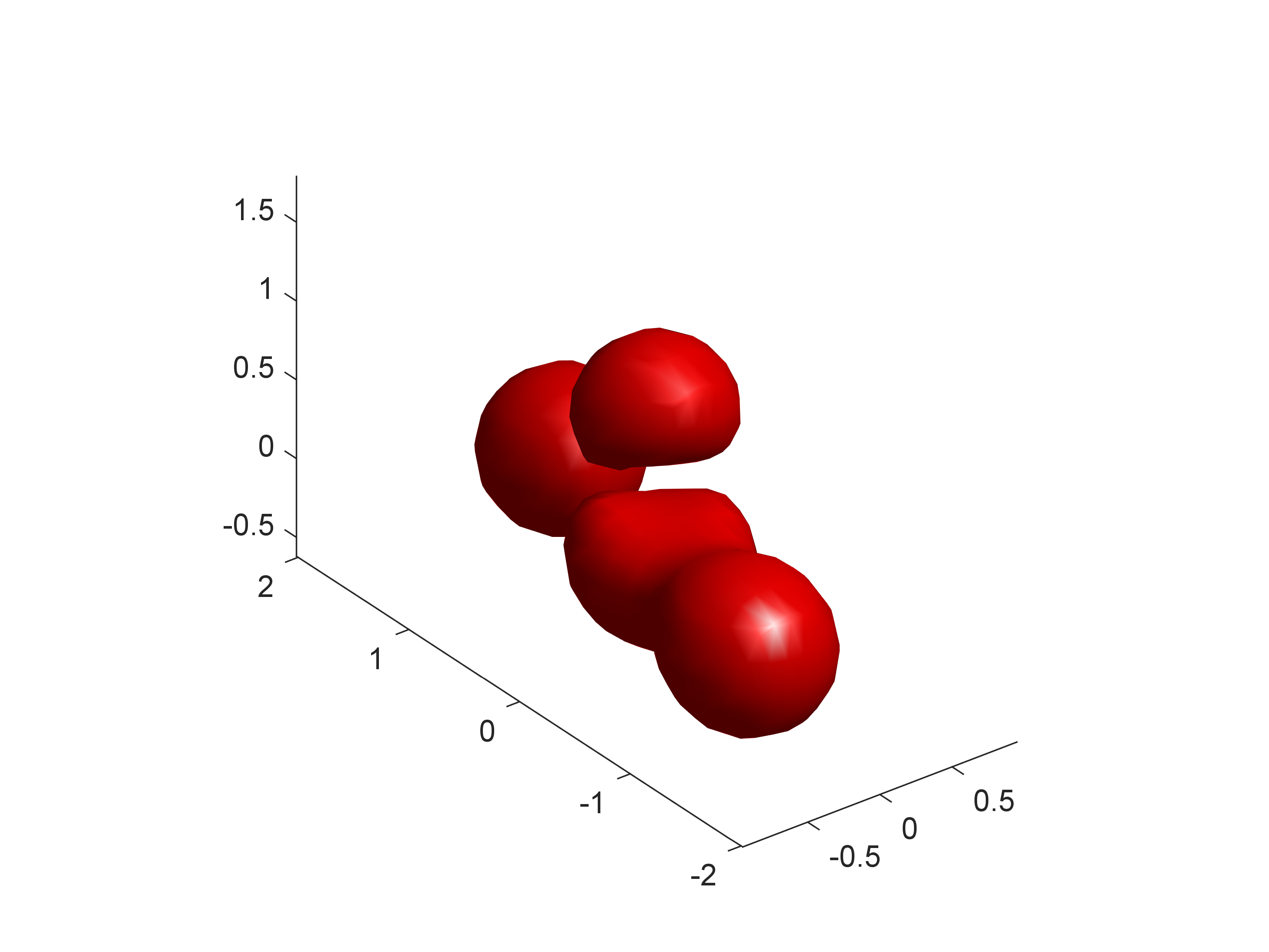}
         \caption{}
         \label{subfig:kappa1_k2=0.9t=0}
     \end{subfigure}
          \begin{subfigure}[c]{0.24\textwidth}
         \centering
         \includegraphics[trim= 70 15 70 40,clip,width=\textwidth]{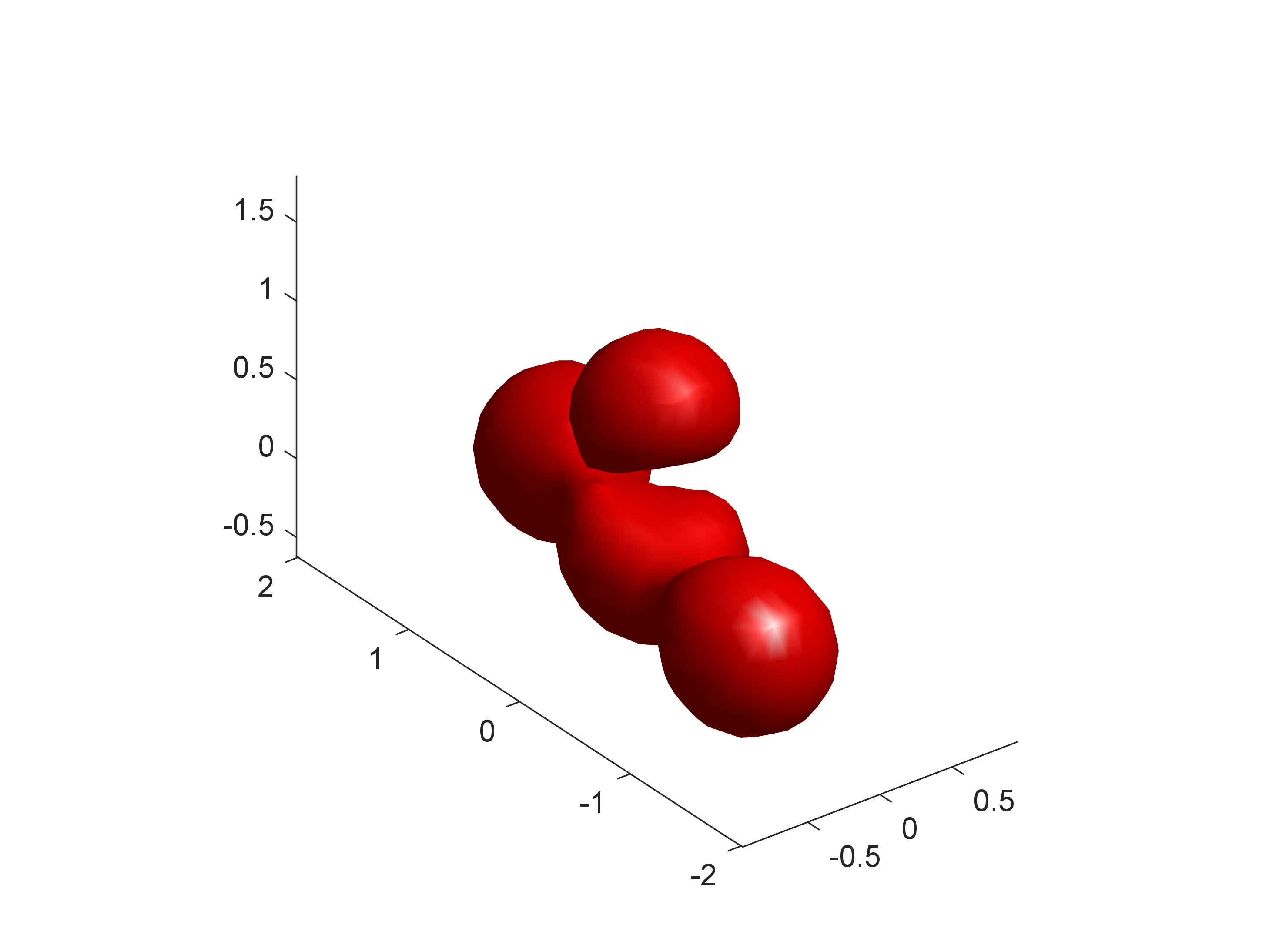}
         \caption{}
         \label{subfig:kappa1_k2=0.9t=0.225}
     \end{subfigure}
          \begin{subfigure}[c]{0.24\textwidth}
         \centering
         \includegraphics[trim= 70 15 70 40,clip,width=\textwidth]{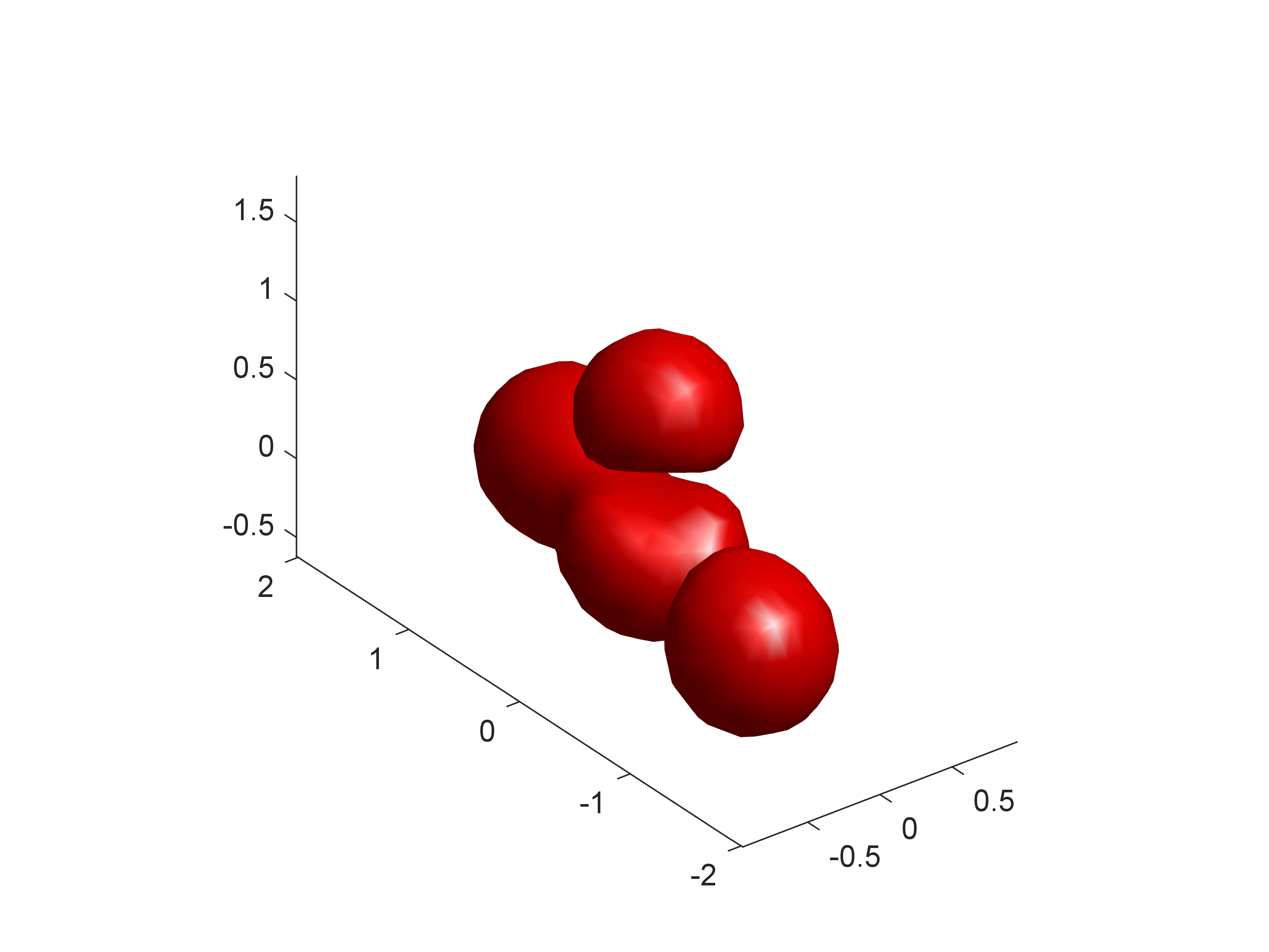}
         \caption{}
         \label{subfig:kappa1_k2=0.9t=0.5}
     \end{subfigure}
          \begin{subfigure}[c]{0.24\textwidth}
         \centering
         \includegraphics[trim= 70 15 70 40,clip,width=\textwidth]{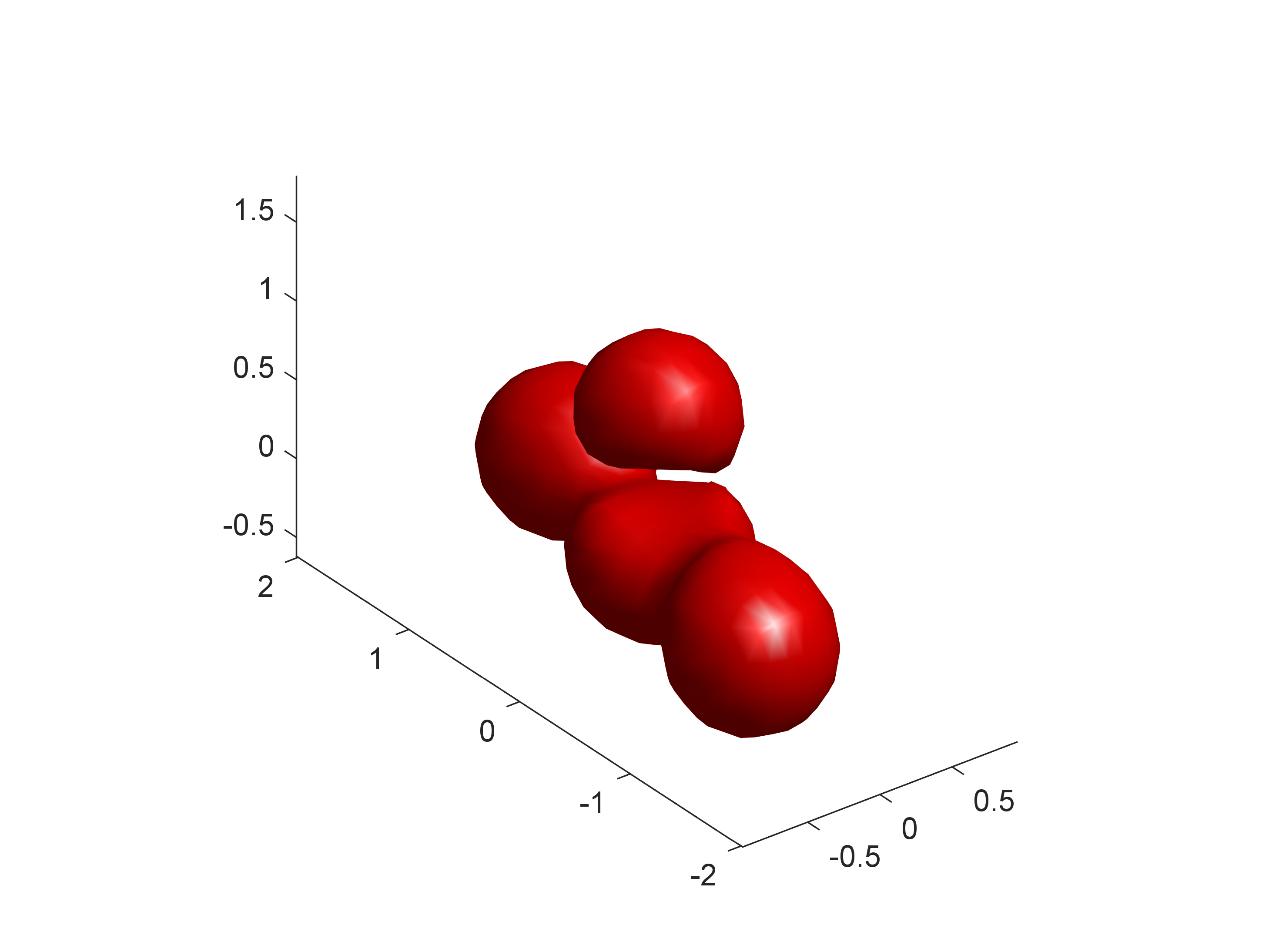}
         \caption{}
         \label{subfig:kappa1_k2=0.9t=0.725}
     \end{subfigure}
    \begin{subfigure}[c]{0.24\textwidth}
         \centering
         \includegraphics[trim= 70 15 70 40,clip,width=\textwidth]{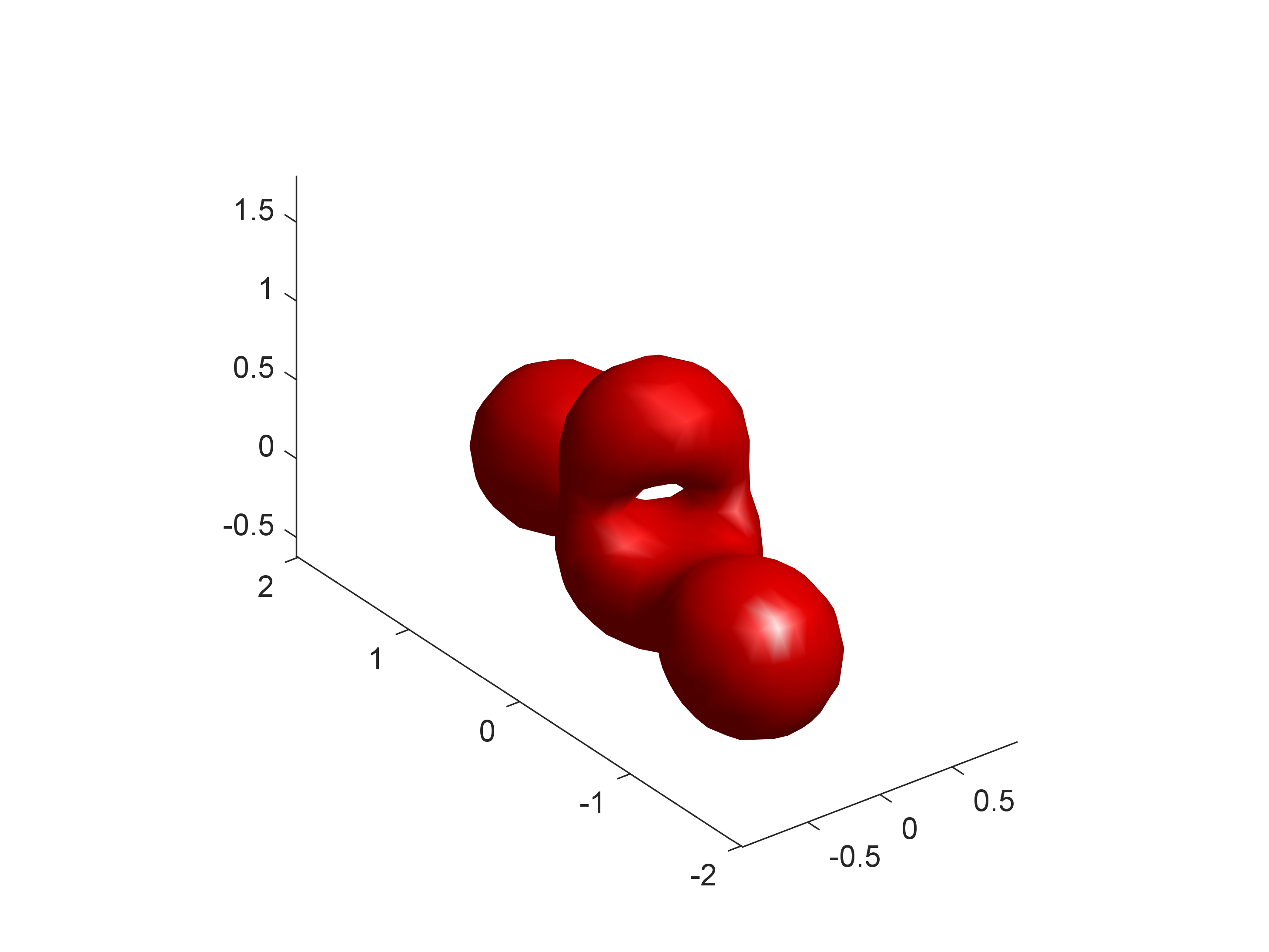}
         \caption{}
         \label{subfig:kappa1_k2=0.8t=0}
     \end{subfigure}
          \begin{subfigure}[c]{0.24\textwidth}
         \centering
         \includegraphics[trim= 70 15 70 40,clip,width=\textwidth]{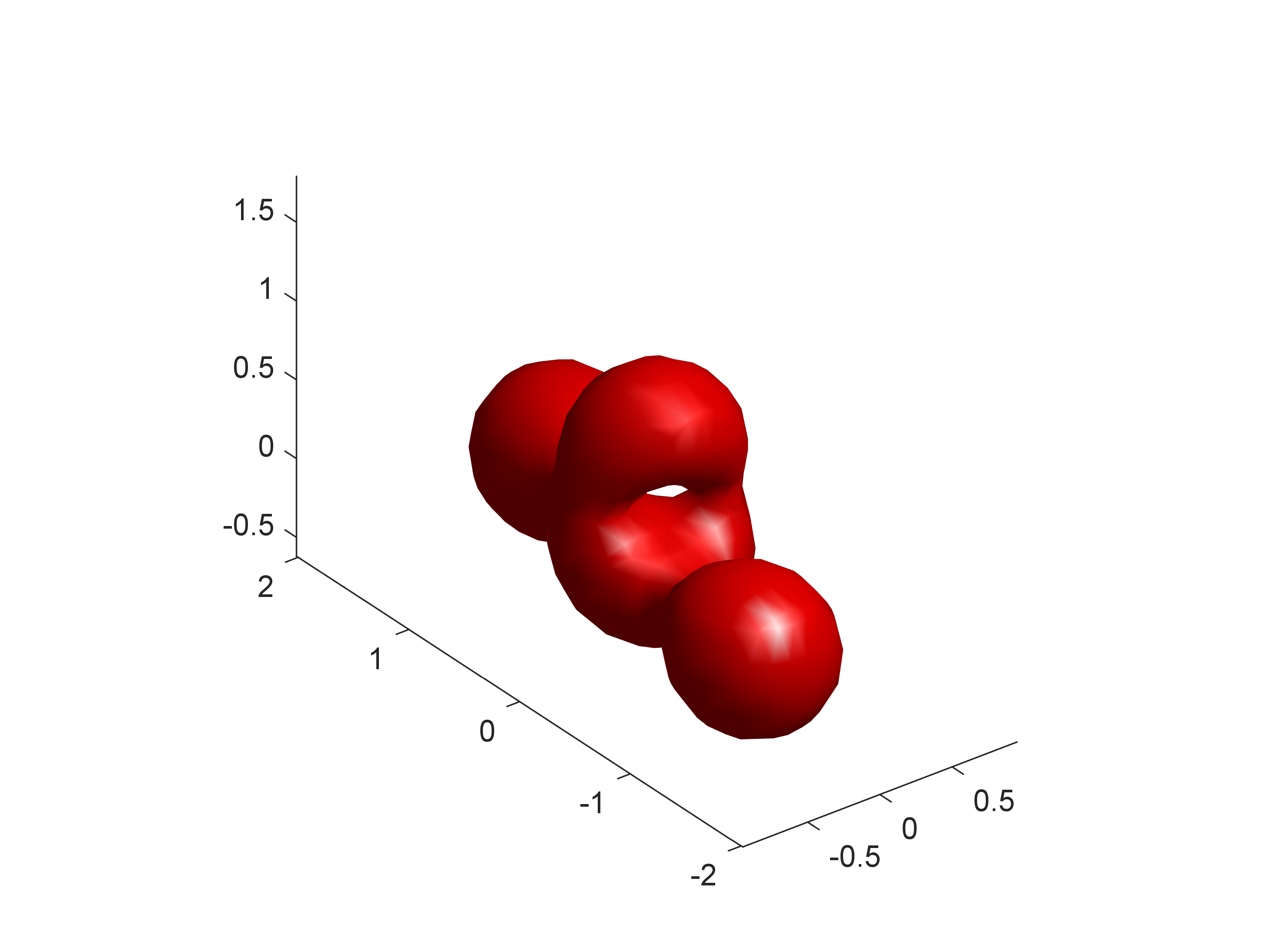}
         \caption{}
         \label{subfig:kappa1_k2=0.8t=0.225}
     \end{subfigure}
          \begin{subfigure}[c]{0.24\textwidth}
         \centering
         \includegraphics[trim= 70 15 70 40,clip,width=\textwidth]{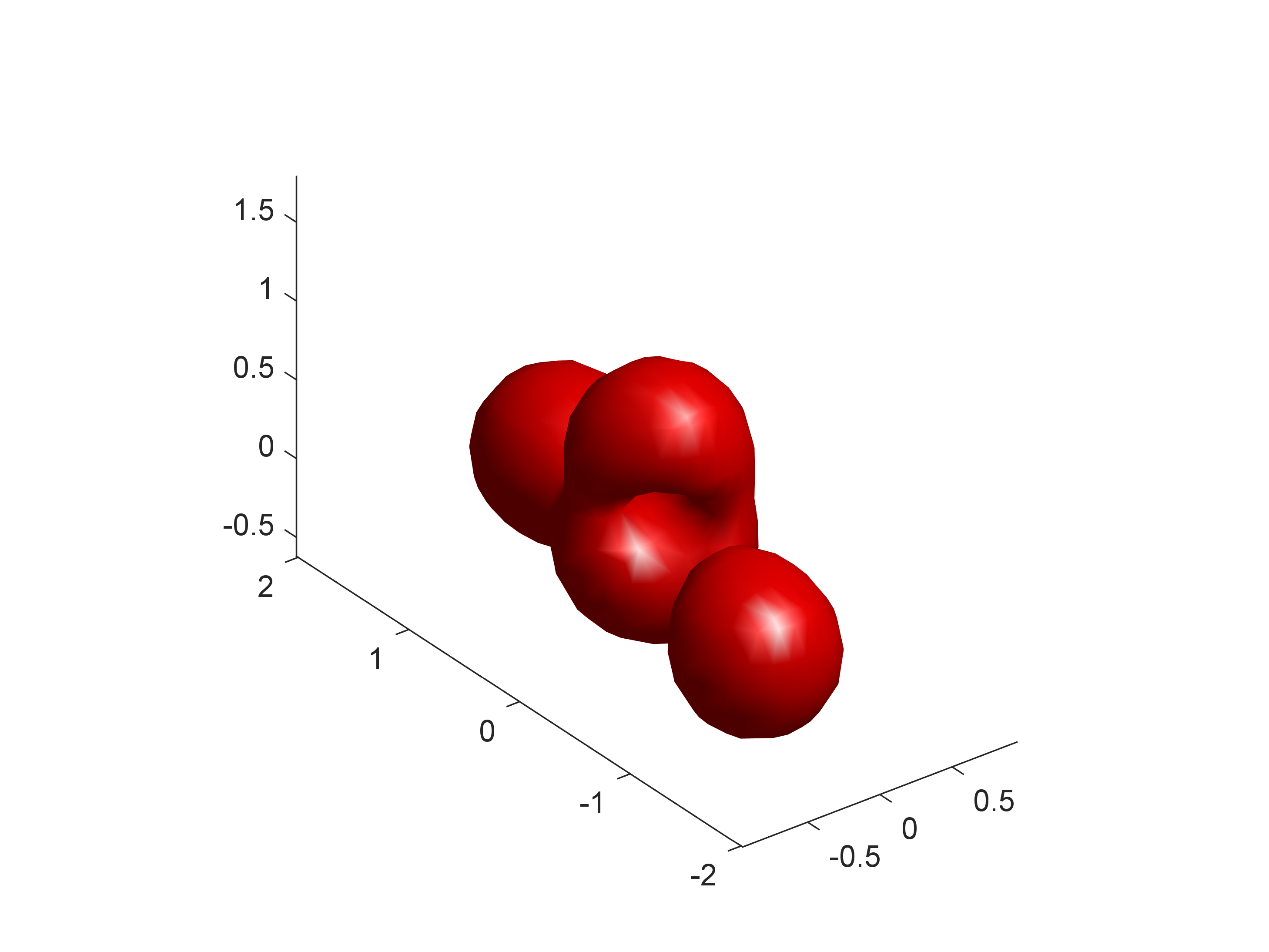}
         \caption{}
         \label{subfig:kappa1_k2=0.8t=0.5}
     \end{subfigure}
          \begin{subfigure}[c]{0.24\textwidth}
         \centering
         \includegraphics[trim= 70 15 70 40,clip,width=\textwidth]{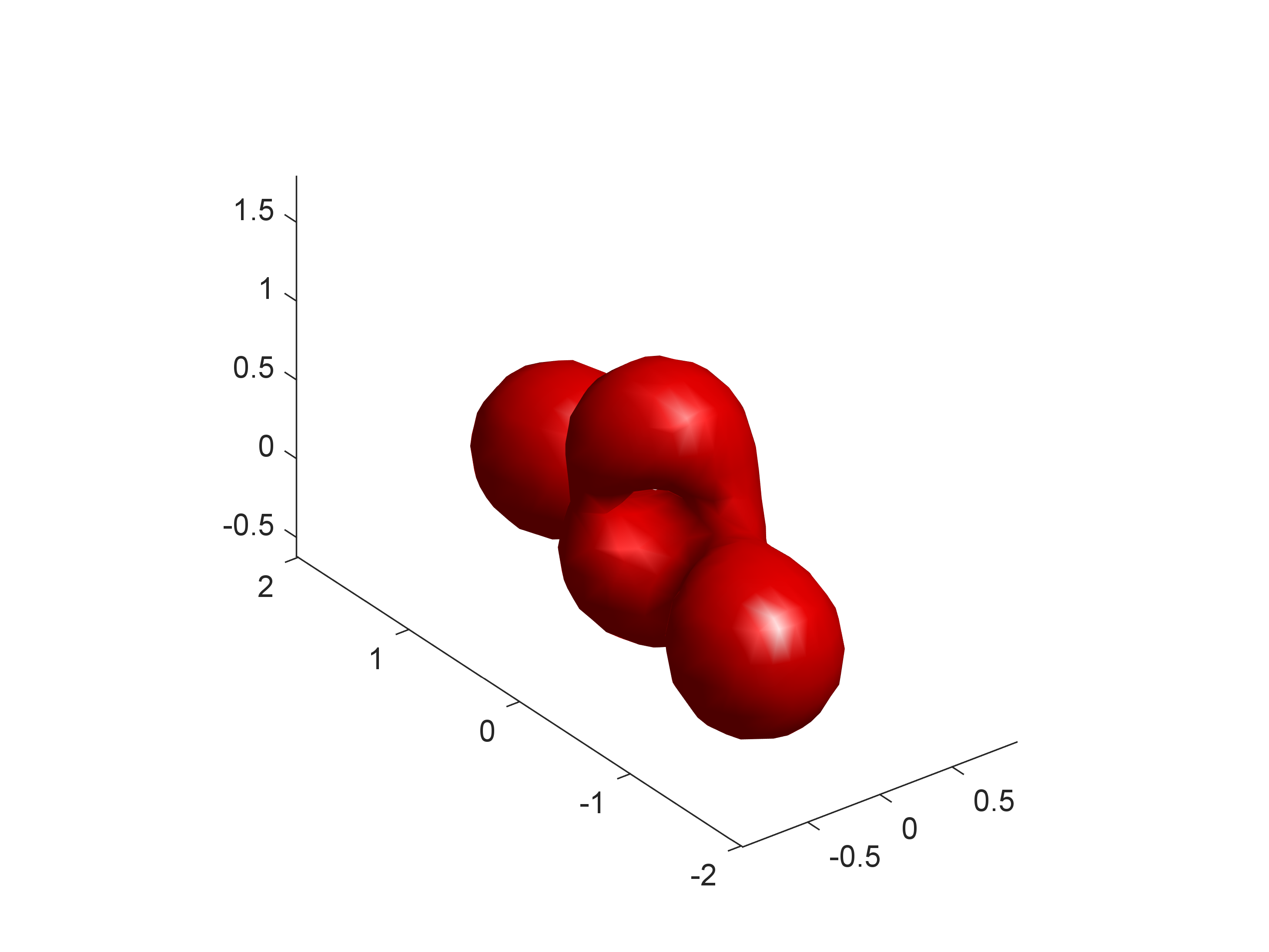}
         \caption{}
         \label{subfig:kappa1_k2=0.8t=0.725}
     \end{subfigure}
    \begin{subfigure}[c]{0.24\textwidth}
         \centering
         \includegraphics[trim= 70 15 70 40,clip,width=\textwidth]{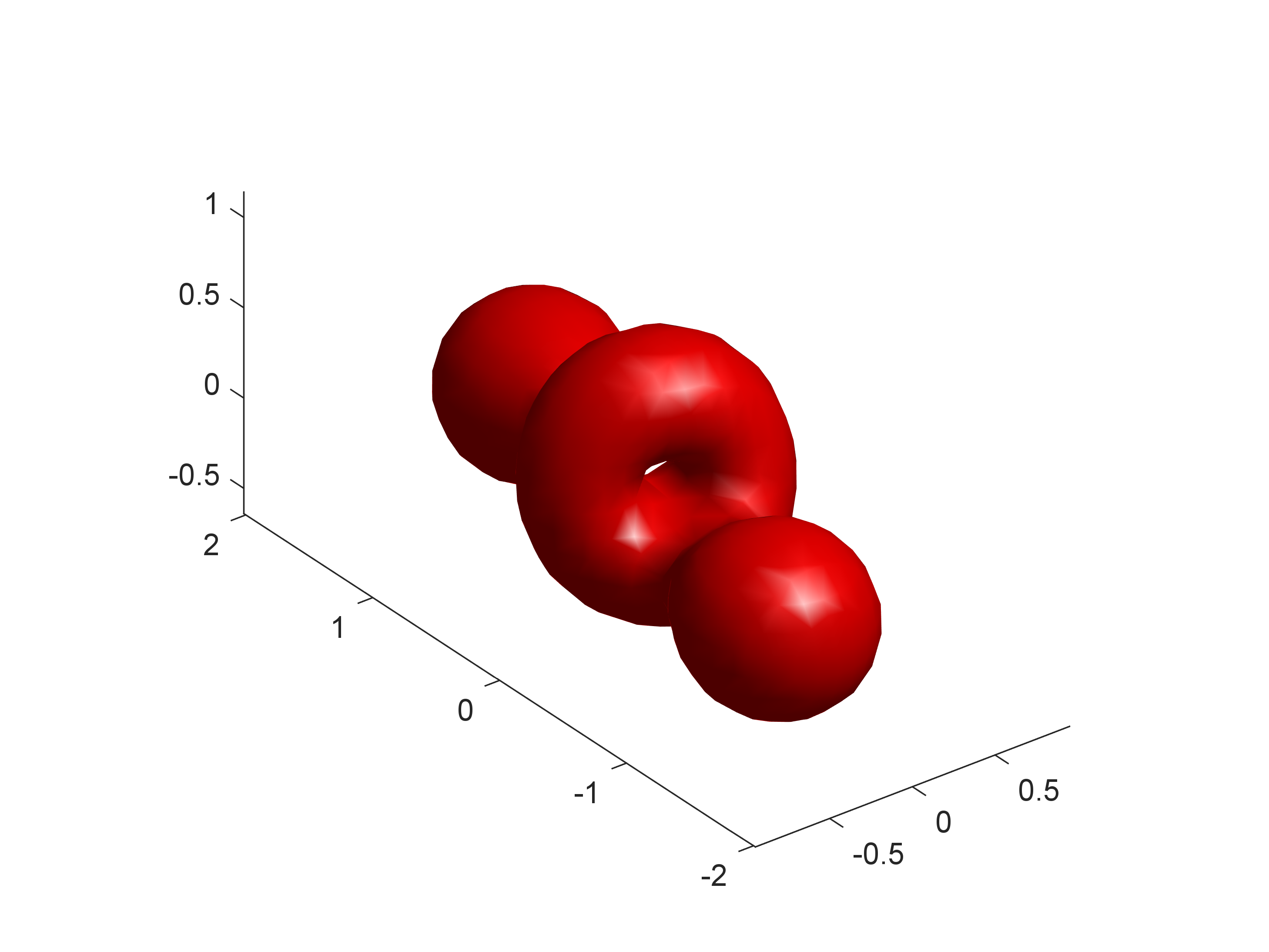}
         \caption{}
         \label{subfig:kappa1-mon-ring t=0}
     \end{subfigure}
          \begin{subfigure}[c]{0.24\textwidth}
         \centering
         \includegraphics[trim= 70 15 70 40,clip,width=\textwidth]{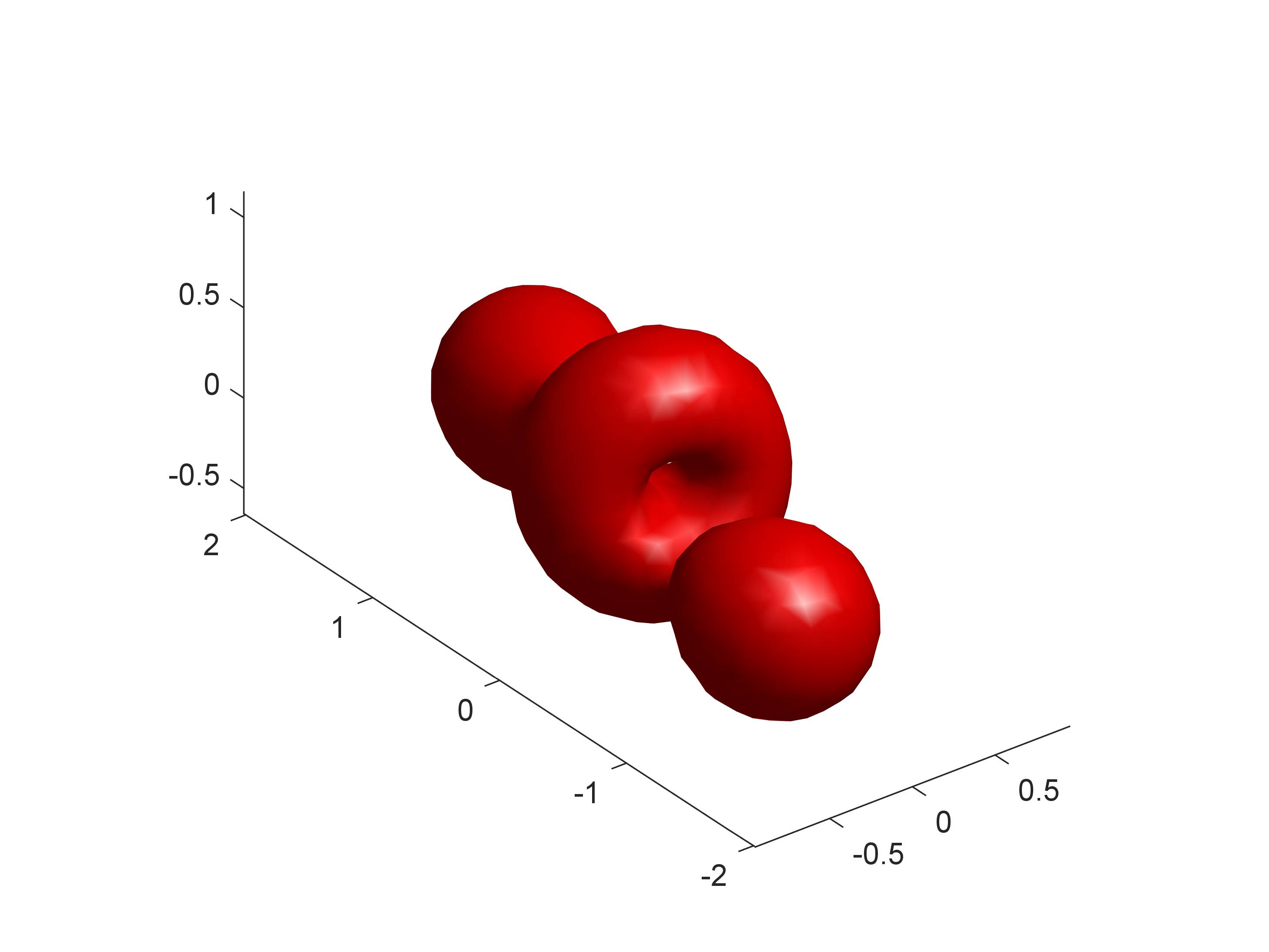}
         \caption{}
         \label{subfig:kappa1-mon-ring t=0.225}
     \end{subfigure}
         \begin{subfigure}[c]{0.24\textwidth}
         \centering
         \includegraphics[trim= 70 15 70 40,clip,width=\textwidth]{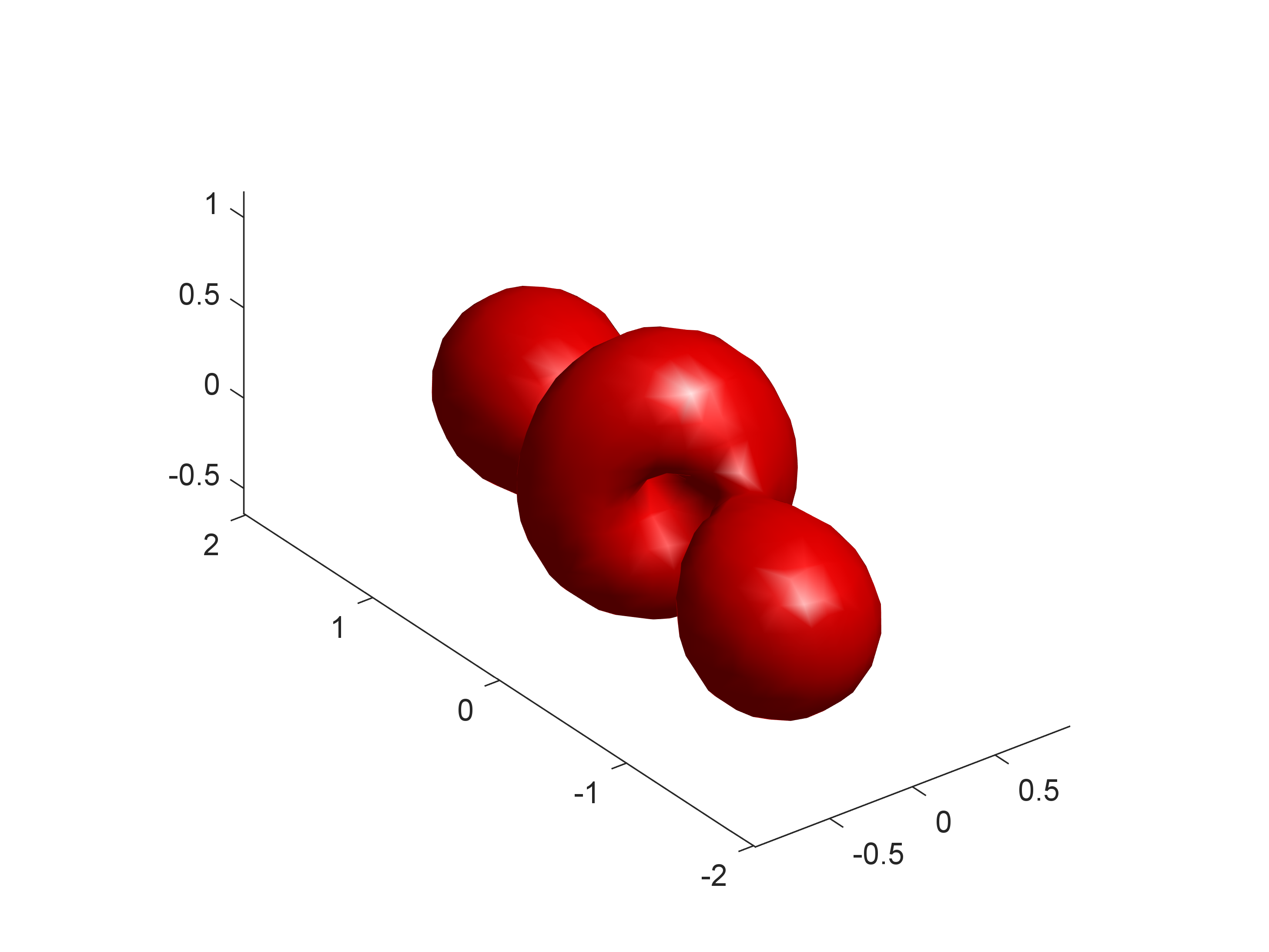}
         \caption{}
         \label{subfig:kappa1-mon-ring t=0.5}
     \end{subfigure}
          \begin{subfigure}[c]{0.24\textwidth}
         \centering
         \includegraphics[trim= 70 15 70 40,clip,width=\textwidth]{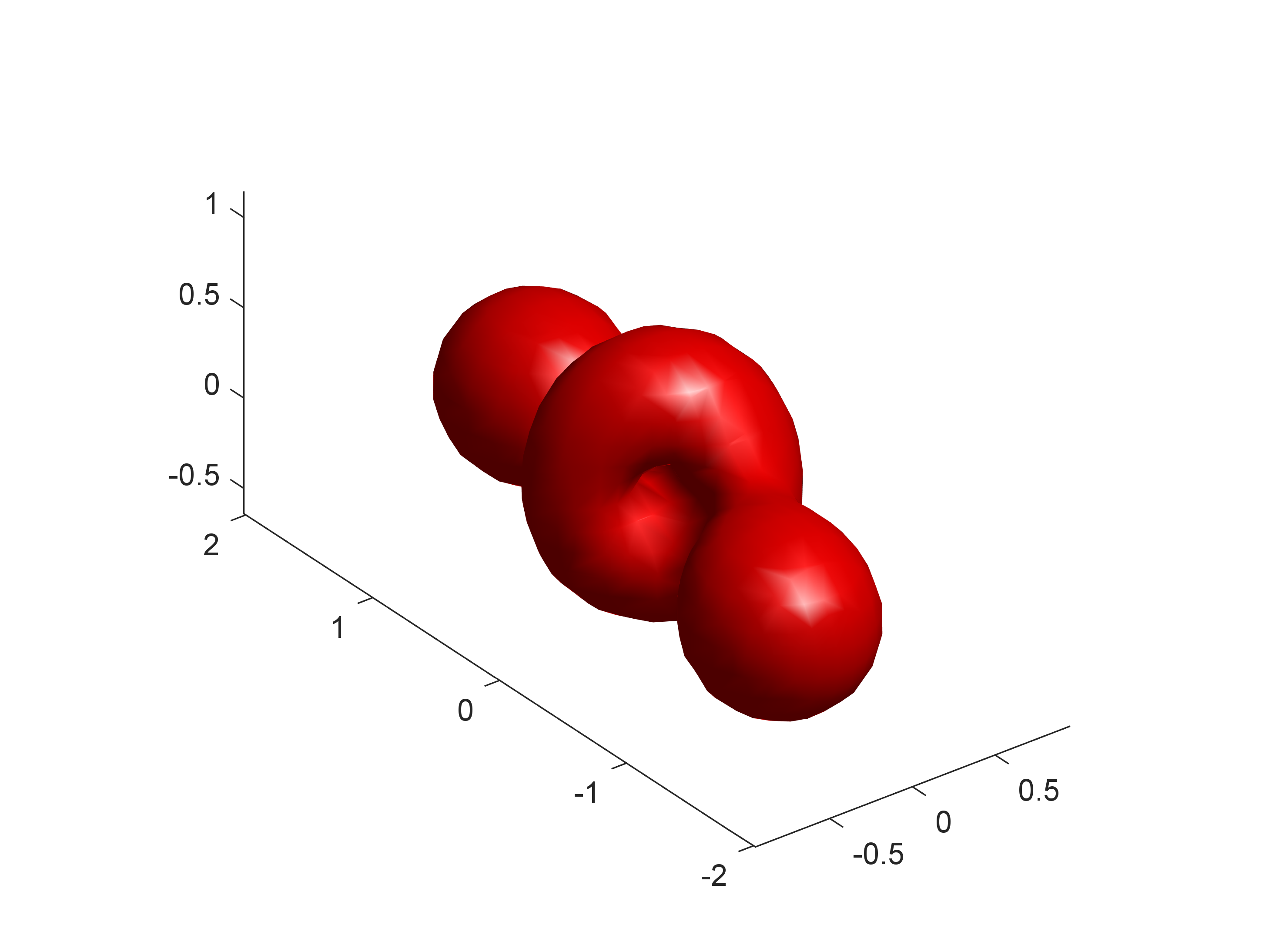}
         \caption{}
         \label{subfig:kappa1-mon-ring t=0.725}
     \end{subfigure}
     \caption{Yang--Mills action density isosurfaces of examples in the $\kappa_1=1$ family of rotating calorons at four fixed time slices (at $t=0,0.225,0.5,0.725$ moving left-to-right). Plots (a)-(d), (e)-(h) and (i)-(l) show the configurations with $\kappa_2=0.9,$ $0.8$, and $0.5$. The other free parameters $(D_2,\tau)$ are determined by $D_2=0.68K(\kappa_2)/\mu$ and $\tau=0.1\tau_-+0.9\tau_+$, with $\tau_\pm$ the bounds of the constraint \eqref{tau-constraint-kappa1=1}. These values give $(D_2,\tau)\approx(0.99,\tfrac{2\pi}{13}),$ $(0.86,\tfrac{14\pi}{81})$, and $(0.73,\tfrac{6\pi}{29})$. The plots are qualitatively similar for all $\kappa_2\leq0.5$. All plots are produced at the same action isovalue $\EE=20.5$.}
     \label{fig: kappa1=1 plots}
\end{figure}
\subsection{Constituent monopole locations}\label{sec-parameters}
The Nahm data \eqref{sol-deformed}-\eqref{U-W-param-sol} correspond via \eqref{caloron-Nahm-from-delay-n=2}-\eqref{caloron-match-n=2} and the Nahm transform to charge $2$ calorons. $\SU(2)$ calorons can be interpreted as consisting of constituent monopoles of opposite magnetic charge, so for charge $2$ calorons there are two constituent monopoles of equal and opposite magnetic charge $2$.

To understand the constituent monopole structure we analyse the corresponding spectral curves. These are defined as the complex level curves
\begin{align}
    \det(\eta+T_p(\zeta))=0,\quad p=1,2,
\end{align}
where
\begin{align}
    T_p(\zeta)=\wh{\beta}_p+(\wh{\alpha}_p+\wh{\alpha}^\dagger_p)\zeta-\wh{\beta}_p^\dagger\,\zeta^2.
\end{align}
For the caloron Nahm data \eqref{caloron-alpha-beta} defined by the solution \eqref{sol-deformed} of the delayed equations these curves are
\begin{align}
    \left(\eta+\tfrac{1}{2}(\lambda^2\cos2\tau+4\Re(\xi))\zeta\right)^2-\frac{D_2^2}{4}\kappa_2'^2(\e^{2\ii\psi_3}+\e^{-2\ii\psi_3}\zeta^4)-\frac{D_2^2}{2}(1+\kappa_2^2)\zeta^2&=0\label{spec-curve1},\\
    (\eta+2\Re(\xi)\zeta)^2+\frac{D_1^2}{4}\kappa_1^2(\e^{\ii(\psi_1+\psi_3)}+\e^{-\ii(\psi_1+\psi_3)}\zeta^4)
    +\frac{D_1^2}{2}(1+\kappa_1'^2)\zeta^2&=0.\label{spec-curve2}
\end{align}
The constituent monopoles can be interpreted as the monopoles arising from these spectral curves. 

The spectral curves help us to identify locations of constituent monopoles. This is understood via the twistor transform, where the spatial $\R^3$ coordinates $(x_1,x_2,x_3)$ are related to the twistor space coordinates $(\zeta,\eta)$ via the constraint
\begin{align}
    \eta+x_1-\ii x_2+2 x_3\,\zeta-(x_1+\ii x_2)\zeta^2=0.
\end{align}
When a $2$-monopole is well-separated, the spectral curve approximately factorises as a product
\begin{align*}
    (\eta+x_1-\ii x_2+2 x_3\,\zeta-(x_1+\ii x_2)\zeta^2)(\eta+x_1'-\ii x_2'+2 x_3'\,\zeta-(x_1'+\ii x_2')\zeta^2)=0
\end{align*}
and the two monopole locations of the constituent corresponding to the given curve are identified as the coordinates $(x_1,x_2,x_3)$ and $(x_1',x_2',x_3')$. There are also highly symmetric configurations where the locations coincide. The simplest of these possible for degree $2$ curves takes the form
\begin{align}\label{axial-curve}
    \eta^2+C\zeta^2=0,
\end{align}
where $C\in\mathbb{R}$; in order to describe a BPS monopole, up to a phase this constant needs to be $C=\tfrac{\pi^2}{4}$ \cite{Ward1981}, but for caloron curves this is not required. Rotations in $\R^3$ act on spectral curves via $\SU(2)$ M\"obius transformations
\begin{align}
    (\zeta,\eta)\mapsto\left(\frac{a\zeta+b}{-\ol{b}\zeta+\ol{a}},\frac{\eta}{(-\ol{b}\zeta+\ol{a})^2}\right),
\end{align}
with $a=\cos\tfrac{\varphi}{2}-\ii n_3\sin\tfrac{\varphi}{2}$ and $b=-\sin\tfrac{\varphi}{2}(n_2+\ii n_1)$ describing a rotation by angle $\varphi$ around the unit axis $(n_1,n_2,n_3)$. With this, we see that the curve \eqref{axial-curve} has axial symmetry around the $x_3$-axis. For later purposes it is worth also highlighting the curve of the form
\begin{align}\label{axial2}
    \eta^2-C(1+\zeta^2)^2=0,
\end{align}
which is axially-symmetric around the $x_2$-axis.

\subsubsection*{$\tau=0$}
For the $\tau=0$ family considered above, the spectral curves \eqref{spec-curve1}-\eqref{spec-curve2} reduce to
\begin{align}
\begin{aligned}
    \left(\eta+\tfrac{1}{2}\lambda^2\zeta\right)^2-\frac{{D_1^2}}{4}\kappa_1^2(1+
    \zeta^4)-\frac{\kappa_1^2D_1^2}{2\kappa_2'^2}(1+\kappa_2^2)\zeta^2&=0,\\
    \eta^2-\frac{D_1^2}{4}\kappa_1^2(1+\zeta^4)
    +\frac{D_1^2}{2}\left(1+\kappa_1'^2\right)\zeta^2&=0.
\end{aligned}
\end{align}
We will vary $\kappa_2$ and fix $\kappa_1$ so that $K(\kappa_1)- K(\kappa_2)\tfrac{\kappa_2'}{\kappa_1}$ is small and positive. When $\kappa_2\approx 1$ this constraint forces $\kappa_1\approx0$, so these are well-approximated by
\begin{align}
\begin{aligned}
    \left(\eta+(\tfrac{1}{2}\lambda^2-D_2)\zeta\right)\left(\eta+(\tfrac{1}{2}\lambda^2+D_2)\zeta\right)&=0,\\
    \eta^2+\frac{D_1^2}{2}\zeta^2&=0.
\end{aligned}
\end{align}
The first represents a constituent $2$-monopole with locations $(0,0,\tfrac{1}{4}\lambda^2\pm \tfrac{1}{2}D_2)$ and the second curve is of the form \eqref{axial-curve} which has axial symmetry around the $x_3$-axis. We see from \eqref{lam-constraint} that $D_2\leq\tfrac{\lambda^2}{2}$, and equality holds if and only if $\kappa_2=1$, corresponding to the constant solutions \eqref{const-sol}. The region $\kappa_2\approx 1$ is a deformation of the constant solutions \eqref{const-sol} which thus resembles an axially-symmetric $2$-monopole at the origin, underneath two separated $1$-monopoles on the positive $x_3$-axis, with the lowest converging to overlap with the torus at the origin when $\kappa_2=1$. This matches the pictures in Figures \ref{subfig:tau0-mon t=0,k2=0.99}-\ref{subfig:tau0-mon t=0.725,k2=0.99}.

When $\kappa_2\approx 0$ and $\kappa_1$ is close to $1$ and the spectral curves are approximated by
\begin{align}\label{approx-spec-curves-tau0-kappa2-0}
\begin{aligned}
    \left(\eta+\tfrac{1}{2}\lambda^2\zeta\right)^2-\frac{D_1^2}{4}(1+
    \zeta^2)^2&=0,\\
    \left(\eta-\tfrac{D_2}{2}(1-\zeta^2)\right)\left(\eta+\tfrac{D_2}{2}(1-\zeta^2)\right)&=0.
\end{aligned}
\end{align}
The first is the spectral curve \eqref{axial2} of a $2$-monopole axially-symmetric around the $x_2$-axis, with the replacement $\eta\mapsto\eta+\tfrac{\lambda^2}{2}\zeta^2$, which corresponds to a translation by $\frac{1}{4}\lambda^2$ in the $x_3$-direction. The second represents two $1$-monopoles with locations along the $x_1$-axis at $x_1\approx \pm\frac{D_2}{2}$. This matches the pictures in Figures \ref{subfig:tau0-mon t=0,k2=0.5}-\ref{subfig:tau0-mon t=0.725,k2=0.5}.
\subsubsection*{$\kappa_1=1$}
In the $\kappa_1=1$ case, note that the second spectral curve \eqref{spec-curve2} corresponding to the constant data factorises exactly as
\begin{multline}
    \left(\eta+\tfrac{D_1}{2}\e^{\tfrac{\ii}{2}(\psi_1+\psi_3+\pi)}-\tfrac{D_1}{2}\e^{-\tfrac{\ii}{2}(\psi_1+\psi_3+\pi)}\zeta^2\right)\\
    \times\left(\eta-\tfrac{D_1}{2}\e^{\tfrac{\ii}{2}(\psi_1+\psi_3+\pi)}+\tfrac{D_1}{2}\e^{-\tfrac{\ii}{2}(\psi_1+\psi_3+\pi)}\zeta^2\right)=0.\label{spec-curve2-kap1}
\end{multline}
This spectral curve represents a constituent $2$-monopole with locations $$\pm(\tfrac{D_1}{2}\sin(\tfrac{1}{2}(\psi_1+\psi_3)),\tfrac{D_1}{2}\cos(\tfrac{1}{2}(\psi_1+\psi_3)),0),$$
i.e., on antipodal points of a circle of radius $\tfrac{D_1}{2}$ in the $(x_1,x_2)$-plane with overall orientation is determined by the angle $\tfrac{1}{2}(\psi_1+\psi_3)$. In Figures \ref{subfig:kappa1_k2=0.9t=0}-\ref{subfig:kappa1_k2=0.9t=0.725}, \ref{subfig:kappa1_k2=0.8t=0}-\ref{subfig:kappa1_k2=0.8t=0.725}, and \ref{subfig:kappa1-mon-ring t=0}-\ref{subfig:kappa1-mon-ring t=0.725}, this angle is always very small, approximately given by $\tfrac{\pi}{15}$, $\tfrac{3\pi}{47}$, and $\tfrac{2\pi}{37}$ respectively, explaining clearly why the constituent pair lies near the $x_2$ axis.

The other constituent monopole has spectral curve
\begin{align}
    \left(\eta+\tfrac{1}{2}\lambda^2\cos2\tau\zeta\right)^2-\frac{D_2^2}{4}\kappa_2'^2(\e^{2\ii\psi_3}+\e^{-2\ii\psi_3}\zeta^4)-\frac{D_2^2}{2}(1+\kappa_2^2)\zeta^2&=0.
\end{align}
From this we see that when $\kappa_2\approx1$ the curve is approximated by
\begin{align}
    (\eta+(\tfrac{1}{2}\lambda^2\cos2\tau-D_2)\zeta)(\eta+(\tfrac{1}{2}\lambda^2\cos2\tau+D_2)\zeta)=0,
\end{align}
representing two separated $1$-monopoles located at $(0,0,\tfrac{1}{4}\lambda^2\cos2\tau\pm\tfrac{1}{2}D_2)$. Likewise from \eqref{D1-kappa1=1} we also have $D_1$ becomes small, so the other constituents get close together. This is a different deformation of the constant solution to the $\tau=0$ case, and visualised in Figures \ref{subfig:kappa1_k2=0.9t=0}-\ref{subfig:kappa1_k2=0.9t=0.725}. When $\kappa_2\approx0$ this curve instead takes the form
\begin{align}
    \left(\eta+\tfrac{1}{2}\lambda^2\cos2\tau\zeta\right)^2-\frac{D_2^2}{2}(\e^{\ii\psi_3}+\e^{-\ii\psi_3}\zeta^2)^2&=0,
\end{align}
which is a rotation by angle $\psi_3$ around the $x_3$-axis, followed by a translation in the $x_3$-direction, of the curve \eqref{axial2} of an axially-symmetric monopole. This is reflected in Figures \ref{subfig:kappa1-mon-ring t=0}-\ref{subfig:kappa1-mon-ring t=0.725}.
\section{Conclusion and further directions}\label{sec:concl}
We have developed a Nahm transform for rotating calorons and used it to construct charge 1 rotating calorons with nontrivial asymptotic holonomy numerically.  Our transform is complete for rational rotation angles.  But the problem of constructing rotating calorons with irrational rotation angle remains open.

Apart from the Nahm transform, another successful tool for constructing anti-self-dual gauge fields is the Kobayashi--Hitchin correspondence.  This constructs an anti-self-dual gauge field by taking a holomorphic vector bundle and solving the hermitian Yang--Mills equation.  This seems a natural approach to take for rotating calorons because, of the many complex structures on $\R^4$, only one is fixed by the glide rotation \eqref{twisted-gf}.  A Kobayashi--Hitchin correspondence for calorons has been established in \cite{CharbonneauHurtubise2008Rat.map}, where the corresponding holomorphic bundle lives on the compactification $\mathbb{CP}^1\times\mathbb{CP}^1$ of $\R^3\times S^1=\C\times\C^\ast$.  A similar construction might work for rotating calorons, because the quotient of $\R^4$ by the action \eqref{twisted-gf} is also biholomorphic to $\C\times\C^\ast$.  One could try to construct rotating calorons by identifying suitable holomorphic bundles on $\mathbb{CP}^1\times\mathbb{CP}^1$ and solving the hermitian Yang--Mills equation with suitable boundary conditions.  This might be more effective than the Nahm transform in the case of irrational rotation angles.

We would like to point out that the delayed Nahm equations are integrable, in the sense that they may be written as a Lax pair.  More precisely, equations \eqref{delay-eqs} with $U_\pm=V_\pm=0$, are equivalent to $[P(\zeta),Q(\zeta)]=0$, where $P$ and $Q$ are operators of the form
\begin{align}
    P(\zeta)=\frac{\d}{\d s}+\alpha-S_-\beta^\dagger\zeta,\quad Q(\zeta)=\beta S_++(\alpha+\alpha^\dagger)\zeta-S_-\beta^\dagger\zeta^2.
\end{align}
Here $S_\pm$ are the shift operators $(S_\pm v)(s)=v(s\pm\theta)$ defined earlier, and $\zeta\in\C$ is a spectral parameter.  Other examples of integrable delay-differential equations have been studied in \cite{delay1,delay2,delay3,delay4,delay5,delay6,delay7,delay8,delay9,delay10,delay11,delay12,delay13,delay14,delay15,delay16,delay17}.  The Nahm equations are an important integrable system that includes the Toda and Euler top equations as special cases \cite{ErcolaniSinha1989monopoles,braden2011,BradenCherkisQuinones2022,BradenDisneyHogg2023,BradenDisneyHogg2024}.  So it would be interesting to explore the integrability properties of the delayed Nahm equations.
\subsection*{Acknowledgements}
JC acknowledges support provided by the University of Leeds School of Mathematics Research Visitor Centre, where much of this work was completed.
\bibliographystyle{unsrt}
\bibliography{refs,rotatingQCD}
\end{document}